\documentclass[onecolumn]{IEEEtran}

\usepackage{macro}

\title{List Decoding Random Euclidean Codes and Infinite Constellations}
\author{
\IEEEauthorblockN{
Yihan Zhang\IEEEauthorrefmark{1},
Shashank Vatedka\IEEEauthorrefmark{3}
}\\
\IEEEauthorblockA{
\IEEEauthorrefmark{1}Faculty of Computer Science, Technion Israel Institute of Technology \\
\IEEEauthorrefmark{3}Department of Electrical Engineering, Indian Institute of Technology Hyderabad
}
}

\begin{document}
\maketitle
\footnotetext[1]{
This work was done when Shashank Vatedka was at the Chinese University of Hong Kong, where he was supported in part by CUHK Direct Grants 4055039 and 4055077.
He would like to acknowledge funding from a seed grant offered by IIT Hyderabad and the Start-up Research Grant (SRG/2020/000910) from the Science and Engineering Board, India.
Yihan Zhang has received funding from the European Union’s Horizon 2020 research and innovation programme under grant agreement No 682203-ERC-[Inf-Speed-Tradeoff].

This paper was presented in part at 2019 IEEE International Symposium on Information Theory (ISIT), Paris, France \cite{zhang-vatedka-2019-listdec-real-isit}. 
} 

\begin{abstract}
We study the list decodability of different ensembles of codes over the real alphabet under the assumption of an omniscient adversary. It is a well-known result that when the source and the adversary have power constraints $ P $ and $ N $ respectively, the list decoding capacity is equal to $ \frac{1}{2}\log\frac{P}{N} $. 
Random spherical codes achieve constant list sizes, and the goal of the present paper is to obtain a better understanding of the smallest achievable list size as a function of the gap to capacity. 
We show a reduction from arbitrary codes to spherical codes, and derive a lower bound on the list size of typical random spherical codes. We also give an upper bound on the list size achievable using nested Construction-A lattices and infinite Construction-A lattices. We then define and study a class of infinite constellations that generalize Construction-A lattices and prove upper and lower bounds for the same. 
Other goodness properties such as packing goodness and AWGN goodness of infinite constellations are proved along the way. Finally, we consider random lattices sampled from the Haar distribution and show that if a certain number-theoretic conjecture is true, then the list size grows as a polynomial function of the gap-to-capacity.
\end{abstract}


\tableofcontents

\section{Introduction}\label{sec:ld_reals}





In this paper, we study communication in presence of a power-constrained adversary. 
This is a point-to-point communication problem where a sender wants to communicate a message $ m\in\{0,1\}^{nR} $ of $ nR $ bits  to a receiver through a real-valued channel corrupted by a malicious \emph{omniscient} adversary. The transmitter uses $ n $ channel uses to send a signal\footnote{We use underline to denote vectors of length $n$.} $ \vx\in\bR^n $ across the channel.
The adversary can corrupt the transmitted signal by adding a noise vector $ \vs\in\bR^n $, which is allowed to be any noncausal function of the transmitted signal and the transmission protocol. The sender and the adversary have power constraints of $ P $ and $ N $ respectively, i.e., we impose the restriction that $ \Vert\vx\Vert\leq \sqrt{nP} $ and $ \Vert\vs\Vert\leq \sqrt{nN} $\footnote{Without further specification, $\|\cdot\|$ denotes the $ L^2 $-norm.}. The goal is to design a transmission scheme that provides a high data rate $ R $ while ensuring a zero probability of error of decoding at the receiver. This problem turns out to be equivalent to the sphere packing problem, which asks for the maximum $ R $ such that there is a set of $ 2^{nR} $ points within a ball of radius $ \sqrt{nP} $ such that every pair of points is spaced $ 2\sqrt{nN} $ apart. Finding the capacity of this channel remains an open problem, but nonmatching upper and lower bounds are known~\cite{kabatiansky-1978, blachman-1962}.

We study a slight variant of this problem, 
where instead of uniquely decoding the transmitted message $ m $, the receiver attempts to recover a  list of $ L $ codewords with the guarantee that the transmitted codeword lies in this list. This is called the \emph{list decoding} problem, also known as \emph{multiple packing}, which is well studied at least in the context of binary adversarial channels~\cite{guruswami-lncs}. In this paper, we attempt to systematically study upper and lower bounds on achievable list sizes for various ensembles of random codes for the real channel.

List decoding for adversarial channels is an interesting problem in its own right, but can also be a very useful tool in several other problems. For instance, Langberg~\cite{langberg-focs2004} showed that if there exists a coding scheme that achieves a list size that is at most polynomial in the blocklength $ n $, then even a small amount of shared secret key (just about $ \Theta(\log n) $ bits kept secret from the adversary) between the sender-receiver pair suffices to ensure that the true message can be uniquely decoded by the receiver. List decoding can also serve as a useful proof technique for studying other adversarial channels~\cite{dey-sufficiently-tit,zhang-quadratic-arxiv,zhang-2020-twoway}.

For the quadratically constrained adversarial channel, it is known that if the transmission rate $ R $ is greater than $ \frac{1}{2}\log\frac{P}{N} $, then no coding scheme can achieve subexponential (in $ n $) list sizes. On the other hand, it is also known that random spherical codes of rate $ R<\frac{1}{2}\log\frac{P}{N} $ can achieve constant (in $ n $) list sizes. We can therefore call $ \frac{1}{2}\log\frac{P}{N} $ to be the \emph{list decoding capacity} of this channel. Once this is established, it is of interest to find the least possible list sizes that are achievable as a function of $ \delta \coloneqq \frac{1}{2}\log\frac{P}{N}-R $. We will show that in order to find the order-optimal list sizes as a function of the rate gap to capacity, it suffices to only study spherical codes, where all codewords $ \vx $ have norm $ \Vert\vx \Vert = \sqrt{nP} $. It is known that random spherical codes have list sizes upper bounded by $ \cO(\frac{1}{\delta}\log\frac{1}{\delta}) $. We show that ``typical'' random spherical codes have list sizes which grow as $ \Omega(1/\delta) $. 
In an attempt to devise more ``practical'' coding schemes that achieve the list decoding capacity, we look for structured codes that can guarantee small list sizes. Specifically, we investigate a class of nested lattice codes and find lower bounds on the list size.  We show that random nested Construction-A lattices achieve list sizes $ 2^{\cO(\frac{1}{\delta}\log^2\frac{1}{\delta})} $. To the best of our knowledge, this is the first such result which shows that lattice codes can achieve constant list sizes. However, the list sizes are exponentially worse than the list sizes for random spherical codes. We conjecture that there exist lattice codes that achieve list sizes of $ \cO(\frac{1}{\delta}\log\frac{1}{\delta}) $ and provide some heuristic calculations which suggest this.

We then relax the power constraint of the transmitter and study the list decodability of \emph{infinite constellations} (ICs). Infinite constellations  generalize lattices, and to the best of our knowledge, were first studied systematically in the context of channel coding by Poltyrev~\cite{poltyrev1994coding}. Poltyrev showed that there exist ICs that are good codes for the additive white Gaussian noise (AWGN) channel. In this paper, we introduce an ensemble of periodic infinite constellations and study upper and lower bounds on the list size of typical ICs. A list decodable code for the power-constrained (for both  the transmitter and the adversary) adversarial channel can be obtained by taking the intersection of the IC with a ball of radius $ \sqrt{nP} $. We show that the code obtained by taking this intersection achieves list size $ \cO(\frac{1}{\delta}\log\frac{1}{\delta}) $.

\section{Overview of our results}\label{sec:results}
Let us now formally describe the problem.
The sender encodes a message $ m\in \{0,1\}^{nR} $ into a codeword $ \vx $ in $ \bR^n $ which is intended for the receiver. The sender has a transmit power constraint, which is modeled by demanding that the $ L^2 $ norm $ \Vert\vx\Vert $ must be no larger than $ \sqrt{nP} $ for some $ P>0 $. The transmission is observed noncausally by an adversary who corrupts the transmitted vector by adding a noise vector $ \vs $ to $ \vx $. The adversary has a power constraint of $ N $, which means that $ \Vert\vs\Vert\leq \sqrt{nN} $ for some $ N>0 $. However, $ \vs $ is allowed to otherwise be any function of $ \vx $ and the codebook.  The receiver obtains $ \vy=\vx+\vs $. 
The list decoder takes $ \vy $ as input and outputs a  list of $ L $ messages\footnote{We sometimes abuse terminology and interchangeably talk about lists of messages and lists of codewords. This does not introduce confusion since in this paper we are only concerned with codes whose encoder is a one-to-one map, i.e., a deterministic encoder.}, and an error is said to have occurred if the true message $ m $ is not in this list. 

\begin{definition}[List decodability over $\bR$]
	Let $ P,N\in \bR_{>0} $ and $ L\in\bZ_{>0} $.
	We say that a code $ \cC\subset \bR^n $ is \emph{$ (P,N,L) $-list decodable} if
	\begin{itemize}
		\item The code satisfies a maximum power constraint of $ P $, i.e., we have $ \Vert \vx \Vert_2^2\leq nP $ for all $ \vx\in\cC $.
		\item An omniscient adversary with power $ N $ cannot enforce a list size greater than $ L $, i.e., for all $ \vx\in\cC $ and all $ \vs\in \cB(0,\sqrt{nN}) $\footnote{Here $ \cB(\vc, r) $ denotes an $ L^2 $ ball in $ \bR^n $ of radius $r$ centered at $\vc$. We at times also write $ \cB^n(\vc,r) $ to emphasize the ambient dimension.}, we have $ |\cC\cap \cB(\vx+\vs,\sqrt{nN})|\leq L $.
	\end{itemize} 
	The \emph{rate} of $ \cC $ is defined as $ R(\cC)\coloneqq \frac{1}{n}\log |\cC| $\footnote{Here $ \log $ is short for $ \log_2 $.}.
	
	A rate $ R\in\bR $ is said to be \emph{achievable} for $ (P,N,L) $-list decoding if for infinitely many $ n $, there exist codes $ \cC\subset \bR^n $ having rate $ R(\cC)\geq R $ that are $ (P,N,L) $-list decodable.
	\label{defn:PNLlistdecoding}
\end{definition}
In the definition above, we do not prohibit $ L $ from being a function of $ n $. In many applications, it suffices to have list sizes that grow as $ \cO(n^{\gamma}) $ for a suitably small $ \gamma $. However, in this paper, we aim for constant list sizes.

\begin{definition}[List decoding capacity]
	Fix any $ P,N>0 $.
	We say that $ C(P,N) $ is the \emph{list decoding capacity} if for every $ \delta>0 $, there exists a $ \gamma>0 $ such that $ C(P,N)-\delta $ is achievable for $ (P,N,\cO(n^{\gamma})) $-list decoding, and for every $ \delta>0 $, there exist no codes of rate  $ C(P,N)+\delta $ which are $ (P,N,2^{o(n)}) $-list decodable.
\end{definition}

The following result is folklore, and a proof can be found in~\cite{zhang-quadratic-arxiv}.
\begin{theorem}[Folklore,~\cite{zhang-quadratic-arxiv}]
	For any $ P,N>0 $, 
	\[
	C(P,N) = \left[ \frac{1}{2}\log\frac{P}{N} \right]^+.\footnote{Here $ [a]^+ $ is defined as $ \max\{a,0\} $.}
	\]
\end{theorem}



Again, we are in search of structured ensembles of codes achieving ideally the same list decoding performance as random codes. This problem is not as extensively studied as in the finite field case.

The class of problems that we are interested in is the following: 
\begin{itemize}
	\item Suppose that we desire a target rate $ R=C(P,N)-\delta $, for some small $ \delta>0 $. Then what is the smallest list size $ L $ that we can achieve? Specifically, we are interested in the dependence of $ L $ on $ \delta $.
	\item What are the fundamental lower bounds on the list size for a fixed $ \delta $? 
	\item If we restrict ourselves to structured codes, e.g., nested lattice codes~\cite{erez2005lattices}, then what list sizes are achievable?
\end{itemize}
It was shown in~\cite{zhang-quadratic-arxiv} that $ \cO\left(\frac{1}{\delta}\log\frac{1}{\delta}\right) $ list sizes are achievable using random spherical codes.
If we define $ \cS^{n-1}(0,\sqrt{nP})\coloneqq \{ \vx\in\bR^n:\Vert\vx\Vert=\sqrt{nP} \} $ to be the $ (n-1) $-dimensional sphere of radius $ \sqrt{nP} $, then 
\begin{lemma}[\cite{zhang-quadratic-arxiv}]
	Let $ P>N>0 $. Fix any $ \delta>0 $, and let $ R\coloneqq C(P,N)-\delta $. Let $ \cC $ be a random codebook of rate $R$ obtained by choosing the codewords independently and uniformly over $ \cS^{n-1}(0,\sqrt{nP}) $, then
	\[
	\Pr\left[\cC \text{ is not }\left(P,N,\cO\left(\frac{1}{\delta}\log\frac{1}{\delta}\right)\right)\text{-list decodable}\right] \leq 2^{-\Omega(n)}.
	\]
	\label{lemma:achievablelistsize_spherical}
\end{lemma}
Our contributions for $ (P,N,L) $-list decoding are summarized as follows.
\begin{itemize}
	\item We derive lower bounds on the list size of random spherical codes. We show that if $ R=C(P,N)-\delta $, then $ L $ grows as $ \Omega(1/\delta) $ with high probability (whp).
	\item We then investigate the achievable list sizes for random nested  lattice codes, and show that if $ R=C(P,N)-\delta $, then $ L=2^{\cO\left(\frac{1}{\delta}\log^2\frac{1}{\delta}\right)} $ is achievable using Construction-A lattices.
	\item Conditioned on a conjecture for random lattices, we provide heuristic calculations which suggest that lattice codes might achieve list sizes that grow as $ \cO(\mathrm{poly}(1/\delta)) $.
\end{itemize}

We then perform a systematic study of the problem of list decoding infinite constellations in $ \bR^n $. An \emph{infinite constellation} is defined as a countable subset of $ \bR^n $. 
\begin{definition}\label{def:ic}
	An infinite constellation $ \cC \subset \bR^n$ is said to be \emph{$ (N,L) $-list decodable} if for every $ \vy\in\bR^n $, we have 
	\[
	|\cC\cap \cB(\vy,\sqrt{nN})|\leq L.
	\]
	The \emph{density} of the constellation is defined as
	\[ \Delta(\cC)\coloneqq \limsup_{a\to\infty}\frac{|\cC\cap [-a/2,a/2]^n|}{a^n}. \]
	The \emph{normalized logarithmic density}, defined as $ R(\cC)\coloneqq\frac{\log\Delta(\cC)}{n} $, is a measure of the ``rate'' of an infinite constellation. The \emph{effective volume} of $ \cC $ is defined as $ V(\cC)=1/\Delta(\cC) $, and the \emph{effective radius} 
	$\reff(\cC)$ is defined as the radius of a ball having volume equal to $ V(\cC) $.
\end{definition}

\begin{remark}
In general, one cannot replace $\limsup$ with $\lim$ in the definition of $ \Delta(\cC) $. 
However, in this paper, we are only concerned with periodic ICs, i.e., ICs that are unions of translations of a finite set. 
For such ICs, the limit does exist. 
See~\cite{groemer-1963-packing-density} and~\cite{cohn-2003-sphere-packing} for more discussions on the definition of density. 
\end{remark}

\begin{remark}
For a periodic IC $ \cC $, the NLD of $\cC$ remains the same if we replace the cube $ [-a/2,a/2]^n $ with any centrally symmetric connected compact set $ a\cB\subset\bR^n $ with nonempty interior and compute the NLD in the following way
\begin{align}
\Delta(\cC) &= \limsup_{a\to\infty} \frac{\card{\cC\cap a\cB}}{a^n\vol(\cB)}. \notag 
\end{align}
\end{remark}

Clearly, every lattice is an infinite constellation. We show that if $ \Lf $ is a random Construction-A lattice with $~\reff(\Lf)\geq \sqrt{nN}2^{\delta} $, then $ \Lf $ is $ (N,2^{\cO(\frac{1}{\delta}\log^2\frac{1}{\delta})}) $-list decodable. We also introduce a class of random periodic infinite constellations $ \cC $ with $~\reff(\cC)= \sqrt{nN}2^{\delta} $ which have list sizes that grow as $ \cO(\frac{1}{\delta}\log\frac{1}{\delta}) $.
Additionally, we show a matching lower bound on the list size for these random infinite constellations.

\begin{remark}
A code satisfying the requirements in Definition~\ref{defn:PNLlistdecoding} can be list decoded when used on a channel governed by an \emph{omniscient} adversary with a maximum probability of error constraint. This is because the decoder will output a small list for every transmitted codeword $ \vx\in\cC $ and every attack vector $ \vs\in\cB(0,\sqrt{nN}) $. For $ L=1 $, our problem  reduces to the problem of packing nonintersecting balls of radius $ \sqrt{nN} $ such that their centers lie within $ \cB(0,\sqrt{nP}) $. 

We could relax the problem by assuming that the decoder uses an average probability of error criterion (where the average is over messages picked uniformly at random) and the adversary knows only the codebook but not the transmitted codeword. 
This models an \emph{oblivious} adversary and the problem was studied by Hosseinigoki and Kosut~\cite{hosseinigoki-kosut-2018-oblivious-gaussian-avc-ld}.
They showed that the list decoding capacity for this problem is $ \frac{1}{2}\log(1+P/N)\indicator{L>N/P} $.

An intermediate model that lies between the omniscient and the oblivious models is the \emph{myopic} model. 
In this model, we assume that a myopic adversary sees a noncausal noisy version of the transmitted codeword. This problem was studied in~\cite{zhang-quadratic-arxiv}.
\end{remark}

\section{Organization of the paper}\label{sec:org}
In Sec.~\ref{sec:finitefield_prior_work}, we survey the literature on list decoding over finite fields. 
Notation and prerequisite facts and lemmas are listed in Sec.~\ref{sec:notation} and Sec.~\ref{sec:prelim}, respectively. A table of frequently used notation can be found in Appendix~\ref{sec:tab_notation}. A lower bound on list sizes of random spherical codes is provided in  Sec.~\ref{sec:lb_ls_spherical} while some of the calculations are deferred to Appendix~\ref{sec:calc_prop_lbound_randomspherical}. In Sec.~\ref{sec:ld_constr_a}, we turn to study list decodability of random nested Construction-A lattice codes. For the benefit of the readers who are not familiar with lattices,  a quick primer is provided in Appendix~\ref{sec:primer_lattices}. We define infinite constellations in Sec.~\ref{sec:ic}, and give matching upper and lower bounds on list sizes of an ensemble of regular infinite constellations in Sec.~\ref{sec:regular_ic}. 
Results on other goodness properties of ICs are presented in Appendix~\ref{sec:ic_goodness}.
Finally we give some heuristic results on the list sizes achieved by lattice codes. We recall the Haar distribution on the space of lattices in Sec.~\ref{sec:haar_meas_on_lattices}, then introduce two important integration formulas by Siegel and Rogers in Sec.~\ref{sec:siegel_rogers_avg_formulas} and their improvements in Sec.~\ref{sec:improvement_on_rogers}. We prove a list size upper bound conditioned on a conjecture in Sec.~\ref{sec:cond_ld_haar}. We conclude the paper in Sec.~\ref{sec:rk_open_prob} with several open problems. 

\section{Prior work}\label{sec:finitefield_prior_work}








\subsection{Prior work on list decoding over finite fields}\label{sec:ld_finite}
Given a prime power $ q $ and $ R\in (0,1) $, how to construct a subset $\cC$ of $\bF_q^n$ of size $q^{nR}$ such that the points in $ \cC $ are as far apart (in Hamming distance) as possible?  This question is motivated by communication through noisy channels and is studied under different notions of ``far apart''. Consider a transmitter who wishes to convey an arbitrary $q$-ary message of length $nR$ to a receiver through a noisy channel. To protect the information against noise, the transmitter can add some redundancy and send a coded version of the message through the channel. Classical coding theory is devoted to the study of the situation where the codeword is a length-$n$ vector over $\bF_q$ and the adversary who has access to the transmitted codeword is allowed to change any $np$  symbols where $0<p<1$ is a constant. The receiver is then required to figure out  the original message given a maliciously corrupted word. For fixed $q$ and $p$, the goal is to design an as-large-as-possible set of codewords so as to get as-much-as-possible information through in $n$ uses of the channel,  while ensuring that the receiver can  decode the message correctly under any legitimate attack by the adversary. We are interested in the asymptotic behaviour of the throughput as the blocklength $n$ grows.

It is not hard to see that the above question is equivalent to the question of determining the optimal density of packing Hamming balls of radii $np$ in $\bF_q^n$. The best possible density $R$ is widely open. People thus consider relaxed versions of this problem. Instead of asking the receiver to output a unique correct message, we allow him to output a list of $L$ messages which is guaranteed to contain the correct one. For fixed $q$, there is now a tradeoff among three quantities: $R$, $p$ and $L$. The question is nontrivial only when $L$ is required to be small, otherwise outputting all $q^{nR}$ messages is always a valid scheme. This problem is called \emph{list decoding}.
It is sometimes also referred to as \emph{multiple packing} since it can be alternatively thought of as packing balls such that they can overlap but there are no more than $L$ balls on top of any point in the space. Let $ d_{\mathrm{H}} $ denote the Hamming distance and let $ \bham^n(\vy,np) \coloneqq \{ \vx\in\bF_q^n: d_{\mathrm H}(\vy,\vx)\leq np \} $ denote the Hamming ball of radius $ np $ centered at $ \vy $.

\begin{definition}[List decodability over $\bF_q$]\label{def:list_dec_fq}
A code $\cC\subseteq\bF_q^n$ is said to be \emph{$(p,L)$-list decodable} if for any $\vy\in\bF_q^n$, $\card{\cC\cap\bham^n(\vy,np)}\le L$. 
\end{definition}

It turns out that such a relaxation indeed makes the problem easier. In this case, we entirely understand the information-theoretic limit of list decoding. Specifically, Zyablov and Pinsker~\cite{zyablov-pinsker-1981-list-dec-cap} showed:
\begin{theorem}[List decoding capacity  over $\bF_q$,~\cite{zyablov-pinsker-1981-list-dec-cap}]\label{thm:list_dec_cap_fq}
For any constant $\delta>0$, any $0<p<1-1/q$ and any $n$ large enough,  there exists a $(p,1/\delta)$-list decodable code $\cC \subseteq\bFq^n $ of rate $1-H_q(p)-\delta$; on the other hand,  any code $\cC \subseteq\bFq^n $ of rate $1-H_q(p)+\delta$ is $(p,q^{\Omega(n\delta)})$-list decodable.
\end{theorem}
The sharp threshold $1-H_q(p)$ around which the list size exhibits a phase transition is known as the list decoding capacity, denoted by $C$.  The stunning point of the above theorem is that the list size can be made a constant, independent of $n$, while a rate close to the list decoding capacity is still achieved.

Throughout the paper, we use $\delta$ to denote the gap between the rate that the code is  operating at and the list decoding capacity, i.e., $R=C-\delta$.

Let $\tau$ denote the gap between the adversary's power and the list decoding radius $1-1/q$, i.e., $p=1-1/q-\tau$.

First note that expressing the rate as a function of the list size is equivalent to expressing the list size as a function of the gap to capacity. Lower (resp.\ upper) bounds on the rate naturally translate to upper (resp.\ lower) bounds on the list size and vice versa. Indeed, for any  increasing function $f$, claiming that a rate $R\ge C-\frac{1}{f(L)}$ can be achieved by a $(p,L)$-list decodable code is equivalent to claiming the existence of a rate $R=C-\delta$ code whose list size is at most $L\le f^{-1}(1/\delta)$ under an adversary with power budget $p$. We will state prior results and our results only in the second form.

The aforementioned list decoding capacity theorem (Theorem~\ref{thm:list_dec_cap_fq}) is obtained via standard random coding argument. Indeed, it is well-known and easy to show that the list size of a random code (of which each codeword is sampled uniformly and independently) of rate $1-H_q(p)-\delta$ against a power-$p$ adversary is at most $1/\delta$ with high probability (whp). It also turns out~\cite{guruswami2013combinatorial, li-wootters-2018-ld-rand-lin} that $1/\delta$ is the correct scaling for the list size of a random code. Namely, there is an essentially matching\footnote{Actually, Li and Wootters~\cite{li-wootters-2018-ld-rand-lin} showed that for any constant $\gamma>0$, the list size of a random code is bounded from \emph{below} by $\frac{1-\gamma}{\delta}$ whp.} lower bound $1/\delta$ via second moment method.   

Note that $1-H_q(p)$ is also equal to the Shannon's channel capacity of a Binary Symmetric Channel with crossover probability $p$ (BSC($p$)).  
However, as we will elaborate in subsequent sections, this is not the case over the reals. 

The above list decoding capacity theorem pinpointed the information-theoretic limit of list decoding which is attained by random codes. In computer science, people are generally interested in finding structured or even explicit ensembles of objects with the same asymptotic behavior as the uniformly random ensemble. In the setting of list decoding, given the  threshold up to which constant-in-$n$ list size is possible, the ultimate goal is to construct explicit\footnote{Rigorously, there are two commonly used definitions of explicitness in the literature. To give an explicit linear code, it suffices to 
\begin{enumerate}
    \item either construct its generator matrix in $\poly(n)$ time deterministically; 
    \item or compute each entry of its generator matrix in $\poly\log(n)$ time deterministically.
\end{enumerate}
} codes with the same list decoding performance as random codes. As an intermediate step which is also interesting in its own right, people aim to reduce the amount of randomness used in the construction and shoot  for more ``structured" ensembles of codes. A natural candidate is linear codes. However, sadly, even if we  restrict our attention to linear codes, its  list decodability is still not completely understood. Specifically, a random linear code over $\bF_q$ of rate $R$ refers to a random subspace of $\bF_q^n$ uniformly sampled from all subspaces of a fixed dimension $nR$.

\begin{conjecture}\label{conj:list-dec-rand-lin}
For any $\delta>0$, prime power $q$ and $0<p<1-1/q$, a random linear code of rate $1-H_q(p)-\delta$ over $\bF_q$ is $(p,1/\delta)$-list decodable whp.
\end{conjecture}
The conjecture is known to be true over $\bF_2$~\cite{li-wootters-2018-ld-rand-lin}. However, it is open in other cases 
if we insist the universal constant in the list size to be one. In particular, the conjecture becomes more challenging when we  work in the high-noise low-rate regime and in large fields. For instance, consider the following scenario, the adversary's power $p$ is so large that close to $1-1/q$, say $p=1-1/q-\tau$ for a very small $\tau>0$ which can even scale with $\delta$. Then by the continuity of the entropy function, the capacity is very low and in particular vanishes as $\tau\to0$. In this case, many existing list decodability results for random linear codes degenerate in the sense that the list size blows up when $\tau\to0$. Another extreme case which is tricky to handle is when the field size $q$ is very large and is potentially an increasing function of $\delta\to0$ and/or $ n\to\infty $. In this case, many techniques in the literature also fail.

From now on, when we talk about  large $q$ (i.e., the large field size regime), we refer to $q$ which can scale with $1/\delta$ or $n$; when we talk about large $p$  or small rate  (i.e., the high-noise low-rate regime), we refer to $\tau$ which can be a function of $\delta$.

Now we survey a (potentially non-exhaustive) list of work regarding the combinatorial list decoding performance of random linear codes. 
\begin{enumerate}
    \item A classical work by Zyablov and Pinsker~\cite{zyablov-pinsker-1981-list-dec-cap} showed that a random linear code of rate $1-H_q(p)-\delta$ is $(p,q^{\cO(1/\delta)})$-list decodable whp. (See also, e.g.,~\cite{guruswami-hastad-kopparty-2010-ld-rand-lin} for a proof sketch.)
    \item Guruswami, H\aa stad, Sudan and Zuckerman~\cite{ghsz-2002} showed the \emph{existence} of a \emph{binary} linear code of rate $1-H(p)-\delta$ which is $(p,1/\delta)$-list decodable. To this end, they defined a potential function as a witness of non-list decodability and analyzed its evolving dynamics during the process of sampling a basis of the random linear code. 
    \item Guruswami, H\aa stad and Kopparty~\cite{guruswami-hastad-kopparty-2010-ld-rand-lin} showed that a random linear code of rate $1-H_q(p)-\delta$ is $(p,C_{p,q}/\delta)$-list decodable whp. However $C_{p,q}$ blows up when $p$ gets close to $1-1/q$ or $q$ is large. They used Ramsey-theoretic  tools to control low-rank lists. As for high-rank lists, naive bounds suffice.

    \item Cheraghchi, Guruswami and Velingker~\cite{cheraghchi-guruswami-velingker-2013-rip-ld-rand-lin} showed that a random linear code of rate $\Omega\paren{\frac{\tau^2}{\log^3(q/\tau)\log q}}$ is $((1-1/q)(1-\tau),\cO(1/\tau^2))$-list decodable with \emph{constant} probability. These parameters are optimal in the low-rate regime \emph{up to polylog factors in $1/\tau$ and $q$}. In their paper, ideas such as  average-radius relaxation, connections to restricted isometry property (RIP) and chaining method were brought into view. These techniques were later extensively explored and significantly developed. 
    \item Wootters~\cite{wootters-2013-ld-rand-lin-low-rate} showed that a random  linear code of rate $\Omega(\tau^2/\log q)$ is $((1-1/q)(1-\tau),\cO(1/\tau^2))$-list decodable whp. This is an improvement on~\cite{cheraghchi-guruswami-velingker-2013-rip-ld-rand-lin} via similar techniques and also fills in the gap in~\cite{guruswami-hastad-kopparty-2010-ld-rand-lin} for large $p$.
    
    \item Rudra and Wootters~\cite{rudra-wootters-2014-puncturing, rudra-wootters-2015-rand-op, rudra-wootters-2018-avg-rad-lr-rand-lin} employed the heavy machinery of generic chaining to provide improved bounds when the field size is very large, say, scaling with $1/\delta$ and even $n$. The parameter regimes become complicated in this situation and we do not copy their results here. 
    
    \item Li and Wootters~\cite{li-wootters-2018-ld-rand-lin} showed that a random \emph{binary} linear code of rate $1-H(p)-\delta$ is $(p,H(p)/\delta+2)$-list decodable whp for \emph{any} $p\in(0,1/2)$ and $\delta>0$. They did so by lifting the existential result in~\cite{ghsz-2002} to a high-probability one.

    \item Most recently, there is an exciting line of work~\cite{mosheiffetal-2019-ldpc,mosheiffetal-2020-bounds-list-dec-rand-lin,gmrsw-2020-sharp-threshold} which makes significant progress on understanding the list sizes of random linear codes. 
    In particular, the authors of these papers characterized the threshold rate (which is difficult to evaluate) of list decodable random linear codes. 
    For \emph{binary} random linear codes, they showed that the list size is (essentially) exactly $ H(p)/\delta $. 
\end{enumerate}

One can find a summary of aforementioned results in Table~\ref{tab:ld_rand_lin_state_of_the_art}.

\begin{table}
    \centering
    \caption{A non-exhaustive list of results on list decodability of random linear codes.}
    \label{tab:ld_rand_lin_state_of_the_art}
    \begin{tabular}{|p{0.07\textwidth}|p{0.15\textwidth}|p{0.17\textwidth}|p{0.12\textwidth}|p{0.16\textwidth}|p{0.07\textwidth}|}
    \hline
        \textbf{Field size} & \textbf{Noise level} & \textbf{Rate} & \textbf{List size} & \textbf{whp / with constant probability / existential} & \textbf{Reference}   \\\hline
        $q\ge2$ & $p\in(0,1-1/q)$ & $R=1-H_q(p)-\delta$ & $L=q^{\cO(1/\delta)}$ & whp &\cite{zyablov-pinsker-1981-list-dec-cap}   \\\hline
        $q=2$ & $p\in(0,1/2)$ & $R=1-H(p)-\delta$ & $L\le1/\delta$ & existential &\cite{ghsz-2002}    \\\hline
        $q\ge2$ & $p\in(0,1-1/q)$ & $R=1-H_q(p)-\delta$ & $L\le C_{p,q}/\delta$ & whp &\cite{guruswami-hastad-kopparty-2010-ld-rand-lin}    \\\hline
        $q\ge2$ & $p=(1-1/q)(1-\tau)$ & $R=\Omega\paren{\frac{\tau^2}{\log^3(q/\tau)\log q}}$ & $L=\cO(1/\tau^2)$ & with constant probability &\cite{cheraghchi-guruswami-velingker-2013-rip-ld-rand-lin}    \\\hline
        $q\ge2$ & $p=(1-1/q)(1-\tau)$ & $R=\Omega(\tau^2/\log q)$ & $L=\cO(1/\tau^2)$ & whp &\cite{wootters-2013-ld-rand-lin-low-rate}    \\\hline
        $q=2$ & $p\in(0,1/2)$ & $R=1-H(p)-\delta$ & $L\le H(p)/\delta+2$ & whp &\cite{li-wootters-2018-ld-rand-lin}    \\\hline
        $ q=2 $ & $ p\in(0,1/2) $ & $ R = 1-H(p) - \delta $ & $ L\le\floor{H(p)/\delta} + 1 $ & whp & \cite{mosheiffetal-2020-bounds-list-dec-rand-lin} \\\hline
    \end{tabular}
\end{table}

Despite a long line of research regarding list decoding, we are far from a complete understanding. Besides attempts towards Conjecture~\ref{conj:list-dec-rand-lin}, on the negative side, it turns out~\cite{guruswami2013combinatorial} that there is a matching $\Omega(1/\delta)$ lower bound on the list size of random linear codes. Namely, if we sample a linear code uniformly at random, its list size is  $\Omega(1/\delta)$ \emph{whp}. Nonetheless, in general, the best lower bound on list size for an \emph{arbitrary} code is still $\Omega(\log(1/\delta))$~\cite{blinovsky-1986-ls-lb-binary, blinovsky-2005-ls-lb-qary, blinovsky-2008-ls-lb-qary-supplementary,guruswami-vadhan-2005-ls-lb,  guruswami2013combinatorial}.
There is an exponential gap between the best upper and lower bounds even over $\bF_2$. Closing this gap is also a long standing open problem.
For arbitrarily list decodable codes with list size $L$, Blinovsky's bound was improved by Ashikhmin, Barg and Litsyn~\cite{abl-2000-list-size-2} for the $L=2$ case and by Polyanskiy~\cite{polyanskiy-2016-list-dec} for the case where $L\ge3$ is odd. 
For general omniscient adversarial channels beyond bitflip and erasure channels, the critical $L^* $ at which the list-$L$ capacity vanishes has recently been determined by Zhang, Budkuley and Jaggi~\cite{zhang-2019-list-dec-general}.

\subsection{Prior work on erasure list decoding over finite fields}\label{sec:erasure_ld_finite}
Similar questions were also posed and studied under the erasure model. In this case, the adversary is allowed to replace any $np$ coordinates of the codeword with question marks. 
\begin{definition}[List decodability under erasures over $\bF_q$]\label{def:list_dec_erasure}
    A code $\cC\subseteq\bF_q^n$ is said to be $(p,L)$-erasure list decodable if for any $\cT\subseteq[n]$ of cardinality $(1-p)n$ and any $\vy\in\bF_q^{(1-p)n}$, $\card{\curbrkt{\vx\in\cC\colon \vx|_\cT=\vy}}\le L$. 
\end{definition}
It is known that the erasure list decoding capacity is $1-p$, coinciding with the capacity of a  Binary Erasure Channel with erasure probability $p$ (BEC($p$)).
\begin{theorem}[List decoding capacity under erasures over $\bF_q$,~\cite{guruswami-lncs}, Theorem 10.3, 10.8]\label{thm:list_dec_cap_erasure}
For any small constant $\delta>0$, any $0<p<1$ and any $n$ large enough,  there exists a $(p,\cO(1/\delta))$-erasure list decodable code $\cC \subseteq\bFq^n $ of rate $1-p-\delta$; on the other hand,  any code $\cC \subseteq\bFq^n $ of rate $1-p+\delta$ is $(p,q^{\Omega(n\delta)})$-erasure list decodable.
\end{theorem}
The achievability part again follows from a random coding argument and the scaling $\Theta(1/\delta)$ of the list size of a random code is tight whp via a second moment computation~\cite{guruswami2013combinatorial}. For general codes, it can be shown that the list size is at least $\Omega(\log (1/\delta))$ using an erasure version of the punctured Plotkin-type bound (see, e.g.,~\cite{guruswami-lncs}, Theorem 10.8).

On the other hand, if we restrict our attention to \emph{linear} codes, the situation seems a little worse. The list size of a random linear code turns out to be $2^{\Theta(1/\delta)}$ whp (see, e.g.,~\cite{guruswami-lncs}, Theorem 10.6 for an upper bound and~\cite{guruswami2013combinatorial} for a matching lower bound). Intuitively, for a linear code $\cC$, the list $\curbrkt{\vx\in\cC\colon \vx|_\cT=\vy|_\cT}$ corresponding to a received word $\vy\in(\bF_q\cup\curbrkt{?})^{n}$ with $(1-p)n$ unerased locations in $\cT\subseteq[n]$ forms an affine subspace. The list size is therefore exponential in the rank of the list. For general linear codes, it can be shown that the list size is at least $\Omega(1/\delta)$ using a connection to generalized Hamming weights~\cite{cohen-etal-it1994}. Although we do not have a provable separation working uniformly in all parameter regimes, it is believed that the list size of linear codes is larger than that of nonlinear codes under erasure list decoding.

Narrowing down the exponential gap for either general codes or linear codes seems to be a particularly tricky task. 

Upper and lower bounds on list sizes of ensembles of random codes and arbitrary codes are listed in Table~\ref{tab:ld_error_erasure} for comparison.

\begin{table}
    \centering
    \caption{Upper and lower bounds on list sizes of random codes and arbitrary codes.}
    \label{tab:ld_error_erasure}
    \begin{tabular}{|p{0.12\textwidth}|p{0.18\textwidth}|p{0.21\textwidth}|p{0.16\textwidth}|}
    \hline
        \textbf{Channel model} & \textbf{Code}  & \textbf{List size} & \textbf{Reference} \\\hline
        \multirow{3}{0.12\textwidth}{Error} & Random  codes & $L\le1/\delta$ whp & Folklore  \\\cline{2-4}
         & Random binary codes & $L\ge\frac{1-2^{-\Omega_p(1/\delta)}}{\delta}$ whp &\cite{li-wootters-2018-ld-rand-lin} \\\cline{2-4}
         & Random binary linear codes & $L\le H(p)/\delta+2$ whp &\cite{li-wootters-2018-ld-rand-lin} \\\cline{2-4}
         & Random binary linear codes & $ L\le\floor{H(p)/\delta} + 1 $ whp & \cite{mosheiffetal-2020-bounds-list-dec-rand-lin} \\\cline{2-4}
         & Random linear codes & $L=\cO_{p,q}(1/\delta)$ whp &\cite{guruswami-hastad-kopparty-2010-ld-rand-lin} \\\cline{2-4}
         & Random linear codes & $ L\ge \floor{H_q(p)/\delta+0.99}-1 $ whp & \cite{mosheiffetal-2020-bounds-list-dec-rand-lin} \\\cline{2-4}
         & Arbitrary codes & $L=\Omega_{p,q}(\log(1/\delta))$ &\cite{blinovsky-1986-ls-lb-binary, blinovsky-2005-ls-lb-qary, blinovsky-2008-ls-lb-qary-supplementary,guruswami-vadhan-2005-ls-lb,  guruswami2013combinatorial} \\\hline
        \multirow{6}{0.12\textwidth}{Erasure} & Random {binary} codes & $L\le \frac{1-p+H(p)}{\delta}-1$ whp &\cite{guruswami-lncs}, Theorem 10.9 \\\cline{2-4}
         & Random codes & $L\ge\frac{1-p}{2\delta}$ whp &\cite{guruswami2013combinatorial} \\\cline{2-4}
         & Arbitrary {binary} codes & $L\ge\log(1+p/\delta)$ &\cite{guruswami-lncs}, Theorem 10.14 \\\cline{2-4}
         & Random {binary} linear codes & $L\le 2^{H(p)/\delta}-1$ whp &\cite{guruswami-lncs}, Theorem 10.11 \\\cline{2-4}
         & Random linear codes & $L\ge\frac{1}{q}2^{\frac{p(1-p)}{16\delta}}$ whp &\cite{guruswami2013combinatorial}  \\\cline{2-4}
         & Arbitrary {binary} linear codes & $L\ge1+p/\delta$ &\cite{cohen-etal-it1994} \\\hline
    \end{tabular}
\end{table}

It is worth mentioning that recently there are several breakthroughs towards explicit constructions of ``good" codes in the high-noise low-rate regime using tools from pseudorandomness.  Ta-Shma~\cite{ta-shma-2017-explicit-gv} constructed an \emph{explicit} $\delta$-balanced\footnote{A binary linear code is said to be \emph{$\delta$-balanced} if the weight of each codeword is between $\frac{1-\delta}{2}n$ and $\frac{1+\delta}{2}n$.} binary linear code of rate $\Omega(\delta^{2+\gamma})$ where $\gamma=\Theta\paren{\paren{\frac{\log\log\frac{1}{\delta}}{\log\frac{1}{\delta}}}^{1/3}}=o(1)$, \emph{almost} matching the Gilbert--Varshamov bound in this regime. Ta-Shma's beautiful and ingenious construction is done by casting the problem in the language of $\delta$-balanced sets and analyzing a \emph{long} random walk on a carefully constructed expander graph. His result is concerned with explicit codes with good distances. A more relevant result to our topic is an explicit construction of a \emph{near-optimal} erasure list decodable code~\cite{ben-aroya-doron-ta-shma-2018-explicit-erasure-ld}. The authors viewed the problem as constructing explicit dispersers and managed to construct an explicit binary \emph{nonlinear} $(1-\tau,\poly\log\frac{1}{\tau})$-erasure list decodable code of rate $\tau^{1+\gamma}$ (for any small constant $\gamma>0$), borrowing tools from the theory of extractors. The list size and the rate they got are both near-optimal. 

Going beyond combinatorial bounds and constructions, there is also research regarding efficient list decoding algorithms. For instance, recently Dinur et al.~\cite{dinur-harsha-kaufman-navon-ta-shma-2018-eff-ld} showed how \emph{double samplers} give rise to a generic way of amplifying distance so as to enable efficient list decoding algorithms. 
Followup works by Alev et al.~\cite{alev-2020-list-dec-explicit}, Jeronimo et al.~\cite{jeronimo-unique-dec-explicit} and Jeronimo et al.~\cite{jeronimo-etal-2021-linear-time-dec-tashma} equipped Ta-Shma's codes (and their variants) with efficient list decoding and unique decoding algorithms using connections to high-dimensional expanders and the Sum-of-Square hierarchy. 


As we saw, the list size problem is not well understood under the adversarial model. However, it is completely characterized if we are willing to further relax the problem by limiting the adversary to be \emph{oblivious}. Specifically, we call the adversary \emph{omniscient} if the error pattern is a function of the transmitted codeword, i.e., the adversary sees the codeword before he designs the attack vector. Otherwise, an adversary is said to be \emph{oblivious} if the error pattern is independent of the transmitted codewords, i.e., the adversary knows nothing more than the codebook and has to design the attack vector before the codeword is transmitted. The list size-vs.-rate tradeoff is known to a fairly precise extent for \emph{general} discrete memoryless oblivious adversarial channels.

For a general oblivious discrete memoryless Arbitrarily Varying Channel (AVC) $W(y|x,s)$ \emph{without} constraints, Hughes~\cite{hughes-1997-list-avc} completely characterized its list decoding capacity under \emph{any} $L$. Specifically, 
the list-$L$ capacity $C(L)$ equals 
\begin{align}
C(L) = \max_{P_\bfx}\min_{P_\bfs}I(\bfx;\bfy) \label{eqn:list_l_cap}
\end{align}
if 
\begin{align}
L >& L^* \coloneqq \max\curbrkt{ \ell\in\bZ_{\ge0}\colon 
\begin{array}{l}
\exists U(s|x_1,\cdots,x_\ell),\forall x_0,x_1,\cdots,x_\ell,\forall\pi\in S_{\ell+1},\\
\sum_sU(s|x_1,\cdots,x_\ell)W(y|x_0,s) = \sum_sU(s|x_{\pi(1)},\cdots,x_{\pi(\ell)})W(y|x_{\pi(0)},s) 
\end{array}
}; \label{eqn:sym}
\end{align}
and $ C(L) = 0 $ otherwise. 
For oblivious AVCs \emph{under} constraints, the critical list size $ L^* $ is known~\cite{zhang-2020-obli-list-dec} though the exact capacity $C(L)$ is open. 

In this paper, we will focus on combinatorial/information-theoretic limits of list decoding various ensembles of random codes over $\bR$ against omniscient adversaries.


\subsection{Prior work on list decoding over the reals}
In this section, we briefly recall what is known about list decoding over the reals. As we will see, it is much less studied and understood, which motivates this work. 
\begin{enumerate}
	\item Shlosman and Tsfasman~\cite{shlosman-tsfasman-2000-random-packing} studied the sphere packing density of a) a random lattice sampled from the Haar distribution; b) a Poisson point process (PPP).
	\item Blinvosky~\cite{blinovsky-2005-random-packing} later generalized their results on PPP to list-$L$ packing for any constant $L\in\bZ{{\ge1}}$. 
	However, the proof therein is problematic and is recently corrected in~\cite{zhang-vatedka-2021-multipack} without affecting the result. 
	PPPs are a natural family of ICs. 
	The heuristic results for Haar lattices in Sec.~\ref{sec:cond_ld_haar_dist} of this paper can be cast as (rigorous) list size bounds for PPPs. 
	The bounds in \cite{blinovsky-2005-random-packing,zhang-vatedka-2021-multipack} provide finite list size bounds for PPPs which subsume our bounds in Sec.~\ref{sec:cond_ld_haar_dist}. 
	However, our proof for $ \cO\paren{\frac{1}{\delta}\log\frac{1}{\delta}} $ list sizes is much simpler. 
	\item Some bounds~\cite{blinovsky-1997-list-dec-real,blinovsky-litsyn-2009-list-dec-real} on the list-$L$ capacity of arbitrary spherical codes are also known.  
	Note that the proof in~\cite{blinovsky-1997-list-dec-real} is problematic and is recently corrected in~\cite{zhang-vatedka-2021-multipack} without affecting the result. 
	\item In computer science, there have been results~\cite{grigorescu-peikert-2012-list-dec-barnes-wall,mook-peikert-2020-list-dec-lattice} on efficient list decoding algorithms for explicit lattices, e.g., Barnes--Wall lattices, Barnes--Sloane lattices, etc. 
	However, these lattices have rate/NLD way below the capacity. 
	\item Recently, the authors of the current paper systematically revisited the list decoding problem over $ \bR $ with \emph{constant} list sizes. 
	Various upper and lower bounds~\cite{zhang-vatedka-2021-multipack} were derived for this problem. 
\end{enumerate}



\section{Notation}\label{sec:notation}
\noindent\textbf{General notation.}
We use standard Bachmann-Landau (Big-Oh) notation for asymptotic functions. 

For any $q\in\bR_{>0}$, we write $\log_q(\cdot)$ for the logarithm to the base $q$. In particular, let $\log(\cdot)$ denote the logarithm to the base two and let $\ln(\cdot)$ denote the logarithm to the base $e$.

\noindent\textbf{Sets.}
For any two sets $\cA$ and $\cB$ with additive and multiplicative structures, let $\cA+\cB$ and $\cA\cdot\cB$ denote the Minkowski sum and Minkowski product of them which are defined as
\[\cA+\cB\coloneqq\curbrkt{a+b\colon a\in\cA,b\in\cA},\quad\cA\cdot\cB\coloneqq\curbrkt{a\cdot b\colon a\in\cA,b\in\cB}.\]
If $\cA=\{x\}$ is a singleton set, we write $x+\cB$ and $x\cB$ for $\{x\}+\cB$ and $\{x\}\cdot\cB$.
For any finite set $\cX$ and any integer $0\le k\le |\cX|$, we use $\binom{\cX}{k}$ to denote the collection of all subsets of $\cX$ of size $k$, i.e., 
\[\binom{\cX}{k}\coloneqq\curbrkt{\cY\subseteq\cX\colon\card{\cY}=k}.\]

For $M\in\bZ_{>0}$, we let $[M]$ denote the set of the first $M$ positive integers $\{1,2,\cdots,M\}$. 

For a subset $\cT=\curbrkt{i_1,\cdots,i_t}\subseteq[n]$ of $t$ coordinates and a vector $\vx\in\cX^n$ over some alphabet $\cX$, we use $\vx|_\cT$ to denote the vector obtained by restricting $\vx$ to the coordinates in $\cT$, i.e., 
\[\vx|_\cT\coloneqq[\vx_{i_1},\cdots,\vx_{i_t}]^\top.\]
Similar notation can be defined for a subset $\cA$ of $\cX^n$
\[\cA|_\cT\coloneqq\curbrkt{\vx|_\cT\colon \vx\in\cA}.\]

For any $\cA\subseteq\cX$, the indicator function of $\cA$ is defined as, for any $x\in\cX$, 
\[\ind{\cA}(x)=\begin{cases}1,&x\in \cA\\0,&x\notin \cA\end{cases}.\]
At times, we will slightly abuse notation by saying that $\ind{\mathsf{A}}$ is $1$ when event $\mathsf{A}$ happens and zero otherwise.

Let $\|\cdot\|_2$ denote the Euclidean/$L^2$-norm. Specifically, for any $\vx\in\bR^n$,
\[\|\vx\|_2\coloneqq\paren{\sum_{i=1}^n\vx_i^2}^{1/2}.\]
For brevity, we also write $\|\cdot\|$ for the $L^2$-norm.

Let $\vol_n(\cdot)$ denote the $n$-dimensional Lebesgue volume of an Euclidean body. Specifically, for any Euclidean body $\cA\subseteq\bR^n$,
\[\vol_n(\cA)=\int_\cA\diff \vx=\int_{\bR^n}\ind{\cA}(\vx)\diff\vx,\]
where $\diff\vx$ denotes the differential of $\vx$ with respect to (wrt) the Lebesgue measure on $\bR^n$. When the dimension $n$ is obvious from the context, we will also use the shorthand notation $|\cdot|$ for $\vol_n(\cdot)$. 
If $\cA\subseteq\bR^n$ is an $n$-dimensional body with nonempty interior, we write $\vol(\cA)$ for $\vol_n(\cA)$; if $\cA\subseteq\bR^n$ is an $(n-1)$-dimensional hypersurface, we write $\area(\cA)$ for $\vol_{n-1}(\cA)$.

Sets are denoted by capital letters in calligraphic typeface, e.g., $\cC,\cI$, etc. In particular, let $\cS_2^{n-1}$ denote the $(n-1)$-dimensional unit Euclidean sphere wrt $L^2$-norm, i.e., 
\[\cS_2^{n-1}\coloneqq\curbrkt{\vy\in\bR^n\colon \|\vy\|_2=1}.\]
Let $\cB_2^n$ denote the $n$-dimensional unit Euclidean ball wrt $L^2$-norm, i.e., 
\[\cB_2^n\coloneqq\curbrkt{\vy\in\bR^n\colon \|\vy\|_2\le1}.\]
An $(n-1)$-dimensional Euclidean sphere centered at $\vx$ of radius $r$ is denoted by
\[\cS_2^{n-1}(\vx,r)=\vx+r\cS_2^{n-1}=\{\vy\in\bR^n:\|\vy\|_2=r\}.\]
An $n$-dimensional Euclidean ball centered at $\vx$ of radius $r$ is denoted by
\[\cB_2^n(\vx,r)=x+r\cB_2^n=\{\vy\in\bR^n:\|\vy\|_2\le r\}.\]

For any $\vx\in\bF_q^n$, let $\wth{\vx}$ denote the Hamming weight of $\vx$, i.e., the number of nonzero entries of $\vx$:
\[\wth{\vx}\coloneqq\curbrkt{i\in[n]\colon \vx_i\ne0}.\]
For any $\vx,\vy\in\bF_q^n$, let $\disth(\vx,\vy)$ denote the Hamming distance between $\vx$ and $\vy$, i.e., the number of locations where they differ:
\[\disth(\vx,\vy)\coloneqq\wth{\vx-\vy}=\curbrkt{i\in[n]\colon \vx_i\ne\vy_i}.\]
We can  define balls and spheres in $\bF_q^n$ centered around some point of certain radii wrt Hamming metric as well:
\[\sham^n(\vx,r)\coloneqq\curbrkt{\vy\in\bF_q^n\colon \disth(\vx,\vy)=r},\quad\bham^n(\vx,r)\coloneqq\curbrkt{\vy\in\bF_q^n\colon \disth(\vx,\vy)\le r}.\]

We will drop the subscript and superscript for the associated metric and dimension  when they are clear from the context.

\noindent\textbf{Probability.}
Random variables are denoted by lower case letters in boldface or capital letters in plain typeface, e.g., $\bfm,\bfx,\bfs,U,W$, etc. Their realizations are denoted by corresponding lower case letters in plain typeface, e.g., $m,x,s,u,w$, etc. Vectors of length $n$, where $n$ is the blocklength, are denoted by lower case letters with underlines, e.g., $\vbfx,\vbfs,\vx,\vs$, etc. The $i$-th entry of a vector is denoted by a subscript $i$, e.g., $\vbfx_i,\vbfs_i,\vx_i,\vs_i$, etc. Matrices are denoted by capital letters in boldface, e.g., $\bfI,\mathbf{\Sigma}$, etc. The $(i,j)$-th entry of a matrix $\bfM$ is denoted by $\bfM_{ij} $.

The probability mass function (pmf) of a discrete random variable $\bfx$ or a random vector $\vbfx$ is denoted by $p_{\bfx}$ or $p_{\vbfx}$ respectively. Here with a slight abuse of notation, we use the same to denote the probability density function (pdf) of $\bfx$ or $\vbfx$ if they are continuous.
If every entry of $\vbfx$ is independent and identically distributed (iid) according to $p_{\bfx}$, then we write $\vbfx\sim p_{\bfx}^{\tn}$. In other words,
\[p_{\vbfx}(\vx)=p_{\bfx}^{\tn}(\vx)\coloneqq\prod_{i=1}^np_{\bfx}(\vx_i).\]
Let $\cU(\Omega)$ denote the uniform distribution over some probability space $\Omega$. Let $\cN(\underline{\mu},\mathbf\Sigma)$ denote the $n$-dimensional Gaussian distribution with mean vector $\underline{\mu}$ and covariance matrix $\mathbf\Sigma$. 

We use $H(\cdot)$ to denote interchangeably Shannon entropy and differential entropy; the exact meaning  will usually be clear from context.
In particular, if $p_{\vbfx}\colon \bR^n\to\bR_{\ge0}$ is a pdf of a random vector $\vbfx$ in $\bR^n$ , $H(\vbfx)$ denotes the differential entropy of $\vbfx\sim p_{\vbfx}$,
\[H(\vbfx)=-\int_{\bR^n} p_{\vbfx}(\vx)\log p_{\vbfx}(\vx)\diff \vx.\]
For any $p\in[0,1]$, $H(p)$ denotes the binary entropy 
\[H(p)=p\log\frac{1}{p}+(1-p)\log\frac{1}{1-p}.\]
For any $q\in\bZ_{\ge2}$ and any $p\in[0,1]$, $H_q(p)$ denotes the $q$-ary entropy 
\[H_q(p)=p\log_q(q-1)+p\log_q\frac{1}{p}+(1-p)\log_q\frac{1}{1-p}.\]

\noindent\textbf{Algebra.}
For any field $F$, we use $\speclin(n,F)$ and $\genlin(n,F)$ to denote the special linear group and the general linear group over $F$ of degree $n$, i.e., 
\[\speclin(n,F)\coloneqq\curbrkt{\bfM\in F^{n\times n}\colon \det(\bfM)=1},\quad \genlin(n,F)\coloneqq\curbrkt{\bfM\in F^{n\times n}\colon \det(\bfM)\ne 0}.\]

For any vector space $V$ of dimension $n$ and any integer $0\le k\le n$, the Grassmannian $\gr(k,V)$ is the collection of all $k$-dimensional subspaces of $V$, i.e., 
\[\gr(k,V)\coloneqq\curbrkt{U\le V\text{ subspace}\colon \dim_F U=k}.\]

\section{Preliminaries}\label{sec:prelim}
\noindent\textbf{Probability.} 
We will need the following form of Chernoff bound.
\begin{lemma}
\label{lem:chernoff}
Let $ X_1, \cdots, X_N $ be independent Bernoulli random variables and let $ X\coloneqq\sum_{i = 1}^NX_i $. 
Then for any $ 0\le \delta\le1 $, we have
\begin{align}
\prob{X\ge(1+\delta)\expt{X}} \le& \exp\paren{-\frac{\delta^2}{3}\expt{X}}, \notag \\
\prob{X\le(1-\delta)\expt{X}} \le& \exp\paren{-\frac{\delta^2}{2}\expt{X}}. \notag 
\end{align}
\end{lemma}

The following lemma is an easy corollary of Chebyshev inequality.
\begin{lemma}\label{lem:cheb_cor}
	For any nonnegative random variable $X$, $\prob{X=0}\le\var{X}/\expt{X}^2$.
\end{lemma}

Recall two facts about the moments of Gaussian and Poisson random variables.
\begin{fact}\label{fact:gaussian_mmt}
	Let $\bfg\sim\cN(0,1)$, then 
	\[\expt{\bfg^k}=\begin{cases}
	0,&k\text{ odd}\\
	(k-1)!!,&k\text{ even}
	\end{cases},\]
	where $\ell!!\coloneqq \ell(\ell-2)(\ell-4)\cdots3\cdot 1$ denotes the double factorial of $\ell$ odd.
\end{fact}
\begin{fact}\label{fact:pois_mmt}
	Let $\bfp\sim\pois(\lambda)$, then
	\[\expt{\bfp^k}=e^{-\lambda}\sum_{i=0}^\infty\frac{i^k}{i!}\lambda^i .
	\]
\end{fact}
Poisson random variables are additive.
\begin{fact}\label{fact:pois_add}
	If $\bfp_1\sim\pois(\lambda_1)$ and $\bfp_2\sim\pois(\lambda_2)$ are independent, then $\bfp_1+\bfp_2\sim\pois(\lambda_1+\lambda_2)$.
\end{fact}
We know the following tail bounds for Poisson random variables.
\begin{lemma}\label{lem:pois_tail}
	Let $\bfp\sim\pois(\lambda)$ and $\ell>\lambda,m<\lambda$, then
	\begin{align*}
	\prob{\bfp> \ell}<& {\frac {e^{-\lambda }(e\lambda )^{\ell}}{\ell^{\ell}}},\quad
	\prob{\bfp<m}< {\frac {e^{-\lambda }(e\lambda )^{m}}{m^{m}}}.
	\end{align*}
\end{lemma}

\begin{lemma}[\cite{canonne-poisson-tail}]
\label{lem:pois-tail-clement}
Let $ \bfp\sim\pois(\lambda) $ and $ \Delta>0 $, then
\begin{align}
\prob{\bfp - \lambda\ge\Delta} &\le e^{\frac{\Delta^2}{2(\lambda+\Delta)}}, \notag \\
\prob{\bfp - \lambda\le-\Delta} &\le e^{\frac{\Delta^2}{2(\lambda+\Delta)}}, \notag \\
\prob{|\bfp - \lambda|\ge\Delta} &\le 2e^{\frac{\Delta^2}{2(\lambda+\Delta)}}. \notag 
\end{align}
\end{lemma}

\noindent\textbf{Geometry.}
It is well-known that Stirling's approximation gives an asymptotic expression for factorials.
\begin{lemma}\label{lem:stirling}
	For any $n\in\bZ_{>0}$, $n!=\sqrt{2\pi n}(n/e)^n(1+o(1))$.
\end{lemma}
We can use the above lemma to obtain the asymptotic behaviour of binomial coefficients. At times, we also resort to the following cheap yet convenient bounds.
\begin{lemma}\label{lem:bd_binom_coeff}
	For any $n\in\bZ_{>0}$ and $0\le k\le n$, $(n/k)^k\le\binom{n}{k}\le(en/k)^k$.
\end{lemma}

Recall the formulas and asymptotics of the volume of a unit Euclidean ball and the area of a unit Euclidean sphere.
\begin{fact}\label{fact:vol_ball}
	$V_n\coloneqq\vol(\cB_2^n)=\frac {\pi ^{\frac {n}{2}}}{\Gamma ({ {n}/{2}}+1)}= \frac{1}{\sqrt{\pi n}}\paren{\frac{2\pi e}{n}}^{n/2}(1+o(1))$.
\end{fact}
\begin{fact}\label{fact:area_sphere}
	$A_{n-1}\coloneqq\area( \cS^{n-1}_2)={\frac {n\pi ^{\frac {n}{2}}}{\Gamma ({ {n}/{2}}+1)}}= \sqrt{\frac{n}{\pi}} \paren{\frac{2\pi e}{n}}^{n/2}(1+o(1))$.
\end{fact}

\section{List decodability of spherical codes}\label{sec:lb_ls_spherical}
We now investigate lower bounds on the list size $ L $ for codes that operate at rate $ R=C(P,N) -\delta$.

\subsection{A reduction from an arbitrary code to a spherical code}
We first show that it suffices to prove a lower bound on list size for spherical codes.

\begin{lemma}
	Suppose there exists a $ (P,N,L) $-list decodable code $ \cC \subset\cB_2^n(0,\sqrt{nP}) $ of rate $ R $. Then, there exists a $\paren{P,N,\frac{P}{4N}L}$-list decodable code $ \cC' \subset \cS_2^{n-1}(0,\sqrt{nP})$ of asymptotically the same rate.
	\label{lemma:reduction_balltospherecode}
\end{lemma}
\begin{proof}
	\begin{figure}
		\centering
		\includegraphics{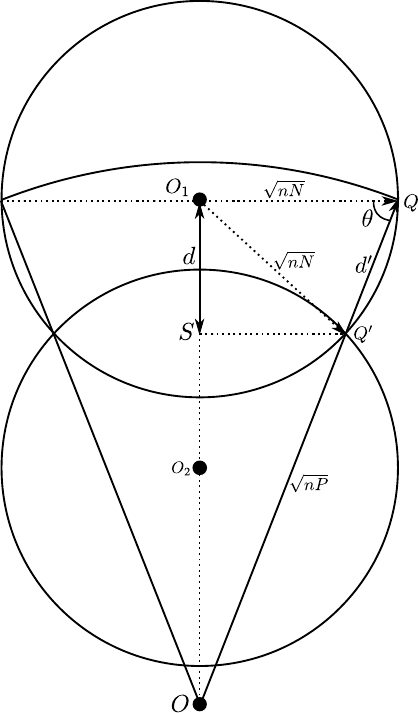}
		\caption{A covering of the cone using balls.}
		\label{fig:cone_covering}
	\end{figure}
	Given any $\paren{P,N,L}$-list decodable (ball) code $\cC$ in $\cB_2^n(0,\sqrt{nP})$, we can construct a $\paren{P,N,\frac{P}{4N}L}$-list decodable spherical code $\cC'$ on $ \cS_2^{n-1}(0,\sqrt{nP}) $. Indeed, we just project all codewords radially onto $\cS^{n-1}\paren{0,\sqrt{nP}}$. Then we know that for any direction $ \vtheta\in\cS^{n-1} $, 
	\begin{align}
	\card{\cC'\cap\C^{n-1}(\vtheta,\sqrt{nN})}\le\card{\cN}L, \label{eqn:bound_spherical_code}
	\end{align}
	where 
	$\C^{n-1}(\vtheta,\sqrt{nN})$ is a cap of radius $\sqrt{nN}$ on the sphere $ \cS^{n-1}(0,\sqrt{nP}) $ along direction $\vtheta$,
	\[\C^{n-1}(\vtheta,\sqrt{nN})\coloneqq\curbrkt{\vx\in\cS^{n-1}(0,\sqrt{nP})\colon \inprod{\vx}{\vtheta}\ge\sqrt{n(P-N)}};\]
	and
	$\cN$ is a $\sqrt{nN}$-covering\footnote{A \emph{$\Delta$-covering} (a.k.a. a $\Delta$-net) $\cN$ of a metric space $ (\cX,d) $ is a subset $ \cN\subset\cX $ satisfying that for any $ x\in\cX $, there exists an $ x'\in\cN $ such that $ d(x,x')\le\Delta $.} of the cone 
	\[\cK(\vtheta)\coloneqq\curbrkt{\lambda\vx\colon \vx\in\C^{n-1}(\vtheta,\sqrt{nN}),\lambda\in[0,1]}\]
	induced by the cap. 

	We can upper bound $\card{\cN}$ by  
	\begin{align}
	|\cN| \le& \frac{\sqrt{P-N}}{2d}, \label{eqn:bound_covering_size}
	\end{align}
	where $d$ is shown in Fig.~\ref{fig:cone_covering} and will be computed momentarily. This can be seen by staring at the geometry of a covering as shown in Figure~\ref{fig:cone_covering}. 
	One way to cover the cone $\cK(\vtheta)$ is to align the centers of the balls $ \cB(\cdot,\sqrt{nN}) $ on the ray rooted at the center $O$ of $ \cS^{n-1}(0,\sqrt{nP}) $ in the direction $ \vtheta $.
	We enumerate such ball in ascending order from the surface to the center of the sphere $ \cS^{n-1}(0,\sqrt{nP}) $. 
	That is, the ball whose center is closest to the surface of $ \cS^{n-1}(0,\sqrt{nP}) $ is the 1-st one and the ball whose center is closest to the center of $ \cS^{n-1}(0,\sqrt{nP}) $ is the $ |\cN| $-th one.
	Let $ 2d_i $ denote the distance between the centers of the $i$-th and the $(i+1)$-st balls. 
	Since all centers are on the segment $ OO_1 $ of length $ \sqrt{n(P-N)} $ and we apparently have $ d_1<d_2<\cdots<d_{|\cN|-1} $, we can upper bound $|\cN|$ by $ \frac{\sqrt{P-N}}{2d} $ where $ d \coloneqq d_1 $. 

	We now compute $d$.
	By symmetry, the distance between the centers $O_1$ and $O_2$ of the first two balls is equal to $2d$ where $d\coloneqq{SO_1}/\sqrt{n}={SO_2}/\sqrt{n}$. Since the triangles $\Delta OSQ'$ and $\Delta OO_1Q$ are similar,  $d$ is given by the following equation
	\begin{equation}
	\frac{\card{OS}}{\card{OO_1}}=\frac{\card{OQ'}}{\card{OQ}} \iff \frac{\sqrt{P-N}-d}{\sqrt{P-N}}=\frac{\sqrt{P}-d'}{\sqrt{P}},
	\label{eqn:sim_tri}
	\end{equation}
	where $d'\coloneqq Q'Q/\sqrt{n}$. On the other hand, the triangle $\Delta O_1Q'Q$ is isosceles with side length $O_1Q'=O_1Q=\sqrt{nN}$. Let $ \theta\coloneqq\angle O_1QQ' $. It is immediate that $d'=2\sqrt{N}\cos\theta=2N/\sqrt{P}$ since $\cos\theta=QO_1/QO=\sqrt{N/P}$ in $\Delta QO_1O$.
	Plugging it into  Eqn.~\eqref{eqn:sim_tri} and solving $d$, we have $d=2\frac{N}{P}\sqrt{P-N}$. Hence by Eqn.~\eqref{eqn:bound_covering_size}, $\card{\cN}\le\frac{P}{4N}$.
	Substituting it to Eqn.~\eqref{eqn:bound_spherical_code} finishes the proof. 
	
\end{proof}

\begin{remark}
In fact, since we are covering a cone rather than a cylinder, the most economical way of covering is not to align the balls with consecutive distance $2d$. Indeed, the optimal covering $\cN^*$ has strictly increasing distances $2d=d_{1}<d_{2}<\cdots<d_{|\cN^*|-1}$, where $d_{i}$ is the half distance between the centers $O_i$ and $O_{i+1}$ of the $i$-th and the $(i+1)$-st balls. One can compute each $d_i$ explicitly. Although our bound is crude, it  is still a valid and simple upper bound and is tight for covering a cylinder.
\end{remark}

\subsection{List size lower bound for uniformly random spherical codes}
Although we are not able to obtain a lower bound for arbitrary spherical codes as in \cite{blinovsky-1997-list-dec-real,blinovsky-litsyn-2009-list-dec-real}, we can obtain a lower bound for uniformly random spherical codes.

\begin{proposition}
	Fix $ P>N>0 $, and let $ C=\frac{1}{2}\log\frac{P}{N} $. For every $ \delta>0 $, if $ \cC $ is a random spherical code on $ \cS^{n-1}(0,\sqrt{nP}) $ of rate $ C-\delta $, then
	\[
	\Pr\left[\cC\text{ is }\paren{P,N,\frac{c'}{\delta}-1}\text{-list decodable}\right] \leq 2^{-\Theta(n)},
	\]
	for every $ c'>C $.
	\label{prop:lbound_randomspherical}
\end{proposition}
\begin{proof}
	The proof follows a  second-moment method as in Guruswami and Narayanan~\cite{guruswami2013combinatorial} for binary codes. 
	
	Choose a $\sqrt{n\eps}$-net $\cY$ for $\cS^{n-1}\paren{0,\sqrt{n\paren{P-N}}}$ for some constant $ \eps>0 $. In other words, $ \cY\subset \cS^{n-1}\paren{0,\sqrt{n\paren{P-N}}} $ and for all $ \vy\in \cS^{n-1}\paren{0,\sqrt{n\paren{P-N}}} $, we have $ \min_{\vu\in \cY}\Vert\vy-\vu\Vert \leq \sqrt{n\epsilon} $.
	
	For any spherical code $ \cC $, define
	\begin{equation}
	    W\coloneqq\sum_{\vy\in\cY}\sum_{\curbrkt{m_1,\cdots,m_L}\in\binom{\cM}{L}}\indicator{\psi\paren{m_1},\cdots,\psi\paren{m_L}\in\cB^n\paren{\vy,\sqrt{nN}}},
	    \label{eqn:w_spherical}
	\end{equation}
	where $\cM\coloneqq\curbrkt{0,1,\cdots,2^{nR}-1}$ is the set of messages and $\psi\paren{m}$ denotes the codeword corresponding to $m$. Let $M\coloneqq\card{\cM}=2^{nR}$. Clearly, $ W=0 $ if and only if (iff) $\cC$ is $ (P,N,L-1) $-list decodable. 
	\begin{align}
	\prob{\cC\text{ is }(P,N,L-1)\text{-list decodable}}&=\prob{\bigcap_{\vy\in\cB^n\paren{0,\sqrt{nP}+\sqrt{nN}}\setminus\cB^n\paren{0,\sqrt{nP}-\sqrt{nN}}}\curbrkt{\card{\cC\cap\cB^n\paren{\vy,\sqrt{nN}}}<L}}\notag\\
	&\le\prob{\bigcap_{\vy\in\cY}\curbrkt{\card{\cC\cap\cB^n\paren{\vy,\sqrt{nN}}}<L}}\notag\\
	&=\prob{W=0}\notag\\
	&\le\var{W}/\expt{W}^2,\label{eq:p_ld_bound1}
	\end{align}
	where the last inequality~\eqref{eq:p_ld_bound1} follows from Lemma~\ref{lem:cheb_cor}.
	Let 
	\[\mu\coloneqq\frac{\area\paren{\C^{n-1}\paren{\sqrt{nN}}}}{\area\paren{\cS^{n-1}\paren{\sqrt{nP}}}}, \quad 
	\nu \coloneqq \frac{\card{\cY\cap\C^{n-1}(\sqrt{nN(P-N)/P})}}{|\cY|}.\]
	Then, we show that 
	\begin{align}
	\bE[W] \geq \paren{{M}/{L}}^L |\cY|\mu^L,
	\label{eqn:expt-toshow}
	\end{align}
	and
	\begin{align}
	\var{W} \leq |\cY|^2L\nu^2M^L\mu^{L-1}.
	\label{eqn:var-toshow}
	\end{align}
	See Appendix~\ref{sec:lbound_EW} and~\ref{sec:ubound_varW} for the details.
	Plugging these in Eqn.~\eqref{eq:p_ld_bound1}, we get
	\begin{align}
	\prob{\cC\text{ is }(P,N,L-1)\text{-list decodable}}&\leq L^{2L+1}\nu^2\mu^{-L-1}M^{-L}.\label{eqn:bound}
	\end{align}
	We need an upper bound on $ \nu $ which is given by Eqn.~\eqref{eqn:nu-bound}. 
	Let $ c_2 \coloneqq \frac{3\sqrt{P}}{2(\sqrt{N(P-N)/P} + 3\sqrt{\eps}/2)} $. 
	Eqn.~\eqref{eqn:nu-bound} implies
	\begin{align}
	\nu &\le c_2\paren{\frac{\sqrt{N(P-N)/P} + 3\sqrt{\eps}/2}{\sqrt{P-N}}}^{n} \notag \\
	&= c_22^{n\log\paren{\sqrt{N/P} + \frac{3\sqrt{\eps}}{2\sqrt{P-N}}}} \notag \\
	&= c_22^{-\frac{n}{2}\log{\frac{P}{N}} + n\log\paren{1 + \frac{3\sqrt{\eps}\sqrt{P/N}}{2\sqrt{P-N}}}}  \notag \\
	&\le c_22^{-n\paren{\frac{1}{2}\log\frac{P}{N} - \eps'}}, \notag 
	\end{align}
	where the last inequality follows since $ \log(1+x)\le2x $ for $ x\ge0 $ and $ \eps'\coloneqq 3\sqrt{\frac{P\eps}{N(P-N)}} $.

	We also need a lower bound on $ \mu $: 
	\begin{align}
	\mu&\ge \frac{\vol\paren{\cB^{n-1}\paren{0,\sqrt{nN}}}}{\area\paren{\cS^{n-1}\paren{0,\sqrt{nP}}}}\notag\\
	&=c_3 2^{-n\left(\frac{1}{2}\log\frac{P}{N}+o(1)\right)},\label{eqn:mu_lb}
	\end{align}
	for some constant $ c_3>0 $. 
	The probability~\eqref{eqn:bound} we want to upper bound is at most
	\begin{align}
	\prob{\cC\text{ is }(P,N,L-1)\text{-list decodable}}&\leq L^{2L+1}c_2^2c_3^{-L-1} 2^{-2n\paren{\frac{1}{2}\log\frac{P}{N} - \eps'}} 2^{-n\paren{-L-1}\paren{\frac{1}{2}\log\frac{P}{N}+o(1)}-nRL} \notag \\ 
	&=L^{2L+1}c_2^2c_3^{-L-1}2^{n\paren{\delta L-\frac{1}{2}\log\frac{P}{N}+2\eps'+o(1)}} \notag \\
	&= L^{2L+1}c_2^2c_3^{-L-1} 2^{n\paren{\delta L - \frac{1}{2}\log\frac{P}{N} + \delta + o(1)}}. \notag 
	\end{align}
	In the last equation, we set $ \eps' = \delta/2 $, i.e., $ \eps = \frac{N(P-N)}{P}\paren{\frac{\delta}{6}}^2 $. 
	The constant-in-$n$ terms downstairs and the $o(n)$ term in the exponent are not important. The probability that $\cC$ is list decodable vanishes in $n$ when $L<\frac{\frac{1}{2}\log\frac{P}{N}}{\delta} - 1$. That is to say, for a uniformly random spherical code to be $\paren{P,N,L-1}$-list decodable with high probability, $L$ has to be at least $C/\delta - 1$, where $C=\frac{1}{2}\log\frac{P}{N}$.	
\end{proof}

We would like to emphasize that the above result only implies that a typical random code is not $ (P,N,c'/\delta-1) $-list decodable \emph{with high probability}. This does not claim the \emph{non-existence} of $ (P,N,c'/\delta-1) $-list decodable codes of rate $ C-\delta $.

\section{List decodability of nested Construction-A lattice codes}\label{sec:ld_constr_a}

\subsection{Nested lattice codes}

Recall that a lattice  $ \Lf $ is a discrete subgroup of $ \bR^n $, and can be written as $\bfG\bZ^n$ where $\bfG$ is called a generator matrix of $\Lf$. For a quick introduction to lattices and related definitions, see Appendix~\ref{sec:primer_lattices}. The concepts we use are quite standard in the literature on lattices~\cite{erez2005lattices,zamir2014latticebook}. 

We say that a lattice $ \Lc $ is nested in $ \Lf $ if $ \Lc\subsetneq \Lf $. Let $ Q_{\Lf}(\vx) $ denote the closest point in $ \Lf $ to $ \vx $, and $ [\vx]\bmod\Lf\coloneqq \vx-Q_{\Lf}(\vx) $. Let $ \cV(\Lc) \coloneqq Q_{\Lf}^{-1}(0) $ denote the fundamental Voronoi region of $ \Lc $. Further, let $ \rcov(\Lf) ,\reff(\Lf),\rpack(\Lf)$ respectively denote the covering, effective, and packing radii of $ \Lf $. The determinant (or covolume) of $ \Lf $ is equal to the volume of $\cV(\Lf)  $ and is denoted by $ \det \Lf $.

Our goal is to construct good nested lattice pairs $ (\Lf,\Lc) $ with $ \Lc\subset \Lf $, and our nested lattice code will be defined as $ \cC \coloneqq \Lf\cap\cV(\Lc) $. The nested lattice code satisfies the power constraint if $ \rcov(\Lc)\leq \sqrt{nP} $. 

We now prove an upper bound on the list size for nested lattice codes.
Our goal is to show the following:
\begin{theorem}
	Let $0<\delta<0.9$ and $P>N$.
	There exist nested lattice codebooks of rate $\frac{1}{2}\log_2\frac{P}{N}-\delta$ that are $(P,N,2^{\cO(\frac{1}{\delta}\log_2^2\frac{1}{\delta})})$-list decodable.
	\label{thm:nestedconstructiona_listsize}
\end{theorem}

\subsection{List size upper bound for nested Construction-A lattice codes}\label{sec:nestedlatticeensemble}

We start with a (full rank) coarse lattice $\Lc$ that satisfies
\begin{equation}
\frac{\rcov(\Lc)}{\reff(\Lc)}\leq 2^{\delta/8}
\label{eq:rcov_condition}
\end{equation}
and
\begin{equation}
\frac{\rpack(\Lc)}{\reff(\Lc)} > \frac{1}{4}.
\label{eq:rpack_condition}
\end{equation}
Such lattices are guaranteed to exist (for sufficiently large $n$) by~\cite{erez2005lattices} (See Appendix~\ref{sec:primer_lattices}). 
The lattice is suitably scaled so that $\rcov(\Lc) = \sqrt{nP}$ and this will ensure that the codebook satisfies the power constraint.
Note that scaling the lattice by a constant factor scales $\rpack,\reff$ and $\rcov$ by the same amount, and the ratios in Eqn.~\eqref{eq:rcov_condition} and~\eqref{eq:rpack_condition} remain unchanged.
Let $\bfG_{\Lc}$ be a generator matrix for $\Lc$, and $q$ be the smallest prime number 
that satisfies
\begin{equation}
1+\frac{\sqrt{P}}{q\sqrt{N}} \leq 2^{\delta/8}.
\label{eq:prime_condition}
\end{equation}
Note that $q$ is independent of $n$ and is of order $ q=\Omega(1/\delta) $.
Bertrand's postulate guarantees that for every positive integer $m$, there exists a prime number between $m$ and $2m$.
Therefore, 
\begin{equation}
\frac{\sqrt{P/N}}{2^{\delta/8}-1} \leq q\leq 2\frac{\sqrt{P/N}}{2^{\delta/8}-1} +2.
\label{eq:primecondition_2}
\end{equation}
Let $R = \frac{1}{2}\log_2\frac{P}{N}-\delta$, and $\kappa$ be an integer such that\footnote{More accurately, $\kappa$ is the integer closest to $nR/\log_2q$. But we assume that $\kappa$ as defined above is an integer so that our proofs are cleaner.} 
\begin{align}
\frac{\kappa}{n}\log_2 q = R. \label{eqn:cond_kappa}
\end{align}

We define an ensemble of fine lattices as follows: Choose an $n\times \kappa$ generator matrix $\bfG_\lin$ uniformly over $\bFq^{n\times \kappa}$. This defines a linear code $\cC(\bfG_\lin)=\bfG_\lin\bFq^\kappa $ where the arithmetics are over $\bFq$.
Let $\Lf'\coloneqq \frac{1}{q}\Phi(\cC(\bfG_\lin)) + \bZ^n$, where $ \Phi $ is the natural embedding of $ \bFq^n $ into $ \bR^n $ and the arithmetics are over $\bR$. In other words, $ \Phi $ operates componentwise on vectors, and maps $ 0,1,\ldots,q-1\in\bF_q $ to $ 0,1,\ldots,q-1\in\bR $. 
Note that $ \bZ^n\subset\Lf'\subset q^{-1}\bZ^n $. 
Our fine lattice is $\Lf \coloneqq \bfG_{\Lc}\Lf'$. It is easy to verify that $\Lc$ is always a sublattice of $\Lf$.
In fact, $ \Lc\subset\Lf\subset q^{-1}\Lc $ forms a chain of nested lattices. 
Our nested lattice codebook is then $\cC\coloneqq\Lf\cap\cV(\Lc)$.

We will show the following result, which implies Theorem~\ref{thm:nestedconstructiona_listsize}.
\begin{theorem}
	If $P>N$, then
	\[
	\Pr[\Lf\cap \cV(\Lc) \text{ is not }(P,N,2^{\cO((\log_2^2\delta)/\delta)})\text{-list decodable}] = 2^{-\Omega(n)}.
	\]
	\label{thm:nestedlattice_listdecoding}
\end{theorem}
Note that the only randomness involved is in the choice of the generator matrix $\bfG_\lin$ that is used to construct the fine lattice $\Lf$.

We now discuss some intermediate lemmas which will be used to prove Theorem~\ref{thm:nestedconstructiona_listsize}. The formal proofs will be given in Sec.~\ref{sec:prf_constructiona_list}.

Fix any $\vy\in\bR^n$.
Fundamental to the proof is counting the number of lattice points within a ball of radius $r$ around $\vy$.
We will need bounds on $\card{\frac{1}{q}\Lc\cap \cB(\vy,r)}$. We can write it as $\card{\curbrkt{ \vx\in\bZ^n: \normtwo{ \vy-\frac{1}{q}\bfG_{\Lc}\vx }\leq r }}$.
A simple argument generalizing~\cite[Lemma 1]{ordentlich2016simple} can be used to show that this is upper (resp. lower) bounded by a the volume ratio between the ball (whose radius is lengthened (resp. shortened) by the covering radius of $q^{-1}\Lc $) and the fundamental Voronoi region of $q^{-1}\Lc$.
This can be formally stated as follows:
\begin{lemma}
	Let $V_n$ denote the volume of the unit ball in $\bR^n$, and $\Lc$ be a full-rank lattice in $\bR^n$. Then, for any $r>\rcov(\Lc)/q=\rcov(q^{-1}\Lc)$ and $\vy\in\bR^n$, we have
	\[
	\frac{q^n V_n}{\vol(\cV(\Lc))}\left( r-\frac{\rcov(\Lc)}{q} \right)^n\leq \left|\frac{1}{q}\Lc\cap \cB(\vy,r)\right| \leq \frac{q^n V_n}{\vol(\cV(\Lc))}\left( r+\frac{\rcov(\Lc)}{q} \right)^n.
	\]
	\label{lemma:count_latticepoints}
\end{lemma}
Observe that there is a bijection between $ \bFq^\kappa $ and $ \Lf\cap\cV(\Lc) $.
The encoder maps $m\in \bFq^\kappa$ to a nested lattice codeword (with slight abuse of notation\footnote{Here, we use the natural embedding of $\bFq$ in $\bZ$,  $\Phi$, to identify elements in $\bfG_\lin$ with the corresponding values in $\bZ$. To be rigorous, we should have written $ \Phi(\bfG_\lin) $ instead of $\bfG_\lin $.}) 
\begin{align}
\psi(m)\coloneqq \left[ \frac{1}{q}\bfG_{\Lc}\paren{[\bfG_\lin m]\bmod(q\bZ^n)} \right]\bmod\Lc, \notag
\end{align}
where all arithmetics are over $\bR$.

\begin{lemma}
	Fix $m\in\bFq^\kappa\backslash\{ {0}\}$ and $\vy\in\bR^n$. We have
	\begin{equation}
	\Pr[\psi(m)\in \cB(\vy,r)] \leq \left(\frac{r}{\sqrt{nP}} 2^{\delta/8}\left( 1+\frac{\sqrt{nP}}{qr} \right)\right)^n. \label{eq:prf1}
	\end{equation}
\end{lemma}
\begin{proof}
Since $ \bfG_\lin\in\curbrkt{0,1,\cdots,q-1}^{n\times\kappa} $ is a uniformly random matrix, $ [\bfG_\lin m]\bmod(q\bZ^n) $ is uniformly distributed in $ \curbrkt{0,1,\cdots,q-1}^n $. 
Consequently $ \psi(m) $ is uniformly distributed in $ q^{-1}\Lc\cap\cV(\Lc) $.
	Since the codeword $\psi(m)$ is guaranteed to be in $\cV(\Lc)$, we have
	\begin{align}
	\Pr[\psi(m)\in \cB(\vy,r)] &= \Pr\left[ \psi(m)\in [\cB(\vy,r)]\bmod\Lc \right]&\notag\\
	&= \frac{1}{q^n}\left| \frac{1}{q}\Lc\cap [\cB(\vy,r)]\bmod\Lc \right|&\notag\\
	&= \frac{1}{q^n}\left| \frac{1}{q}\Lc\cap \cB(\vy,r) \right|&\notag\\
	&\leq \frac{V_n}{\vol(\cV(\Lc))} \left( r+\frac{\rcov(\Lc)}{q} \right)^n&\notag
	\end{align}
	using Lemma~\ref{lemma:count_latticepoints}. Simplifying this, we get
	\begin{align}
	\Pr[\psi(m)\in \cB(\vy,r)] &\leq \frac{r^n}{(\reff(\Lc))^n}  \left( 1+\frac{\rcov(\Lc)}{qr} \right)^n &\notag \\
	& \leq \frac{r^n}{(\rcov(\Lc))^n} 2^{n\delta/8}\left( 1+\frac{\rcov(\Lc)}{qr} \right)^n &\notag \\
	& = \left(\frac{r}{\sqrt{nP}} 2^{\delta/8}\left( 1+\frac{\sqrt{nP}}{qr} \right)\right)^n, &\notag
	\end{align}
	where we have used Eqn.~\eqref{eq:rcov_condition} in the second step.
\end{proof}

\subsection{Proof of Theorem~\ref{thm:nestedlattice_listdecoding}}\label{sec:prf_constructiona_list}
If $ m_1,\ldots, m_\ell$ are linearly independent vectors in $\bFq^\kappa$ and $ \bfG_\lin $ is uniform, then $\psi( m_1),\ldots,\psi( m_\ell)$ are statistically independent and hence,
\begin{equation}
\Pr[\psi( m_1),\ldots,\psi( m_\ell)\in \cB(\vy,r)] = \left( \Pr[\psi(m)\in \cB(\vy,r)] \right)^\ell.
\label{eq:prf2}
\end{equation}
Every set of $L+1$ distinct vectors $ m_1,\ldots m_{L+1}$ in $\bFq^\kappa$ contains a subset of $\ell\coloneqq\log_q(L+1)$ linearly independent vectors.
\begin{align}
&\Pr[\psi( m_1),\ldots,\psi( m_{L+1})\in \cB(\vy,r)\text{ for some distinct } m_1,\ldots m_{L+1}]&\notag\\
\leq& \Pr[\psi( m_1),\ldots,\psi( m_{\ell})\in \cB(\vy,r)\text{ for some linearly independent } m_1,\ldots m_{L+1}]&\notag\\
\leq& \begin{pmatrix}2^{nR}\\ \ell\end{pmatrix} \Pr[\psi( m_1),\ldots,\psi( m_\ell)\in \cB(\vy,r)] &\label{eq:prf3}\\
\leq& 2^{nR\ell} \Pr[\psi( m_1),\ldots,\psi( m_\ell)\in \cB(\vy,r)],&\notag
\end{align}
where in Eqn.~\eqref{eq:prf3}, $ m_1,\ldots, m_\ell$ is a fixed (but arbitrary) set of linearly independent vectors in $ \bFq^\kappa $. Using Eqn.~\eqref{eq:prf1} and~\eqref{eq:prf2} in the above, we get
\begin{align}
&\Pr[\psi( m_1),\ldots,\psi( m_{L+1})\in \cB(\vy,r)\text{ for some distinct } m_1,\ldots m_{L+1}]&\notag\\
\leq& 2^{nR\ell }\left(\frac{r}{\sqrt{nP}} 2^{\delta/8}\left( 1+\frac{\sqrt{nP}}{qr} \right)\right)^{n\ell},&\notag
\end{align}
and hence,
\begin{align}
&\frac{1}{n}\log_2\Pr[\psi( m_1),\ldots,\psi( m_{L+1})\in \cB(\vy,r)\text{ for some distinct } m_1,\ldots m_{L+1}]&\notag\\
\leq&   \ell \left(R -\log_2\left(\frac{\sqrt{nP}}{r}\right) +\frac{\delta}{8}+\log_2\left( 1+\frac{\sqrt{nP}}{qr} \right)\right).&\label{eq:prf4}
\end{align}
This suggests that if $R$ and $r$ are not too large, then for any fixed but arbitrary $\vy$, the probability that there are more than $L$ lattice points within distance $r$ of $\vy$ is small. 
We want to show that this happens for every $\vy\in\bR^n$. First, observe that if $\vy\notin \cV(\Lc)+\cB(0,\sqrt{nN})$, then all codewords are at least $\sqrt{nN}$-away from $\vy$.
Therefore, it is enough to consider only those $\vy$ in $\cV(\Lc)+\cB(0,\sqrt{nN})$. A second observation is that if (for a positive integer $\alpha$) $Q(\vy)$ denotes the closest point in $\frac{1}{\alpha}\Lc$ to $\vy$,  
then 
\begin{align}
&\Pr[\psi( m_1),\ldots,\psi( m_{L+1})\in \cB(\vy,\sqrt{nN})\text{ for some distinct } m_1,\ldots m_{L+1}] &\notag\\
\leq& \Pr\left[\psi( m_1),\ldots,\psi( m_{L+1})\in \cB\left(Q(\vy),\sqrt{nN}+\frac{\rcov(\Lc)}{\alpha}\right)\text{ for some distinct } m_1,\ldots m_{L+1}\right].&\notag
\end{align}
The idea here is to quantize the $\vy$'s using $\frac{1}{\alpha}\Lc$ and then use a union bound. We want to make sure that $\alpha$ is sufficiently large, but not too large.
Specifically, $\alpha$ is the smallest integer greater than $\sqrt{P/N}/(2^{\delta/8}-1)$. Therefore, $\alpha$ satisfies
\begin{equation}
1+\frac{1}{\alpha}\sqrt{\frac{P}{N}} < 2^{\delta/8},
\label{eq:alphacondition_1}
\end{equation}
and 
\begin{equation}
\alpha < \frac{\sqrt{P/N}}{(2^{\delta/8}-1)} + 2.
\label{eq:alphacondition_2}
\end{equation}
Note that $ \alpha = \Theta(1/\delta) $.

Letting $r=\sqrt{nN}+\frac{\rcov(\Lc)}{\alpha} = \sqrt{nN}+\frac{\sqrt{nP}}{\alpha} $, we have
\begin{align}
&\Pr[\psi( m_1),\ldots,\psi( m_{L+1})\in \cB(\vy,\sqrt{nN})\text{ for some distinct } m_1,\ldots m_{L+1} \text{ and } \vy\in\bR^n] &\notag\\
=&\Pr[\psi( m_1),\ldots,\psi( m_{L+1})\in \cB(\vy,\sqrt{nN})\text{ for some distinct } m_1,\ldots m_{L+1} \text{ and } \vy\in\cV(\Lc)+\cB(0,\sqrt{nN})] &\notag\\
\leq& \Pr\left[\psi( m_1),\ldots,\psi( m_{L+1})\in \cB(\vy,r)\text{ for some distinct } m_1,\ldots m_{L+1} \text{ and } \vy\in\frac{1}{\alpha}\Lc\cap (\cV(\Lc)+\cB(0,\sqrt{nN}))\right]. &\notag
\end{align}
From Eqn.~\eqref{eq:rcov_condition}, \eqref{eq:rpack_condition} and the fact that $P>N$, we have $\sqrt{nN}\leq \sqrt{nP}=\rcov(\Lc) \le2^{\delta/8}\reff(\Lc) \leq 4\cdot2^{\delta/8}\rpack(\Lc)\le4\cdot2^{0.9/8}\rpack(\Lc)<4.4\rpack(\Lc) $. Therefore, $\cB(0,\sqrt{nN})\subset4.4\cB(0,\reff(\Lc))\subset 4.4\cV(\Lc)$.
We can therefore take a union bound over $\frac{1}{\alpha}\Lc\cap(5.4\cV(\Lc))$ which gives us\footnote{We would like to remark that we have not optimized these constants, and the bounds obtained may be loose. However, this would not change the overall result.}
\begin{align}
&\Pr[\Lf\cap\cV(\Lc)\text{ is not }(P,N,L)\text{-list decodable}] &\notag \\
=&\Pr[\psi( m_1),\ldots,\psi( m_{L+1})\in \cB(\vy,\sqrt{nN})\text{ for some distinct } m_1,\ldots m_{L+1} \text{ and } \vy\in\bR^n] &\notag\\
\leq& \Pr\left[\psi( m_1),\ldots,\psi( m_{L+1})\in \cB(\vy,r)\text{ for some distinct } m_1,\ldots m_{L+1} \text{ and } \vy\in\frac{1}{\alpha}\Lc\cap (5.4\cV(\Lc))\right]. &\notag
\end{align}
Using Eqn.~\eqref{eq:prf4} and applying the union bound over $\vy$'s, we have
\begin{align}
&\frac{1}{n}\log_2\Pr[\Lf\cap\cV(\Lc)\text{ is not }(P,N,L)\text{-list decodable}] &\notag \\
\leq& \frac{1}{n}\log_2\left| \frac{1}{\alpha}\Lc\cap (5.4\cV(\Lc))\right|  +   \ell \left(R -\log_2\left(\frac{\sqrt{nP}}{r}\right) +\frac{\delta}{8}+\log_2\left( 1+\frac{\sqrt{nP}}{qr} \right)\right).&\notag
\end{align}
Since $\left| \frac{1}{\alpha}\Lc\cap (5.4\cV(\Lc))\right|=(5.4\alpha)^n$ and $r\ge\sqrt{nN}$, we get
\begin{align}
&\frac{1}{n}\log_2\Pr[\Lf\cap\cV(\Lc)\text{ is not }(P,N,L)\text{-list decodable}] &\notag \\
\leq& \log_2(5.4\alpha)  +\ell \left(R-\log_2\left(\frac{\sqrt{nP}}{\sqrt{nN}+\sqrt{nP}/\alpha}\right) +\frac{\delta}{8}+\log_2\left( 1+\frac{\sqrt{P}}{q\sqrt{N}} \right)\right)&\notag\\
\leq& \log_2(5.4\alpha)  +\ell \left(R-\frac{1}{2}\log_2\left( \frac{P}{N} \right)+\log_2\left(1+\frac{\sqrt{P}}{\alpha\sqrt{N}}\right) +\frac{\delta}{8}+\log_2\left( 1+\frac{\sqrt{P}}{q\sqrt{N}} \right)\right)&\notag\\
\leq& \log_2(5.4\alpha) + \ell \left( -\delta +\log_2\left(1+\frac{\sqrt{P}}{\alpha\sqrt{N}}\right) +\frac{\delta}{8}+\log_2\left( 1+\frac{\sqrt{P}}{q\sqrt{N}} \right)\right).&\notag
\end{align}
Using Eqn.~\eqref{eq:alphacondition_1} and~\eqref{eq:prime_condition}, 
\begin{align}
&\frac{1}{n}\log_2\Pr[\Lf\cap\cV(\Lc)\text{ is not }(P,N,L)\text{-list decodable}] &\notag \\
\leq& \log_2(5.4\alpha) + \ell \left( -\delta +\frac{\delta}{8} +\frac{\delta}{8}+\frac{\delta}{8} \right)&\notag\\
=& \log_2(5.4\alpha) - 5\ell\frac{\delta}{8}. &\label{eq:prf5}
\end{align}
From Eqn.~\eqref{eq:prf5}, we can say that if $\ell>c_1\log_2(\alpha)/\delta$, the probability that a random lattice code is not list decodable goes to zero exponentially in $n$.
For $0<\delta<0.9$, there exist positive constants  $c_2,c_3$ (that could depend on $P,N$ but not on $\delta$) so that $c_2\delta<2^{\delta/8}-1<c_3\delta$ and using Eqn.~\eqref{eq:alphacondition_2}, we can see that $\log_2(\alpha)\leq c_4\log_2\frac{1}{\delta}$
for some positive $c_4$. Likewise, using Eqn.~\eqref{eq:primecondition_2}, we can show that there exist $c_5,c_6$ such that $ c_5\log_2\frac{1}{\delta}\leq \log_2q\leq c_6\log_2\frac{1}{\delta}$ 
for $\delta\in (0,0.9)$. This implies that we can choose $L=q^{\ell}-1=2^{\ell\log_2q}-1$ to be less than $2^{\frac{c}{\delta}\log_2^2\frac{1}{\delta}}$ for a sufficiently large constant $c$, so that
the probability that a random nested lattice code is not $(P,N,L)$-list decodable goes to zero as $2^{-\Omega(n)}$. This concludes the proof of Theorem~\ref{thm:nestedlattice_listdecoding}. \qed

\section{List decodability of infinite lattices}\label{sec:ic}

We now direct our attention to infinite constellations. Recall that an infinite constellation $ \cC $ is a countably infinite subset of $ \bR^n $.
We start by looking at infinite lattices.

\subsection{List size upper bound for infinite lattices}~\label{sec:listdecoding_infiniteconstA}
We claim that list decodability of nested Construction-A lattice codes implies a list decoding result for infinite Construction-A lattices. 
\begin{lemma}
	Let $ (\Lc,\Lf) $ be a pair of nested lattices with $ \Lc\subset\Lf $. Suppose that the nested lattice code $ \Lf\cap\cV(\Lc) $  is $ (P,N,L) $-list decodable. Then, the infinite lattice $ \Lf $ is $ (N,L) $-list decodable.
	\label{lemma:nestedld_implies_ldlattice}
\end{lemma}
\begin{proof}
	The infinite lattice $ \Lf $ is $ (N,L) $-list decodable if for every $ \vy\in\bR^n $, we have $ \card{\Lf\cap\cB(\vy,\sqrt{nN})}\leq L $. Due to the periodic structure of $ \Lf $, we have the property that $\{\vx+\cV(\Lf):\vx\in\Lf\} $ forms a partition of $ \bR^n $. Therefore, it suffices to show that  $ \card{\Lf\cap\cB(\vy,\sqrt{nN})}\leq L $ for all $ \vy\in\cV(\Lf) $ in order to prove $ (N,L) $-list decodability of $ \Lf $. 
	
	But we already have this from the list decodability of $ \Lf\cap\cV(\Lc) $. Since $ \cV(\Lf)\subset\cV(\Lc) $, $(P,N,L)$-list decodability of $ \Lf\cap\cV(\Lc) $ implies $ (N,L) $-list decodability of $ \Lf $.
\end{proof}

Using Lemma~\ref{lemma:nestedld_implies_ldlattice} and Theorem~\ref{thm:nestedconstructiona_listsize}, we have
\begin{theorem}
	For any constant $ 0<\delta<0.9 $, let $ \Lf $ be a random Construction-A lattice drawn from the ensemble of Sec.~\ref{sec:nestedlatticeensemble} with $ \Lc $ having covering radius $ 2\sqrt{nN} $ and $ \rcov(\Lc),\rpack(\Lc),q,\kappa $ satisfying~\eqref{eq:rcov_condition}, \eqref{eq:rpack_condition}, \eqref{eq:prime_condition}, \eqref{eqn:cond_kappa}, respectively. 
	Then, 
	\begin{enumerate}
		\item the normalized logarithmic density of $ \Lf $ is $ R(\Lf) \ge \frac{1}{2}\log\frac{1}{2\pi eN} - \delta $; 
		\item 
		there exists a constant $ c>0 $ independent of $ n,\delta $ such that 
		\[
		\Pr[\Lf \text{ is not }(N,2^{c\frac{1}{\delta}\log^2\frac{1}{\delta}})\text{-list decodable}] = o(1).
		\]
	\end{enumerate}
\end{theorem}
\begin{proof}
	Fix $ P=4N$. Since $ \Lc $ is good for covering, we have $~\reff(\Lc)\geq \rcov(\Lc)2^{-\delta/8}=2^{1-\delta/8}\sqrt{nN} $.
	We know that with high probability $ \Lf\cap\cV(\Lc) $ is $ (P,N,2^{\cO(\frac{1}{\delta}\log^2\frac{1}{\delta})}) $-list decodable, where the implied constant can only depend on $ N $. From Lemma~\ref{lemma:nestedld_implies_ldlattice}, we know that $ \Lf $ is also $ (N,2^{\cO(\frac{1}{\delta}\log^2\frac{1}{\delta})}) $-list decodable with high probability. To complete the proof, it suffices to compute the NLD of $ \Lf $.
	
	The rate of the nested lattice code 
	\begin{align*}
	R(\Lf\cap\cV(\Lc))=& \frac{1}{2}\log\frac{P}{N} - \delta= \frac{1}{2}\log\frac{4N}{N} - \delta = 1-\delta. \notag 
	\end{align*}
	The NLD of the infinite lattice can hence be bounded as follows,
	\begin{align}
	R(\Lf) =& \frac{1}{n}\log\frac{2^{nR}}{\card{\cV(\Lc)}} \notag \\
	=& \frac{1}{n}\log\frac{2^{n(1-\delta)}}{\card{\cB^n\paren{0, \reff(\Lc)}}}  \notag \\
	=& \log\frac{2^{1-\delta}}{V_n^{1/n}\reff(\Lc)} \notag \\
	\asymp& \log\frac{2^{1-\delta}}{\sqrt{\pi n}^{-1/n}\sqrt{2\pi e/n}\reff(\Lc)} \notag \\
	\ge& \log\frac{2^{1-\delta}}{\sqrt{2\pi e/n}\sqrt{4nN}} + o(1) \notag \\
	=& \frac{1}{2}\log\frac{1}{2\pi eN} - \delta+o(1). \notag 
	\end{align}
	This completes the proof.
\end{proof}

\subsection{Remark}
We proved the above theorem for the random infinite lattice $ \bfG_{\Lc}\Lf'(\cC_\lin) $, where $ \Lf'(\cC_\lin) = \Phi(\cC_\lin)+q\bZ^n $ is the ``standard'' Construction-A lattice obtained from a random linear code $ \cC_\lin $, and $ \bfG_{\Lc} $ is a generator matrix of the coarse lattice. We could have instead proved a similar list decoding result for $ \Lf'(\cC_\lin) $ by following the same approach as in Sec.~\ref{sec:ld_constr_a}, but instead taking a union bound on $ \vy $'s within $ [0,q)^n $. Doing so would also give a list size of $ 2^{\cO(\frac{1}{\delta}\log^2\frac{1}{\delta})} $ for all NLD satisfying $ R(\Lf'(\cC_\lin))\le\frac{1}{2}\log\frac{1}{2\pi eN} - \delta $.

Similarly, for lattice codes, via essentially the same arguments, it can be shown that nested random Construction-A lattice codes ${\frac{1}{\alpha}\Lambda'(\cC_\lin)\cap\Lambda'(\cC_\lin)}$  and random Construction-A lattices with ball shaping $\frac{1}{\alpha}\Lambda'(\cC_\lin)\cap\cB(0,\sqrt{nP})$ for proper scaling $1/\alpha$ so as to achieve rate $\frac{1}{2}\log\frac{P}{N}-\delta$ are also $(P,N,2^{\cO\paren{\frac{1}{\delta}\log^2\frac{1}{\delta}}})$-list decodable whp.

\section{List decodability of regular infinite constellations}~\label{sec:regular_ic}

Having established a list decoding result for infinite lattices, we now turn to the problem of determining optimal list sizes for infinite constellations. Do there exist ICs $ \cC $ for which the list size is at most $ \cO(\poly(1/\delta)) $ for all $ R(\cC) \le \frac{1}{2}\log\frac{1}{2\pi eN}-\delta $?


To study this, we define an ensemble of periodic infinite constellations. 
We call this a $ (\Lc, q, M) $ infinite constellation (IC) which is defined as follows. 
Let $ \Lc $ be a (full rank) lattice satisfying Eqn.~\eqref{eq:rcov_condition} and \eqref{eq:rpack_condition}. 
Let $ q $ be a prime. 
An $(\Lc,q,M)$ random IC $ \cC $ is obtained by selecting $ M $ points $ \cC'=\{\vbfx_1,\ldots,\vbfx_M\} $ independently and uniformly at random from $ \cV(\Lc)\cap\frac{1}{q}\Lc $ and then tiling. Therefore, $ \cC = \cC'+\Lc $.
See Fig.~\ref{fig:ic-latticeshaping} for a pictorial illustration of the construction of such an IC ensemble.

\begin{figure}
	\begin{center}
		\includegraphics[width=9cm]{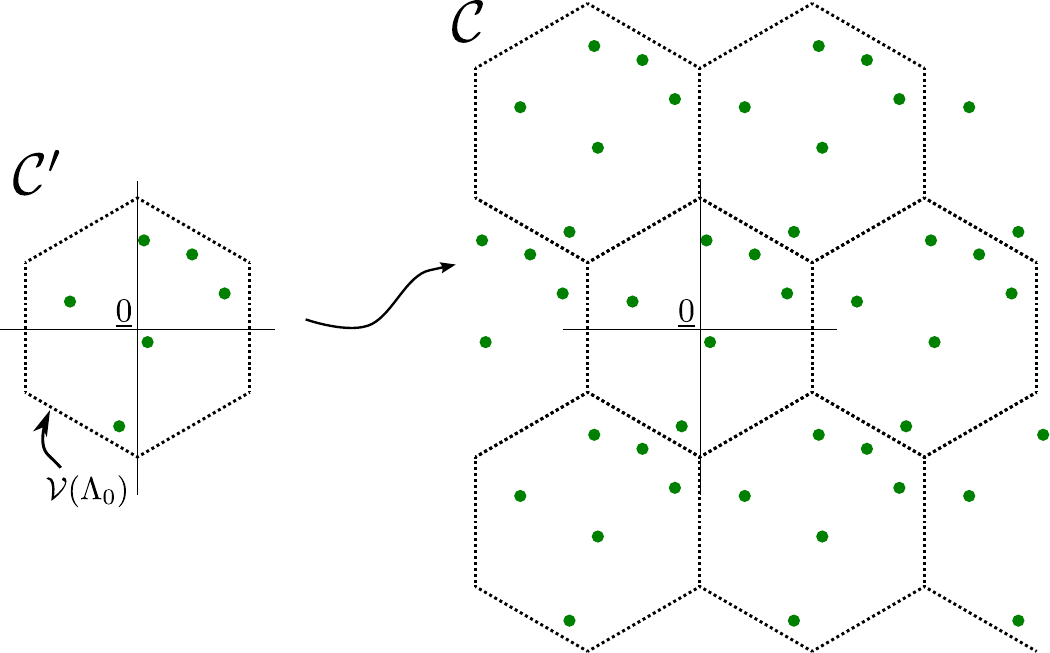}
		\caption{Illustration of the class of infinite constellations studied in Sec.~\ref{sec:regular_ic}. The finite set $\cC'$ is tessellated using $\Lambda_0$ to obtain $\cC=\cC'+\Lambda_0$}
		\label{fig:ic-latticeshaping}
	\end{center}
\end{figure}

The reason why we introduce this new class of ICs is because it is very simple to work with. We can very easily prove several nice properties that would be otherwise complicated for random lattices. We feel that this is a natural counterpart of uniformly random codes over finite fields. Moreover, we can obtain a code for the power-constrained channel by taking the intersection of the IC with $ \cB(0,\sqrt{nP}) $.

\begin{remark}
In fact, one can define a hierarchy of more and more ``uniform'' ICs in a similar manner. 
Under the same construction $ \cC = \cC' + \Lc $ as above, we can choose $\cC'$ such that
\begin{enumerate}
	\item each point in $\cC'$ is independent and uniformly distributed in $ \cV(\Lc) $;
	\item each point in $\cC'$ is independent and uniformly distributed in $ \cV(\Lc)\cap\frac{1}{q}\Lc $; (This choice is the same as that in the above paragraph.)
	\item $\cC'$ is a random subset of $ \cV(\Lc)\cap\frac{1}{q}\Lc $ that forms a group under addition modulo $ \Lc $. 
\end{enumerate}
The above three constructions are decreasingly ``uniform''. 
In particular, the last construction of $ \cC $ forms a lattice. 

We will study the list decodability property of the second construction. 
The same quantitative results in this section hold for the first construction as well since the latter is more uniform. 
In Appendix~\ref{sec:ic_goodness}, we study other goodness properties of the first construction. 
Properties are easiest to prove under this construction. 
\end{remark}

Before presenting our formal results, we need one more definition. 
The effective radius of an infinite constellation is defined as the radius of the $n$-dimensional ball having volume equal to $ 1/\Delta(\cC) $, i.e., 
\[~\reff(\cC) \coloneqq \left(\frac{1}{V_n\Delta(\cC)}\right)^{1/n} \]
where $ V_n $ denotes the volume of a unit $n$-ball.

The following is a simple application of Chernoff bound, and can be used to tightly concentrate the rate of a power-constrained code.
\begin{lemma}
	Let $ \alpha>2\sqrt{nP} $. Let $ \cC $ be an $ (\alpha, M) $ random IC, where $ M $ is chosen so that $~\reff(\cC)<\sqrt{nP} $. Then, for every $ \delta>0 $,
	\[
	\Pr\left[|\cC\cap \cB(0,\sqrt{nP})|<\frac{V_n\sqrt{nP}^nM}{\alpha^n}(1-\delta)\right] < 2^{-\Omega(\delta^2 2^n)}. 
	\]
	\label{lemma:randomic_ball}
\end{lemma}

\begin{proof}
The lemma follows from 
\begin{align}
\expt{\card{\cC\cap\cB(0, \sqrt{nP})}} =& \frac{\vol(\cB(0,\sqrt{nP}))}{\vol(\cA)}M = \frac{V_n\sqrt{nP}^n}{\alpha^n}M, \label{eqn:code_size} 
\end{align}
and the standard Chernoff bound (Lemma~\ref{lem:chernoff}). 
To ensure that the bound is decaying in $n$, we had better make $ \expt{\card{\cC\cap\cB(0, \sqrt{nP})}} $ exponentially large. 
By the relations in Eqn.~\eqref{eqn:param_ic}, the expression \eqref{eqn:code_size} can also be written as $ \paren{\frac{\sqrt{nP}}{\reff(\cC)}}^n $. 
The condition $ \reff(\cC)<\sqrt{nP} $ hence guarantees that $ \expt{\card{\cC\cap\cB(0, \sqrt{nP})}} = 2^{\Theta(n)} $. 
\end{proof}
\begin{remark}
The condition $ \alpha>2\sqrt{nP} $ is to ensure $ \cB(0,\sqrt{nP})\subseteq\cA $ which in turn guarantees that every codeword in $ \cB(0,\sqrt{nP}) $ is independent. 
\end{remark}

\subsection{List size upper bound}

We now study list decodability properties of $ (\Lc,q, M) $ random ICs.

An infinite constellation $\cC$ is $ (N,L) $-list decodable if for every $ \vy\in\bR^n $, we have $ |\cC\cap\cB(\vy,\sqrt{nN})|\leq L $. 


\begin{proposition}
\label{prop:ic_listupperbound}
Fix a small $ \delta>0 $. Let $ C\coloneqq\frac{1}{2}\log\frac{1}{2\pi eN} $. For any $ N>0$, there is an ensemble of random $ (\Lc,q, M) $ ICs with $ M = 2^{n(C-\delta)}|\cV(\Lc)| $ chosen so as to satisfy $~R(\cC)=C-\delta $ is $ (N,\cO(\frac{1}{\delta}\log\frac{1}{\delta})) $-list decodable with probability at least $ 1-2^{-\Theta(n)} $.
\end{proposition}

\begin{proof}
We choose $ \Lc $ such that $ \rpack(\Lc) = 2\sqrt{nN} $. 
Let $ q $ be the smallest integer such that 
\begin{align}
\log\paren{1+\frac{\rcov(\Lc)}{q\sqrt{nN}}} &< 2^{\delta/8}. \label{eqn:prime-condition-ic}
\end{align}
This implies the following lower and upper bounds on $ q $:
\begin{align}
q &\ge \frac{\rcov(\Lc)/\sqrt{nN}}{2^{\delta/8} - 1} 
\ge \frac{\rpack(\Lc)/\sqrt{nN}}{2^{\delta/8} - 1} 
= \frac{2}{2^{\delta/8} - 1}; \label{eqn:prime-ub-ic} \\
q &\le \frac{\rcov(\Lc)/\sqrt{nN}}{2^{\delta/8} - 1}+2
\le \frac{2^{\delta/8}\cdot\reff(\Lc)/\sqrt{nN}}{2^{\delta/8} - 1}+2
\le \frac{2^{\delta/8}\cdot4\cdot\rpack(\Lc)/\sqrt{nN}}{2^{\delta/8} - 1}+2
= \frac{2^{3+\delta/8}}{2^{\delta/8} - 1}+2. \label{eqn:prime-lb-ic} 
\end{align}
Using the elementary inequality 
\begin{align}
\frac{11}{\delta} &< \frac{1}{2^{\delta/8} - 1} < \frac{12}{\delta}, \label{eqn:elementary-ineq}
\end{align}
we get the following looser bounds:
\begin{align}
q &\ge \frac{22}{\delta}, \quad
q \le 2^{4+1/8}\cdot\frac{12}{\delta} + \frac{2}{\delta} \le \frac{107}{\delta}. \label{eqn:prime-ic-loose}
\end{align}
Fix a $ \delta>0 $, and choose $ M $ such that $R(\cC) = C-\delta $. 
We will show that such random $ (\Lc,q,M) $ ICs are list decodable with constant list sizes. 
The proof is quite standard so we only give a brief outline. 
Using $ q$, we define a net for $ \cV(\Lc) + \cB(0,\sqrt{nN}) $ as follows: $ \cY \coloneqq \frac{1}{ q}\Lc\cap(\cV(\Lc)+\cB(0,\sqrt{nN})) $. 
By the same argument as in the proof of Theorem~\ref{thm:nestedlattice_listdecoding}, we have
\begin{align}
\cY\subset\frac{1}{ q}\Lc\cap(5.4\cV(\Lc)), \notag
\end{align}
and $ |\cY| \le (5.4 q)^n $. 
Let $ r\coloneqq\sqrt{nN} +  q^{-1}\rcov(\Lc) $. 
For a set $ \cA\subset\bR^n $, define the shorthand notation $\cA^*\coloneqq[\cA]\bmod\Lc$.
For any $ \vbfx\in\cC' $, by Lemma~\ref{lemma:count_latticepoints},
\begin{align}
\prob{\vbfx\in\cB^*(\vy,r)} &\le \paren{\frac{r}{\rcov(\Lc)}2^{\delta/8}\paren{1+\frac{\rcov(\Lc)}{qr}}}^n. \notag
\end{align}
For any $ \vy $, since each point in $\cC'$ is independent, for any $ 1\le i_1<\cdots<i_{L+1}\le M $, 
\begin{align}
\prob{\forall j\in[L+1],\;\vbfx_{i_j}\in\cB^*(\vy,r)}
&= \prod_{j\in[L]}\prob{\vbfx_{i_j}\in\cB^*(\vy,r)} 
\le \paren{\frac{r}{\rcov(\Lc)}2^{\delta/8}\paren{1+\frac{\rcov(\Lc)}{qr}}}^{n(L+1)}. \notag
\end{align}
By union bound, 
\begin{align}
\prob{\exists L+1\text{ codewords in }\cB^*(\vy,r)}
&\le \binom{M}{L+1} \paren{\frac{r}{\rcov(\Lc)}2^{\delta/8}\paren{1+\frac{\rcov(\Lc)}{qr}}}^{n(L+1)}. \notag 
\end{align}
Finally, 
\begin{align}
\prob{\text{The IC is not }(N,L)\text{ list decodable}} 
&\le \prob{\exists \vy\in\cV(\Lc) + \cB(0,\sqrt{nN}),\;\exists L+1\text{ distinct codewords in }\cB(\vy,\sqrt{nN})} \notag \\
&\le \prob{\exists \vy\in\cY,\;\exists L+1\text{ distinct codewords in }\cB(\vy,\sqrt{nN} + \rcov( q^{-1}\Lc))} \notag \\
&\le \sum_{\vy\in\cY} \prob{\exists L+1\text{ codewords in }\cB^*(\vy,\sqrt{nN} +  q^{-1}\rcov(\Lc))} \notag \\
&\le |\cY| \binom{M}{L+1} \paren{\frac{r}{\rcov(\Lc)}2^{\delta/8}\paren{1+\frac{\rcov(\Lc)}{qr}}}^{n(L+1)}. \notag 
\end{align}
We then work with the exponent.
\begin{align}
& \frac{1}{n}\log\prob{\text{The IC is not }(N,L)\text{ list decodable}} \notag \\
&\le \log(5.4 q) + (L+1)\sqrbrkt{R + \frac{1}{n}\log\paren{V_n\reff(\Lc)^n} + \log\frac{r}{\rcov(\Lc)} + \frac{\delta}{8} + \log\paren{1+\frac{\rcov(\Lc)}{qr}}}. \notag 
\end{align}
We note that
\begin{align}
\frac{1}{n}\log\paren{V_n\reff(\Lc)^n}
&\le \frac{1}{n}\log\paren{V_n\rcov(\Lc)^n} \notag \\
&= \frac{1}{n}\log V_n + \log\rcov(\Lc)\notag \\
&\stackrel{n\to\infty}{\asymp} \frac{1}{n}\log\sqrbrkt{\frac{1}{\sqrt{\pi n}}\paren{\frac{2\pi e}{n}}^{n/2}} + \log\rcov(\Lc)\notag \\
&\stackrel{n\to\infty}{\asymp} \frac{1}{2}\log\frac{2\pi e}{n} + \log\rcov(\Lc) \notag \\
&= \frac{1}{2}\log\paren{\frac{2\pi e}{n}\rcov(\Lc)^2}. \notag 
\end{align}
Therefore 
\begin{align}
& \frac{1}{n}\log\prob{\text{The IC is not }(N,L)\text{ list decodable}} \notag \\
&\le \log(5.4 q) + (L+1)\sqrbrkt{R + \frac{1}{2}\log\paren{\frac{2\pi e}{n}\rcov(\Lc)^2} + \log\frac{r}{\rcov(\Lc)} + \frac{\delta}{8} + \log\paren{1+\frac{\rcov(\Lc)}{qr}}} \notag \\
&= \log(5.4 q) + (L+1)\sqrbrkt{R + \frac{1}{2}\log\paren{\frac{2\pi er^2}{n}} + \frac{\delta}{8} + \log\paren{1+\frac{\rcov(\Lc)}{qr}}}. \notag 
\end{align}
By the choice of $r$, we know
\begin{align}
\frac{r^2}{n} &= \frac{1}{n}\paren{\sqrt{nN} + \frac{\rcov(\Lc)}{ q}}^2
= N\paren{1+\frac{\rcov(\Lc)}{ q\sqrt{nN}}}^2. \notag 
\end{align}
The exponent can be simplified to
\begin{align}
& \log(5.4 q) + (L+1)\sqrbrkt{R + \frac{1}{2}\log(2\pi eN) + \log\paren{1+\frac{\rcov(\Lc)}{ q\sqrt{nN}}} + \frac{\delta}{8} + \log\paren{1+\frac{\rcov(\Lc)}{q\sqrt{nN}}}} \notag \\
&\le \log(5.4 q) + (L+1)\paren{-\delta + \frac{\delta}{8} + \frac{\delta}{8} + \frac{\delta}{8}} \notag \\
&= \log(5.4 q) - \frac{5}{8}\delta(L+1). \notag 
\end{align}
Observe that the above exponent is negative if $ L>\frac{c_1}{\delta}\log q $ for some constant $ c_1>0 $. 
Using the condition on $ q$ (Eqn.~\eqref{eqn:prime-ic-loose}) and following similar calculations that appear at the end of the proof of Theorem~\ref{thm:nestedlattice_listdecoding}, we conclude that there exists a constant $c>0$ such that the probability that a random IC is not list decodable is exponentially small in $n$ if $ L>c\frac{1}{\delta}\log\frac{1}{\delta} $.
This completes the proof.
\end{proof}

\subsection{List decoding capacity theorem for infinite constellations}
\label{sec:list-dec-cap-ic}

We have shown that a random IC of NLD $ \frac{1}{2}\log\frac{1}{2\pi eN} - \delta $ as defined in the current section is $ \paren{N,\cO\paren{\frac{1}{\delta}\log\frac{1}{\delta}}} $-list decodable. 
In fact, there is a converse argument showing that \emph{any} IC of NLD $ \frac{1}{2}\log\frac{1}{2\pi eN} + \delta $ must not be $ \paren{N,2^{\cO(\delta n)}} $-list decodable. 
This is proved in Appendix~\ref{app:converse-ic}. 
Therefore, combining Proposition~\ref{prop:ic_listupperbound} and Proposition~\ref{prop:converse-ic}, we get the following list decoding capacity theorem for ICs. 
\begin{definition}[List decoding capacity]
	Fix any $ N>0 $.
	We say that $ C(N) $ is the \emph{list decoding capacity} if for every $ \delta>0 $, there exists a $ \gamma>0 $ such that $ C(N)-\delta $ is achievable for $ (N,\cO(n^{\gamma})) $-list decoding, and for every $ \delta>0 $, there exist no codes of rate  $ C(N)+\delta $ which are $ (N,2^{o(n)}) $-list decodable.
\end{definition}

\begin{theorem}
For any $ N>0 $, $ C(N) = \frac{1}{2}\log\frac{1}{2\pi eN} $. 
\end{theorem}

\subsection{List size lower bound}

\begin{lemma}
	Let $ \cC $ be an $(\alpha, M)$ random IC chosen so as to satisfy $ \alpha=4\sqrt{nN} $ and $~R(\cC) = C-\delta $ where $ C=\frac{1}{2}\log\frac{1}{2\pi eN} $. Then,
	\[
	\prob{\cC\text{ is }\paren{N,\cO\paren{\frac{1}{\delta}\log\frac{1}{\delta}}}\text{-list decodable}} = o(1).
	\] 
	\label{lemma:randomic_listlowerbound}
\end{lemma}

\begin{proof}

	The proof is almost identical to that of Proposition~\ref{prop:lbound_randomspherical}. We only highlight the main differences here. 
	

	Let $q$ be the smallest integer satisfying Eqn.~\eqref{eqn:prime-condition-ic} (recall that this implies Eqn.~\eqref{eqn:prime-ub-ic} and \eqref{eqn:prime-lb-ic}, or more loosely, Eqn.~\eqref{eqn:prime-ic-loose}).
	Define $ \cC \coloneqq\cC' + \frac{1}{q}\Lc $, where $ \Lc $ is simultaneously good for covering and packing, i.e., it satisfies both Eqn.~\eqref{eq:rcov_condition} and \eqref{eq:rpack_condition}; $ \cC'\subseteq\frac{1}{q}\Lc\cap\cV(\Lc) $ is a set of $M$ uniformly random and independent points $ \vbfx_1,\cdots,\vbfx_M $ in $ \frac{1}{q}\Lc\cap\cV(\Lc) $. 
	Scale $ \Lc $ such that $ \rpack(\Lc) = 2\sqrt{nN} $. 
	Let $ M = 2^{n(C - \delta)}|\cV(\Lc)| $ for some $ 0<\delta<0.9 $. 

	Let $ \cM \coloneqq [M] $ and $ \cY \coloneqq \frac{1}{q}\Lc\cap(\cV(\Lc)+\cB(0,\sqrt{nN})) $. 
	Define the random variable 
	\begin{align}
	W &\coloneqq \sum_{\vy\in\cY} \sum_{\{m_1,\cdots,m_L\}\in\binom{\cM}{L}} \indicator{\vbfx_{m_1},\cdots,\vbfx_{m_L}\in\cB^*(\vy,\sqrt{nN})}, 
	\label{eqn:w_ic}
	\end{align}
	where we use the following notation 
	$\cA^*\coloneqq\sqrbrkt{\cA}\bmod\Lc$ for any $ \cA\subset\bR^n $.
	

	We will upper bound
	\begin{align}
	\prob{\cC\text{ is }(N,L)\text{-list decodable}}=&\prob{\bigcap_{\vy\in\bR^n}\curbrkt{\card{\cC\cap\cB\paren{\vy,\sqrt{nN}}}\le L}}\notag\\
	\le&\prob{\bigcap_{\vy\in\cY}\curbrkt{\card{\cC'\cap\cB^*\paren{\vy,\sqrt{nN}}}\le L}},\notag\\
	\le&\var{W}/\expt{W}^2.\notag
	\end{align}

	\subsubsection{Lower bounding $\expt{W}$}

	It turns out that the expectation
	\begin{align}
	\expt{W} &= \sum_{\vy\in\cY}\sum_{\cL\in\binom{\cM}{L}}\prob{\vbfx_{\cL}\subset\cB^*\paren{\vy,\sqrt{nN}}},\label{eqn:expectation_ic}
	\end{align}
	where $\vbfx_\cL\coloneqq\curbrkt{\vbfx_m\colon m\in\cL}$, can be computed precisely.

	The probability in the summand of the right-hand side (RHS) of Eqn.~\eqref{eqn:expectation_ic} is 
	\begin{align}
	\mu^L &\coloneqq \paren{\frac{\card{\frac{1}{q}\Lc\cap\cB^*(\vy,\sqrt{nN})}}{\card{\frac{1}{q}\Lc\cap\cV(\Lc)}}}^L .
	\notag 
	\end{align}
	Overall the expectation in Eqn.~\eqref{eqn:expectation_ic} equals
	\begin{align}
	\expt{W} &= |\cY|\binom{M}{L}^L\mu^L. \notag 
	\end{align}
	
	
	\subsubsection{Upper bounding $\var{W}$}

	As in the proof of Proposition~\ref{prop:lbound_randomspherical}, $\cE_1,\cE_2$ and $\cE_3$ are similarly defined and their probabilities are similarly bounded. 
	\begin{align}
	\prob{\cE_1}&=\paren{\frac{\card{\cB\paren{\sqrt{nN}}\cap\frac{1}{q}\Lc}}{\card{\cY}}}^2\eqqcolon\eta^2,\notag\\
	\prob{\cE_2\cap\cE_3|\cE_1}&\le \mu^{L - 1}\mu^{L-\ell} = \mu^{2L-\ell-1}.\notag
	\end{align}

	Overall we have
	\begin{align}
	\var{W} &\le |\cY|^2 \sum_{\ell = 1}^LM^{2L-\ell} \eta^2\mu^{2L-\ell-1} 
	\le |\cY|^2LM^L\eta^2\mu^{L-1}. \notag 
	\end{align}
	
	
	
	\subsubsection{Wrapping things up}

	The probability that a random infinite constellation is list decodable is at most
	\begin{align}
	\frac{\var{W}}{\expt{W}^2} &\le 
	\frac{\card{\cY}^2LM^L\eta^2\mu^{L-1}}{\card{\cY}^2(M/L)^{2L}\mu^{2L}} 
	= L^{2L+1}M^{-L}\eta^2\mu^{-L-1}. \notag 
	\end{align}
	We shall upper bound $\eta$ and lower bound $\mu$. 

	For $\eta$, we have
	\begin{align}
	\eta &= \frac{\card{\cB(\sqrt{nN})\cap\frac{1}{ q}\Lc}}{\card{\cY}} 
	\le \frac{\card{\cB(\sqrt{nN})\cap\frac{1}{ q}\Lc}}{ q^n} 
	\le \sqrbrkt{\frac{\sqrt{nN}}{\rcov(\Lc)}2^{\delta/8}\paren{1+\frac{\rcov(\Lc)}{ q\sqrt{nN}}}}^n, \notag 
	\end{align}
	where the first inequality is because
	\begin{align}
	\cY\supset\frac{1}{ q}\Lc\cap\cV(\Lc) \implies |\cY|\ge  q^n. \notag
	\end{align}

	For $ \mu $, we have
	\begin{align}
	\mu &= \frac{\card{\frac{1}{q}\Lc\cap\cB^*(\vy,\sqrt{nN})}}{\card{\frac{1}{q}\Lc\cap\cV(\Lc)}}
	= \frac{\card{\frac{1}{q}\Lc\cap\cB(\vy,\sqrt{nN})}}{q^n}
	\ge \sqrbrkt{\frac{\sqrt{nN}}{\rcov(\Lc)}2^{\delta/8}\paren{1-\frac{\rcov(\Lc)}{q\sqrt{nN}}}}^n. \notag 
	\end{align}

	Therefore,
	\begin{align}
	\frac{1}{n}\log\frac{\var{W}}{\paren{\expt{W}}^2}
	&\le -L\sqrbrkt{R + \frac{1}{2}\log\paren{\frac{2\pi e}{n}\rcov(\Lc)^2}}
	+2\sqrbrkt{\log\frac{\sqrt{nN}}{\rcov(\Lc)} + \frac{\delta}{8} + \log\paren{1+\frac{\rcov(\Lc)}{ q\sqrt{nN}}}} \notag \\
	&\quad-(L+1)\sqrbrkt{\log\frac{\sqrt{nN}}{\rcov(\Lc)} + \frac{\delta}{8} + \log\paren{1-\frac{\rcov(\Lc)}{q\sqrt{nN}}}} \notag \\
	&= -L\sqrbrkt{R + \frac{1}{2}\log\paren{\frac{2\pi e}{n}\rcov(\Lc)^2} + \log\frac{\sqrt{nN}}{\rcov(\Lc)} + \frac{\delta}{8} + \log\paren{1-\frac{\rcov(\Lc)}{q\sqrt{nN}}}} \notag \\
	&\quad + 2\sqrbrkt{\log\frac{\sqrt{nN}}{\rcov(\Lc)} + \frac{\delta}{8} + \log\paren{1+\frac{\rcov(\Lc)}{ q\sqrt{nN}}}} 
	-\sqrbrkt{\log\frac{\sqrt{nN}}{\rcov(\Lc)} + \frac{\delta}{8} + \log\paren{1-\frac{\rcov(\Lc)}{q\sqrt{nN}}}} \notag \\
	&= -L\sqrbrkt{R + \frac{1}{2}\log(2\pi eN) + \frac{\delta}{8} + \log\paren{1-\frac{\rcov(\Lc)}{q\sqrt{nN}}}} 
	+ \log\frac{\sqrt{nN}}{\rcov(\Lc)} + \frac{\delta}{8} \notag \\
	&\quad + 2\log\paren{1+\frac{\rcov(\Lc)}{ q\sqrt{nN}}} - \log\paren{1-\frac{\rcov(\Lc)}{q\sqrt{nN}}}. \notag 
	\end{align}
	Recall that $q$ satisfies Eqn.~\eqref{eqn:prime-condition-ic} which implies
	\begin{align}
	\log\paren{1 - \frac{\rcov(\Lc)/\sqrt{nN}}{q}} 
	&> \log\sqrbrkt{1 - \frac{\rcov(\Lc)/\sqrt{nN}}{(\rcov(\Lc)/\sqrt{nN})/(2^{\delta/8} - 1)}}
	= \log(2-2^{\delta/8}) 
	\ge -\frac{\delta}{7}, \notag
	\end{align}
	where the last inequality is true for any $ \delta\in(0,1) $. 
	Using the bounds on $ q $ and $ q$, we get
	\begin{align}
	\frac{1}{n}\log\frac{\var{W}}{\paren{\expt{W}}^2}
	&\le -L\paren{-\delta + \frac{\delta}{8} - \frac{\delta}{7}} + \log\frac{\sqrt{nN}}{\rcov(\Lc)} + \frac{\delta}{8} + \frac{\delta}{4} + \frac{\delta}{7}. \notag 
	\end{align}
	Recall the relation $ \rcov(\Lc)\ge\rpack(\Lc) = 2\sqrt{nN} $. 
	Then
	\begin{align}
	\frac{1}{n}\log\frac{\var{W}}{\paren{\expt{W}}^2}
	&\le \frac{57}{56} \delta L-1+\frac{29}{56}\delta. \notag 
	\end{align}
	This exponent is negative if $ L<\frac{1-\frac{29}{56}\delta}{\frac{57}{56}\delta} $, or more loosely, $ L<\frac{9}{19\delta} $. 
\end{proof}

\subsection{Other goodness properties}
The random ICs defined in this section have other interesting geometric properties which are much harder to prove for lattices~\cite{erez2005lattices}, for instance, packing goodness, AWGN goodness and covering goodness. See Appendix~\ref{sec:ic_goodness} for statements and proofs.


\section{Haar measure on $\sL_n$}\label{sec:haar_meas_on_lattices}
Let us first ask ourselves: how do we define a \emph{random} lattice? To sample a random lattice from a certain ensemble, we need to define a distribution on the set of all lattices. As we know, a lattice is specified by its generator matrix and thus it suffices to define a distribution over matrices.\footnote{Strictly speaking, each lattice is \emph{non-uniquely} identified with a generator matrix. Those matrices giving rise to the same lattice should be quotiented out when one wants to define a distribution on lattices by defining it on matrices. See Sec.~\ref{sec:haarlattice_introduction} for more formalisms.} There are several ensembles of matrices that are extensively studied in the literature of random matrix theory. Such ensembles, including the Gaussian ensemble, the Bernoulli ensemble, etc.,~\cite{tao-2012-rand-mat-thy}  are mostly defined by sampling entries iid from  simple distributions. However, we believe that such ensembles will \emph{not} give rise to interesting lattices, in the sense that the resulting lattices are not likely to have nontrivial packing and covering efficiencies simultaneously. 

We give a heuristic argument to justify the above statement.
Suppose that we sample an $n$ by $n$ random matrix $\bfG$ over $\bR$ by sampling each entry iid according to $\cN(0,\sigma^2)$ for some fixed constant deviation $\sigma>0$. By the high-dimensional geometry of Gaussian random vectors, each column of $\bfG$ has $L^2$-norm highly concentrated around $\sqrt{n\sigma^2}$ and is approximately orthogonal to other columns. That is to say, the columns $\{\vbfg_1,\cdots,\vbfg_n\}$ of $\bfG$ are basically a mildly perturbed version of the standard orthonormal basis $\{\ve_1,\cdots,\ve_n\}$ scaled by $\sqrt{n\sigma^2}$ (and potentially also rotated, which does not affect most goodness properties of the resulting lattices we are interested in). The lattice $\bfG\bZ^n$ cannot be good for  packing whp since $\bZ^n$ (and also its scaling and rotation) has vanishing packing efficiency $\asymp\sqrt{\frac{\pi e}{2n}}$ as the dimension $n$ tends to infinity. Indeed, there is not much study on lattices resulting from those canonical matrix ensembles. One can find related results along this direction in~\cite{neelamani-dash-sanjeeb-baraniuk-2007-nearly-ortho-rand-lattice}. Not surprisingly, it boils down to understanding the singular value spectrum of $\bfG$.

In the case of finite fields, it is known that a \emph{uniformly random} linear code has good list decoding properties. It would therefore be a natural choice to study a random lattice drawn uniformly over the set of all lattices of a fixed determinant. 
Unfortunately, the space of such lattices is unbounded\footnote{The unboundedness (wrt $L^2 $-distance) of the set of determinant-1 lattices can be seen from the following example. The matrix $\begin{bmatrix}
K^{-\frac{1}{n-1}}&&&\\
&\ddots&&\\
&&K^{-\frac{1}{n-1}}&\\
&&&K
\end{bmatrix}_{n\times n} $ generates a determinant-1 lattice for every value of $K\ge0$. However, the $ L^2 $-norm of (the vectorization of) the matrix is $\sqrt{\paren{K^{-\frac{1}{n-1}}}^2(n-1)+K^2}\to\infty$ as $K$ approaches infinity.} 
and hence it does not make sense to talk about uniform distribution on it. 

In the following section, we introduce the  Haar distribution over lattices, and survey some of the important results pertaining to our discussions. For small enough subsets $ \cB $ of $ \bR^n $, we conjecture that the distribution of the number of lattice points (which is a random variable if the lattice is drawn according to the Haar distribution)  in $ \cB $ looks like a Poisson distribution. Expressions for the first $ o(n) $ moments of the number of lattice points have been derived in the literature. Encouraged by these results, we make a conjecture about the first $\cO(n)$ moments. We then show that if this conjecture is true, then Haar lattices achieve $ \poly(1/\delta) $ list sizes.   

\section{Prior work on Haar lattices}\label{sec:haarlattice_introduction}

Let us first introduce the Haar distribution on the set of all lattices. 
A more detailed exposition can be found in the thesis of Kim~\cite{kim2015thesis}. 

For the convenience of illustration, let us collect all rank-$n$ lattices $\Lambda\subset\bR^n$ with covolume one\footnote{This is without loss of generality since normalization does not affect goodness properties.} into a set $\sL_n$:
\[\sL_n\coloneqq\curbrkt{\Lambda\le\bR^n\text{ lattice}\colon \det(\Lambda)=1}.\]
A lattice in $\sL_n$ is specified by its generator matrix $\bfG\in\speclin(n,\bR)$. However, one lattice $\Lambda$ can have multiple different generator matrices. Indeed, two matrices $\bfG$ and $\wt \bfG$ give rise to the same lattice (i.e., $ \bfG\bZ^n = \wt\bfG\bZ^n $) iff they differ by an $\speclin(n,\bZ)$  matrix, i.e.,  $\bfG \bfG'=\wt \bfG$ where $\bfG'\in\speclin(n,\bZ)$. Hence $\sL_n$ can be identified with the quotient space
\[\sL_n=\speclin(n,R)/\speclin(n,\bZ).\]
Crucial to us is Haar's seminal result on the existence of \emph{Haar measure} on any locally compact topological group. Specialized to our setting, it was shown by Siegel~\cite{siegel-1945-first-order-avg-formula-haar} the existence and finiteness of a certain nicely-behaved distribution on $\sL_n$.
\begin{theorem}[\cite{siegel-1945-first-order-avg-formula-haar}]\label{thm:haar_slnr}
	There is a unique (up to a multiplicative constant factor) measure $\mu$ (called the {Haar} measure or the Haar--Siegel measure) on $\speclin(n,\bR)$ which satisfies the following properties:
	\begin{enumerate}
		\item\label{itm:haar_inv} $\mu$ is left-$\speclin(n,\bR)$-invariant, i.e., for any Borel subset $\cK\subseteq\speclin(n,\bR)$ and any $\bfG\in\speclin(n,\bR)$, $\mu(\cK)=\mu(\bfG\cK)$; 
		\item\label{itm:haar_finite} $\mu$ is finite, i.e., for any compact subset $\cK\subseteq\speclin(n,\bR)$, $\mu(\cK)<\infty$.
	\end{enumerate}
\end{theorem}
Note that we can  normalize the Haar measure $\mu$ to make it a probability distribution, i.e., $\mu(\speclin(n,\bR))=1$. In this paper we always refer to the normalized version when talking about $\mu$ or Haar measure. The Haar distribution on $\sL_n $ naturally inherits that on $\speclin(n,\bR)$. We do not specify measure-theoretic details which can be found in, e.g.,~\cite{quint-2013-dynamical-sys-on-homo-sp}. With abuse of notation, we use the same notation $\mu$ for the Haar measure on $\speclin(n,\bR)$ and the induced Haar measure on $ \sL_n $. Most of the time we refer to the former one which will be clear from the context.

The above result only provides the existence and properties of the Haar measure but does not provide an explicit form of this measure. What does the Haar measure $\mu$ on $\speclin(n,\bR)$ look like?  It can be checked that the Lebesgue measure on $\bR^{n^2}$ satisfies the properties~\ref{itm:haar_inv} and~\ref{itm:haar_finite} required in Theorem~\ref{thm:haar_slnr}, and hence is the Haar measure. Given any measure, besides that we can use it to measure a compact subset of the space, we can also integrate functions on the same space against this measure. Since Haar measure is unique, we know that for $\bfG\in\speclin(n,\bR)$, \[\diff\mu(\bfG)=\diff\vect{\bfG},\] where $ \vect{\bfG}\in\bR^{n^2} $ denotes the vectorization of $\bfG$. Namely, the Haar measure of a matrix in $\speclin(n,\bR)$ is equal to the Lebesgue measure of it when viewed as a vector in $\bR^{n^2}$. 



As a byproduct of the above reasoning, we also know that the Haar measure on $\genlin\paren{n,\bR}$ is just the normalized Lebesgue measure on $\bR^{n^2}$. For $\bfG\in\genlin\paren{n,\bR}$,
\[\diff\mu\paren{\bfG}=\frac{\diff {\vect{\bfG}}}{\det\paren{\bfG}^{1/n}}.\]
It is a valid definition since 
\[\det\paren{\frac{\bfG}{\det(\bfG)^{1/n}}}=\paren{\frac{1}{\det(\bfG)^{1/n}}}^n\det(\bfG)=1,\]
and the definition is reduced to the one on  $\speclin(n,\bR)$.
Here again with abuse of notation, we use the same notation for Haar measure on $\speclin(n,\bR)$ and $\genlin(n,\bR)$.

One may resort to  Iwasawa (KAN) decomposition~\cite{quint-2013-dynamical-sys-on-homo-sp} for a more explicit characterization of the Haar measure.


\subsection{Siegel's and Rogers's averaging formulas}\label{sec:siegel_rogers_avg_formulas}

Let us first recall two fundamental averaging formulas which are heavily used in the literature for understanding the distribution of short vectors of a random lattice drawn from the Haar distribution. 

In the same seminal paper~\cite{siegel-1945-first-order-avg-formula-haar} in which Siegel showed the existence and uniqueness of Haar distribution on the space of unit-covolume lattices, he also proved the following averaging formula.
\begin{theorem}[\cite{siegel-1945-first-order-avg-formula-haar}]\label{thm:first-order-avg-formula}
	Let $\rho:\bR^n\to\bR$ be a bounded, measurable, compactly supported function. Then
	\begin{equation}
	\exptover{\Lambda\sim\mu}{\sum_{\vx\in\Lambda\setminus\{0\}}\rho(\vx)}=\int_{\sL_n}\sum_{\vx\in\Lambda\setminus\{0\}}\rho(\vx)\diff \mu(\Lambda)=\int_{\bR^n}\rho(\vx)\diff \vx.
	\label{eqn:first-order-avg-formula}
	\end{equation}
\end{theorem}
\begin{remark}
	The requirement that we are allowed to evaluate the function $\rho$ only at nonzero lattice points could be potentially inconvenient in applications. One can drop this condition by paying an extra term, i.e., the value of $\rho$ at the origin, on the RHS of Eqn.~\eqref{eqn:first-order-avg-formula} and the formula becomes 
	\begin{equation}
	\exptover{\Lambda\sim\mu}{\sum_{\vx\in\Lambda}\rho(\vx)}=\int_{\sL_n}\sum_{\vx\in\Lambda}\rho(\vx)\diff \mu(\Lambda)=\rho(0)+\int_{\bR^n}\rho(\vx)\diff \vx.
	\label{eqn:first-order-avg-formula-alternate}
	\end{equation}
	These two forms are completely equivalent and we will state only one of them but potentially use any of them without further explanation depending on whichever is convenient. 
\end{remark}
The identity holds in large generality for any reasonably nice function $\rho$. Perhaps the most important consequence of this formula is that it gives a way to estimate the number of lattice points in a measurable set, which is in turn an ubiquitous primitive in applications. Specifically, for our list decoding purposes, essentially the only thing we need to control is the number of lattice points in a ball. If we take 
\[\rho(\vx)\coloneqq\indicator{\vx\in\cB(\vy,r)}\] to be the indicator function of an Euclidean ball centered at $\vy$ of radius $r$ (which obviously satisfies the conditions required by Theorem~\ref{thm:first-order-avg-formula}), then the left-hand side (LHS) of~\eqref{eqn:first-order-avg-formula} is nothing but the expected number of nonzero Haar lattice points in the ball. Siegel's formula tells us that this is equal to the RHS of~\eqref{eqn:first-order-avg-formula} which is actually the volume of the ball. This matches our intuition that the number of lattice points in any measurable set $\cB$ should be roughly the ratio between the volume of $\cB$ and the volume of a Voronoi cell of the lattice, i.e., $ \card{\Lambda\cap\cB} \approx \vol(\cB)/\det(\Lambda)=\vol(\cB)$ since we consider normalized lattices. Siegel's formula indicates that the Haar distribution on $\sL_n$ behaves typically in a sense that such intuition is indeed true in expectation.

One simple application of Theorem~\ref{thm:first-order-avg-formula} is that it allows us to control the rate of a Haar lattice code. 
For a lattice $ \Lf\sim\sL_n $, if we define the lattice code $ \cC $ to be $ \Lf\cap \cB(0,\sqrt{nP}) $, then Theorem~\ref{thm:first-order-avg-formula} lets us conclude that 
\[
\frac{1}{n}\log\bE_{\Lf}\left[|\cC|\right] = \frac{1}{2}\log P +o(1).
\]

It turns out there is a  higher-order generalization of Siegel's formula due to Rogers~\cite{rogers-1955-haar-equiv-ensemble} which we introduce below. 
\begin{theorem}[\cite{rogers-1955-haar-equiv-ensemble}, Theorem 4]\label{thm:higher-order-avg-formula}
	Let $k<n$ be a positive integer. Let
	\begin{equation*}
	\begin{array}{rlll}
	\rho\colon & (\bR^n)^k & \to & \bR
	\end{array}
	\end{equation*}
	be a {bounded Borel measurable} function {with compact support}. Then
	\begin{align}\label{eqn:higher-order-avg-formula}
	\exptover{\Lambda\sim\mu}{\sum_{\vx_1,\cdots,\vx_k\in\Lambda}\rho(\vx_1,\cdots,\vx_k)}=&\int_{\sL_n}\sum_{\vx_1,\cdots,\vx_k\in\Lambda}\rho\paren{\vx_1,\cdots,\vx_k}\diff\mu\paren{\Lambda}\\\notag
	=&\rho\paren{0,\cdots,0}+\int_{\bR^n}\cdots\int_{\bR^n}\rho\paren{\vx_1,\cdots,\vx_k}\diff \vx_1\cdots\diff \vx_k+\cE,\notag
	\end{align}
	where $\cE$ is some crazy-looking error term: 
	\begin{align*}
	\cE\coloneqq\sum_{(\vec\alpha,\vec\beta)}\sum_{\ell=1}^\infty\sum_{\bfD}\paren{\frac{e_1}{\ell}\cdots\frac{e_m}{\ell}}^n\int_{\bR^n}\cdots\int_{\bR^n}\rho\paren{\sum_{i=1}^m\frac{\bfD_{i1}}{\ell}\vx_i,\cdots,\sum_{i=1}^m\frac{\bfD_{ik}}{\ell}\vx_i}\diff \vx_1\cdots\diff\vx_m.
	\end{align*}
	Here the first sum is over all divisions $(\vec\alpha,\vec\beta)=(\alpha_1,\cdots,\alpha_m;\beta_1,\cdots,\beta_{k-m})$ of the numbers $1,\cdots,k$ into two sequences $1\le\alpha_1<\cdots<\alpha_m\le k$ and $1\le\beta_1<\cdots<\beta_{k-m}\le k$ with $1\le m\le k-1$ and $\alpha_i\ne\beta_j$ for any $i,j$. The third sum is taken over all integral $m\times k$ matrices $\bfD\in\bZ^{m\times k}$ such that
	\begin{enumerate}
		\item no column of $\bfD$ vanishes;
		\item the greatest common divisor of all entries is 1;
		\item for all $i\in[m],s\in[m],t\in[k-m]$, $\bfD_{i\alpha_s}=\ell\indicator{i=s}$ and $\bfD_{i\beta_t}=0$ if $\beta_t<\alpha_i$.
	\end{enumerate}
	Finally, $e_i=(\gamma_i,\ell)$, where $\gamma_1,\cdots,\gamma_m$ are the elementary divisors of $\bfD$. 
\end{theorem}

If we take
\[\rho\paren{\vx_1,\cdots,\vx_k}\coloneqq\indicator{\vx_1\in\cB}\cdots\indicator{\vx_k\in\cB},\]
where $\cB\coloneqq\cB(\vy,r)$ is a ball, then Rogers's formula is precisely computing \[\exptover{\Lambda\sim\mu}{\card{\Lambda\cap\cB}^k}\] for $1\le k\le n-1$. 

The proof of Rogers's averaging formula is highly nontrivial and can be divided into three steps. Since the proof contains several ingenious ideas and can be instructive for other purposes, we sketch it below.

\noindent\textbf{Step I.} Consider any real-valued bounded Borel measurable function of bounded support on unit-covolume lattices, 
\[\begin{array}{rlll}
f:&\sL_n&\to &\bR.
\end{array}\]
We will interchangeably think of $f$ as a function on $\speclin\paren{n,\bR}$,
\[\begin{array}{rlll}
f:&\speclin\paren{n,\bR}&\to &\bR
\end{array}\]
by interchangeably thinking of $\Lambda$ as a lattice or its generator matrix. The function $f$ can be naturally extended from $\speclin\paren{n,\bR}$ to $\genlin\paren{n,\bR}$ by defining, for $\Lambda\in\genlin\paren{n,\bR}$,
\begin{equation}
f\paren{\Lambda}\coloneqq f|_{\speclin(n,\bR)}\paren{\det\paren{\Lambda}^{-1/n}\Lambda}.
\label{eqn:def_extension}
\end{equation}
Note that $\det\paren{\Lambda}^{-1/n}\Lambda$ always has determinant one.

Fix $\omega\in\bR_{>0}$. Let $\bfTheta=\bfTheta\paren{\theta_1,\cdots,\theta_{n-1},\omega}\in \bR^{n\times n}$  be drawn from the following ensemble
\begin{equation}
\begin{bmatrix}
\omega&&&&\\
&\omega&&&\\
&&\ddots&&\\
&&&\omega&\\
\omega^{-(n-1)}\theta_1&\omega^{-(n-1)}\theta_2&\cdots&\omega^{-(n-1)}\theta_{n-1}&\omega^{-(n-1)}
\end{bmatrix},
\label{eqn:rogers_ensemble}
\end{equation}
where each ${\theta}_i\sim \cU\paren{[0,1]}$. Note that any matrix of the above form has determinant one. 
\begin{remark}
	The reason behind the choice of this ensemble has connections to number theory. This is well beyond the scope of this paper and we refer interested readers to~\cite{einsiedler-ward-ergodic-thy-number-thy, goldstein-mayer-2003-equidist-hecke-pts} for relevant background. 
\end{remark}
Let $\vec\theta\coloneqq(\theta_1,\cdots,\theta_{n-1})$. The average of $f$ wrt such an ensemble can be written as
\[\exptover{\theta\sim\cU([0,1])}{f\paren{\bfTheta(\vec\theta,\omega)\bZ^n}}=\int_0^1\cdots\int_0^1f\paren{\bfTheta(\vec\theta,\omega)\bZ^n}\diff\theta_1\cdots\diff\theta_{n-1}.\]
Let \[M\paren{f}\coloneqq\lim_{\omega\to0+}\exptover{\theta\sim\cU([0,1])}{f\paren{\bfTheta(\vec\theta,\omega)\bZ^n}}.\]

Rogers~\cite{rogers-1955-haar-equiv-ensemble} showed the following (perhaps surprising) identity.
\begin{theorem}[\cite{rogers-1955-haar-equiv-ensemble}, Theorem 1]\label{thm:rogers_equality}
	Let $\rho:\sL_n\to\bR$ be a bounded, measurable, compactly supported function. Suppose that the limit $M\paren{f}$ exists. 
	Then
	\[\exptover{\Lambda\sim\mu}{f(\Lambda)}=\int_{\sL_n}f\paren{\Lambda}\diff\mu\paren{\Lambda}=M\paren{f}.\]
\end{theorem}
A similar averaging result holds for Construction-A lattices. See~\cite{loeliger-1997-avg-formula-constr-a-lin-codes}.

\noindent\textbf{Step II.} Equipped with the powerful Theorem~\ref{thm:rogers_equality}, computation regarding expectations wrt Haar distribution can be turned into computation wrt the aforedefined concrete ensemble. Rogers then gave a formula for the expectation  of functions of a particular form by computing $M(\cdot)$. It can be shown that Eqn.~\eqref{eqn:higher-order-avg-formula} holds exactly true without the error term if we only sum over \emph{linearly independent/full-rank} $k$-tuples.
\begin{theorem}[\cite{rogers-1955-haar-equiv-ensemble}, Theorem 2, Lemma 1 and Theorem 3]\label{thm:higher-order-avg-formula-lin-indep}
	Let $k$ and $\rho$ be as in the setting of Theorem~\ref{thm:higher-order-avg-formula}. Let 
	\[f'(\Lambda)\coloneqq\sum_{\substack{\vx_1,\cdots,\vx_k\in\Lambda\\\rk\curbrkt{\vx_1,\cdots,\vx_k}=k}}\rho(\vx_1,\cdots,\vx_k).\]
	Then
	\begin{align}
	M(f')=&\rho\paren{0,\cdots,0}+\int_{\bR^n}\cdots\int_{\bR^n}\rho\paren{\vx_1,\cdots,\vx_k}\diff \vx_1\cdots\diff \vx_k.
	\label{eqn:higher-order-avg-formula-lin-indep}
	\end{align}
\end{theorem}

\noindent\textbf{Step III.} Rogers finally completed the proof of Theorem~\ref{thm:higher-order-avg-formula} by dropping the linear independence condition and lifting Theorem~\ref{thm:higher-order-avg-formula-lin-indep} from $f'$ to 
\[f(\Lambda)\coloneqq\sum_{{\vx_1,\cdots,\vx_k\in\Lambda}}\rho(\vx_1,\cdots,\vx_k)\]
as promised in Theorem~\ref{thm:higher-order-avg-formula} at the cost of an extremely complicated error term $\cE$.

\subsection{Improvement on Rogers's formula}\label{sec:improvement_on_rogers}
Although we have Rogers's higher-order averaging formula, it turns out that the error term $\cE$ is very tricky to control even if we just plug in simple product functions. In the original paper by Rogers~\cite{rogers-1955-mmts-of-number-of-lattice-pts, rogers-1956-mmts-of-number-of-lattice-pts-error-term}, he was only able to show convergence of the first \emph{few} moments of number of random lattice points in a \emph{symmetric} set of \emph{fixed} volume. Nevertheless, an intriguing Poisson behaviour was discovered and has been pushed to a greater generality in recent years.\footnote{Actually, Rogers showed that, asymptotically in the number of dimensions $n$, the first $\cO(\sqrt{n})$ moments of the number of random lattices points in a set $\cS$ which is centrally symmetric wrt the origin exhibit the same behaviour as a Poisson moment of the same degree with mean $V/2$, where $V\coloneqq\vol(\cS)$ is a constant independent of $n$. As we will see later, this is too weak for our purpose of list decoding. However, it is the earliest result which kicks off a fantastic adventure towards understanding the statistics of random lattices.} We state below, as far as we know, the strongest results along this direction. 

Let $Y\sim\pois(V/2)$ be a Poisson random variable of mean $V/2$ for some $V$ to be specified later.

Kim showed the following improvement upon Rogers results.
\begin{theorem}[Proposition 3.3 of~\cite{kim2016random}]\label{thm:kim-improvement}
	Let $\cB$ be a centrally symmetric set in $\bR^n$ of volume $V$. There exists constants $C,c>0$ such that, if $n$ is sufficiently large and $V,k\le Cn$, then
	\begin{align*}
	\prob{Y\ge k}-e^{-cn}\le&\prob{\frac{1}{2}\card{(\Lambda\setminus\{0\})\cap\cB}\ge k}\le\prob{Y\ge k}+e^{-cn}.
	\end{align*}
\end{theorem}
Note that the number of \emph{pairs} of lattice points is considered since if $\vx\in\Lambda$ then so is $-\vx$. That is why there is a normalization factor $1/2$ in front of the number of nonzero lattice points in $\cB$.

Str\"ombergsson and S\"odergren provided another improvement on the distribution of short vectors in a random lattice.
\begin{theorem}[Theorem 1.2 of~\cite{ss-2016-gen-gauss-circle}]\label{thm:ss-improvement}
	Let $\cB$ be an $n$-dimensional Euclidean ball centered at the origin of volume $V$. For any $\eps>0$,
	\begin{align*}
	\prob{\frac{1}{2}\card{(\Lambda\setminus\{0\})\cap\cB}\le k}-\prob{Y\le k}\stackrel{n\to\infty}{\to}0,
	\end{align*}
	uniformly wrt all $k,V\ge0$ satisfying $\min\{k,V\}\le\cO_\eps(e^{\eps n})$.
\end{theorem}

We remark that though both results by Kim and Str\"ombergsson--S\"odergren are great extensions of Rogers's results to higher-order averaging formulas, they are not directly comparable. In Kim's Theorem~\ref{thm:kim-improvement}, the set $\cB$ can be any symmetric body, not necessarily convex. This is a good news since in list decoding we care about the number of lattice points in $\cB(\vy,r)$ for any possible received vector $\vy\in\cB(0,\sqrt{nP}+\sqrt{nN})\setminus\cB(0,\sqrt{nP}-\sqrt{nN})$. Kim's result allows us to control that by taking $\cB=\cB(\vy,r)\sqcup\cB(-\vy,r)$ (assuming $ \cB(\vy,r)\cap\cB(-\vy,r)=\emptyset $). Obviously the configuration of lattice points are symmetric in $\cB(\vy,r)$ and $\cB(-\vy,r)$. Hence $|\Lambda\cap\cB|=2|\Lambda\cap\cB(\vy,r)|$. Also, Kim's result holds for $k=\cO(n)$ which is also sufficient in our application, as we will see. Kim also quantified an exponential convergence rate. Unfortunately, his result requires $V$ to be $\cO(n)$, which is not enough for us.
On the other hand, Str\"ombergsson--S\"odergren's result pushed the volume $V$ to exponentially large in $n$ but insists on $\cB$ being a ball centered at the origin.

It should be intuitively clear that the Poissonianity behaviour of the moments will not hold for arbitrarily large degrees and for sets of arbitrarily large volume. The dimension that the lattice is living in is only $n$. If we compute the moments of very high degrees, we should expect to encounter some nontrivial correlation which makes the moments tricky to understand. Moreover, if we compute the moments of the number of lattice points in a very large set, it should not be surprising that at some point linearity of the lattices will kick in and dominate the behaviour of the moments. 

\section{List decodability of Haar lattices}\label{sec:ld_haar}

Given the state of the art of bounds on moments of the number of Haar lattice points, we pose the following conjecture and use it to show conditional results on list decodability of Haar lattices in the next section. The known properties of the Haar distribution that we have outlined previously should hopefully provide reasonable justification for why we believe that our conjectures are true.
\begin{conjecture}[Poisson moment assumption]\label{conj:poissonianity}
	Let $\cB$ be any symmetric set in $\bR^n$ of volume $V=2^{\cO(n)}$ and $k=cn$ for some constant $0<c<1$. Then 
	\[\exptover{Y\sim\pois(V/2)}{Y^k}\le\exptover{\Lambda\sim\mu}{\paren{\frac{\card{\Lambda\cap\cB}}{2}}^k}\le\exptover{Y\sim\pois(V/2)}{Y^k}+o(1).\]
	Recall that the $k$-th moment of a Poisson random variable (Fact~\ref{fact:pois_mmt}) is
	\[\exptover{Y\sim\pois(V/2)}{Y^k}=e^{-V/2}\sum_{i=0}^\infty\frac{i^k}{i!}(V/2)^i.\]
\end{conjecture}

Note that results/conjectures phrased using tail bounds or moment bounds are essentially equivalent since one can be converted to another using the well-known relation between tails and moments. For any (continuous) random variable $X$ with known tails, we can estimate its moment via
\[\expt{|X|^k}=\int_0^\infty kt^{k-1}\prob{|X|>t}\diff t.\]
For any (continuous) random variable $X$ with known moments, we can bound its tail via the Chernoff-type inequality,
\[\prob{|X|>t}\le\frac{\expt{|X|^k}}{t^k}.\]




Previously, we showed that lattices and nested lattice codes can achieve $ 2^{\cO(\frac{1}{\delta}\log^2\frac{1}{\delta})} $ list sizes whereas random spherical codes and periodic ICs achieve list sizes that grow as $ \cO(\frac{1}{\delta}\log\frac{1}{\delta}) $. This leads to the natural question: Do there exist lattices/nested lattice codes that achieve $ \cO(\poly(1/\delta)) $ list sizes? Are the exponential growth of the list sizes a consequence of structural regularity (i.e., linearity of the lattices) or is it an artifact of our proof?  We conjecture that  lattices can indeed achieve $ \cO(\poly(1/\delta)) $ although we are unable to supply a complete proof at present. However, based on some heuristic assumptions, we can ``prove'' that a different ensemble of lattice codes (based on Haar lattices) achieve $ \cO(\poly(1/\delta)) $ list sizes.

\subsection{Conditional list decodability of Haar lattices}\label{sec:cond_ld_haar}
\subsubsection{Codebook construction}\label{sec:cond_ld_haar_cb}
Let $R=\frac{1}{2}\log\frac{P}{N}-\delta$ for some small constant $\delta>0$. Sample a lattice $\Lambda$ from the Haar distribution on $\sL_n$. The lattice codebook is nothing but $\cC\coloneqq \alpha\Lambda\cap\cB(0,\sqrt{nP})$. 
Note that
\begin{align*}
\card{\cC}=&\card{\alpha\Lambda\cap\cB(0,\sqrt{nP})}=\card{\Lambda\cap \alpha^{-1}\cB(0,\sqrt{nP})}.
\end{align*}
By Siegel's formula (Theorem~\ref{thm:first-order-avg-formula}), the expected number of codewords in the codebook is
\begin{align*}
\expt{\card{\cC}}=&\frac{\vol\paren{\cB^n\paren{0,\sqrt{nP}}}}{\alpha^n}=\frac{\sqrt{nP}^nV_n}{\alpha^n}.
\end{align*}
Equaling it $2^{nR}$, we have
\begin{align*}
\alpha=&\frac{\sqrt{nP}V_n^{1/n}}{2^R}\asymp\frac{\sqrt{2\pi eP}}{2^R}.
\end{align*}
This coupled with the proceeding computation will provide the (conditional) \emph{existence} of a $(P,N,\poly(1/\delta))$-list decodable lattice code.

\subsubsection{Under distribution assumption}\label{sec:cond_ld_haar_dist}
Heuristically and  unrealistically, we first assume that the number of lattice points follows exactly a Poisson distribution, i.e., \emph{every} moment of it is Poissonian.

\begin{conjecture}[Poisson distribution assumption]
	\label{conj:pois_dist}
	Let $ \Lf $ be a random lattice drawn from the Haar distribution on $ \sL_n $. If $ \cB $ is any centrally symmetric measurable set with nonempty interior, then 
	$ {|\Lf\cap\cB|}/{2}\sim\pois(\vol(\cB)/2) $. 
\end{conjecture}

This assumption is \emph{not} believed to be true. As we mentioned before, at some point the linearity of the lattice should kick in and the moments are expected to diverge from Poissons as the order of the moments grows. Nevertheless, in this section we still conduct computation under this assumption that seems too good to be true. 
The result sets the bar of the ``best" list decoding performance one can hope for, though it may never be reached in reality. 

Another motivation for doing these calculations is that the same quantitative results under the the distributional assumption can be viewed as \emph{rigorous} results for another code ensemble, that is, a Poisson point process (PPP) restricted to a ball. 
A homogeneous PPP has the property that the number of points in any compact set $ \cB $ with nonempty interior is distributed according to $ \pois(\vol(\cB)) $. 
One subtle difference between this and the distribution assumption is that for lattices we need to normalize the number of lattice points in $ \cB $ by $ 1/2 $. 
This is due to the linear structure of $\Lf$ -- if $ \vx\in\Lf $, then $ -\vx\in\Lf $ with probability 1. 
Therefore, for any conjecture of this kind to make sense, the normalization factor $ 1/2 $ is necessary.

Under the construction in Sec.\ \ref{sec:cond_ld_haar_dist}, invoking Conjecture~\ref{conj:pois_dist}, we can get a high-probability guarantee on the size of the codebook. First note that
\begin{align*}
\vol\paren{{{\alpha}}^{-1}\cB\paren{0,{\sqrt{nP}}}}=&\expt{|\cC|}
=2^{nR}.
\end{align*}
By the Poisson tail bound (Lemma~\ref{lem:pois-tail-clement}), 
\begin{align}
\prob{\abs{\frac{|\cC|}{2} - 2^{nR}}\ge\frac{1}{2}2^{nR}}
&\le 2\exp\paren{-\frac{(0.5\cdot2^{nR})^2}{2\cdot(2^{nR} + 0.5\cdot2^{nR})}}
= 2\exp\paren{-\frac{2^{nR}}{24}}. \notag 
\end{align}
That is to say, with probability at least $1-e^{-\Omega(2^{nR})}$, $ 0.5\cdot2^{nR}< |\cC|/2<1.5\cdot2^{nR} $, i.e., $ 2^{nR}<|\cC|<3\cdot2^{nR} $. 
Therefore, the rate $R(\cC)$ of the code is $\frac{1}{2}\log\frac{P}{N}-\delta+o(1)$.

We then upper bound the following probability of failure of list decoding:
\begin{equation}
\prob{ \exists \vy\in\cB^n\paren{ 0,\sqrt{nP}+\sqrt{nN} },\;\card{ \alpha\Lambda\cap\cB^n\paren{\vy,\sqrt{nN}} }> L }.
\label{eqn:ld_pe}
\end{equation}
Take an optimal $\sqrt{n\eps}$-covering $\cY$ of $\cB^n\paren{ 0,\sqrt{nP}+\sqrt{nN} }$. It can be achieved that
\[\card{\cY}=\paren{\frac{\vol\paren{\cB^n\paren{ 0,\sqrt{nP}+\sqrt{nN} + \sqrt{n\eps} }}}{\vol\paren{\cB^n\paren{0,\sqrt{n\eps}}}}}^{1+o(1)}=\paren{\frac{\sqrt{P}+\sqrt{N}+\sqrt{\eps}}{\sqrt{\eps}}}^{(1+o(1))n}\le\paren{\frac{c_2}{\delta}}^n,\]
where in the last step we set $\eps\coloneqq c_1\delta^2$. 
Then the probability~\eqref{eqn:ld_pe} is upper bounded by
\begin{align}
\prob{ \exists \vy\in\cY,\;\card{ \alpha\Lambda\cap\cB^n\paren{\vy,\sqrt{nN}+\sqrt{n\eps}} }> L }
\le&\sum_{\vy\in\cY}\prob{ \card{\alpha\Lambda\cap\cB^n\paren{\vy,\sqrt{nN}+\sqrt{n\eps}}}>L }.\label{eqn:heuristic}
\end{align}


Let 
\begin{align*}
\cB_1\coloneqq & \alpha^{-1}\cB^n\paren{\vy,\sqrt{nN}+\sqrt{n\eps}} \cup \alpha^{-1}\cB^n\paren{-\vy,\sqrt{nN}+\sqrt{n\eps}} ,\\
\cB_2\coloneqq & \alpha^{-1}\cB^n\paren{\vy,\sqrt{nN}+\sqrt{n\eps}} \cap \alpha^{-1}\cB^n\paren{-\vy,\sqrt{nN}+\sqrt{n\eps}} .
\end{align*}
Note that 
\begin{equation}
\vol(\cB_1) + \vol(\cB_2)=2\vol\paren{\alpha^{-1}\cB^n\paren{\sqrt{nN}+\sqrt{n\eps}}}.
\label{eqn:vol_identity}
\end{equation}
By our assumption (Conjecture~\ref{conj:pois_dist}) in this section,
\begin{align}
\frac{1}{2}\card{\Lambda\cap \alpha^{-1}\cB^n\paren{\vy,\sqrt{nN}+\sqrt{n\eps}}} 
\le& \frac{1}{2}\card{\Lf\cap\cB_1} \notag \\
\sim& \pois(\vol(\cB_1)/2)\notag \\
\le& \pois(\vol(\cB_1)/2) + 2\pois(\vol(\cB_2)/2)\notag \\
= & \pois\paren{\frac{\vol(\cB_1)+\vol(\cB_2)}{2}} \label{eqn:app_pois_add}\\
= & \pois\paren{\vol\paren{\alpha^{-1}\cB^n\paren{\sqrt{nN}+\sqrt{n\eps}}}}. \label{eqn:app_vol_identity}
\end{align}
Eqn.~\eqref{eqn:app_pois_add} and~\eqref{eqn:app_vol_identity} follow from Fact~\ref{fact:pois_add} and Eqn.~\eqref{eqn:vol_identity}, respectively.
Plugging the parameters into the bound in Lemma~\ref{lem:pois_tail}, we can upper bound the probability in Eqn.~\eqref{eqn:heuristic} by
\begin{equation}
\prob{ \card{\Lambda\cap\cB^n\paren{\alpha^{-1}\vy,\frac{\sqrt{nN}+\sqrt{n\eps}}{{\alpha}}}}>L }<\frac{e^{-V}\paren{eV}^{L/2}}{(L/2)^{L/2}},
\label{eqn:pois_tail}
\end{equation}
where 
\begin{equation}
V\coloneqq\vol\paren{\cB^n\paren{\frac{\sqrt{nN}+\sqrt{n\eps}}{{\alpha}}}}=\paren{\frac{\sqrt{nN}+\sqrt{n\eps}}{{\alpha}}}^nV_n\asymp 2^{n\paren{R-\frac{1}{2}\log\frac{P}{N+2\sqrt{N\eps}+\eps}}}\approx 2^{-c_3n\sqrt{\eps}} <L.
\label{eqn:pois_param}
\end{equation}
In the last step of the above chain of equalities, we set $c_3\approx1/\sqrt{c_1} - (\log e)/\sqrt{N}$ and use that $R=\frac{1}{2}\log\frac{P}{N}-\delta$, $\eps=c_1\delta^2$ and $\log(1+x)\approx (\log e)x$. Hence the RHS of the tail~\eqref{eqn:pois_tail} is 
\[\exp\paren{-2^{-c_3n\sqrt{\eps}}}\paren{e2^{-c_3n\sqrt{\eps}}}^{L/2}/(L/2)^{L/2}\asymp\paren{\frac{e}{L/2}}^{L/2}2^{-c_3n\sqrt{\eps}L/2}.\]
Taken a union bound over $\cY$, the overall probability of failure of list decoding (Eqn.~\eqref{eqn:ld_pe}) is at most
\[\paren{\frac{e}{L/2}}^{L/2}2^{-c_3n\sqrt{\eps}L/2}\paren{\frac{c_2}{\delta}}^n=\paren{\frac{e}{L/2}}^{L/2}2^{-n\paren{{c_3\sqrt{c_1}}\delta L/2-\log\frac{c_2}{\delta}}}.\]
The multiplicative factor $\paren{\frac{e}{L/2}}^{L/2}$ is going to be negligible once $n$ is sent to infinity. The exponent is negative if we set $L$ to be $c'\frac{1}{\delta}\log\frac{1}{\delta}$ for some appropriate constant $c'$. 

The above calculations indicate that, under the Poisson distributional assumption of the number of lattice points in a set, a random lattice (appropriately scaled) drawn from the Haar measure performs as well as uniformly random spherical codes.
We therefore have the following result:
\begin{lemma}
	If Conjecture~\ref{conj:pois_dist} is true, then there exists a lattice $ \Lf $ such that $ \Lf\cap\cB(0,\sqrt{nP}) $ has rate $ C(P,N)-\delta $ and is $ \left(P,N ,\cO(\frac{1}{\delta}\log\frac{1}{\delta})\right)$-list decodable.
\end{lemma}

\subsubsection{Under moment assumption}\label{sec:cond_ld_haar_mmt}
Now instead of assuming that the number of lattice points in any symmetric body has Poisson distribution, we only assume that its first $k$ moments match Poisson moments. 


First note that the rate of the code is still well concentrated: 
\begin{align}
\prob{\abs{\frac{|\cC|}{2} - 2^{nR}}\ge 2^{n(R+\delta/2)}} 
&\le \frac{\var{{|\cC|}/{2}}}{\paren{2^{n(R+\delta/2)}}^2}
\le \frac{2^{nR}}{2^{n(2R + \delta)}}
= 2^{-n(R+\delta)}. \notag 
\end{align}
The second inequality follows since the first and second moments of $ |\cC|/2 $ are the same as those of $ \pois(2^{nR}) $. 
Therefore $ R(\cC) = R+\delta/2+o(1) = \frac{1}{2}\log\frac{P}{N} - \delta/2 + o(1) $. 

Let 
\begin{equation}
    X=\frac{1}{2}\card{\Lambda\cap\cB^n\paren{\frac{\sqrt{nN}+\sqrt{n\eps}}{{\alpha}}}}.
    \label{eqn:x_def}
\end{equation}
By Conjecture~\ref{conj:poissonianity} and Fact~\ref{fact:pois_mmt}, for any $0\le m\le k$,
\[\expt{X^m}\asymp\expt{Y^m}=e^{{-\lambda}}\sum _{{j=0}}^{\infty }{\frac  {\lambda^{j}j^{m}}{j!}},\]
where $Y\sim\pois\paren{\lambda}$, $\lambda=V/2$ and $V$ is given by formula~\eqref{eqn:pois_param}. Indeed,
\begin{equation}
\lambda\approx\frac{1}{2}\cdot2^{-c_3n\sqrt{\eps}}
\stackrel{n\to\infty}{\to}0.
\label{eqn:pois_mean}
\end{equation}
Then the probability in Eqn.~\eqref{eqn:heuristic} can be upper bounded by
\begin{align}
\prob{X^k>(L/2)^k}<{\expt{X^k}}/{(L/2)^k}=\expt{X^k}e^{-k\ln(L/2)}.\notag
\end{align}
Let $k\eqqcolon cn$ where $0<c<1$ is a constant. If $\expt{X^k}\le e^{-nD}$ for some $D>0$, then after taking a union bound over $\vy\in\cY$, we are in good shape if
\begin{equation}
e^{-n\paren{D+c\ln\frac{L}{2}-\ln\frac{c_2}{\delta}}}\stackrel{n\to\infty}{\to}0.
\label{eqn:mmt_bd}
\end{equation}
Now let us compute the $k$-th moment.
\begin{align}
\expt{X^k}=&e^{{-\lambda}}\sum _{{j=0}}^{\infty }{\frac  {\lambda^{j}j^{k}}{j!}}\notag\\
\asymp&\sum_{j\ge0}\frac{\lambda^jj^k}{\sqrt{2\pi j}\paren{j/e}^j}\label{eqn:app_stirling}\\
=&\sum_{j\ge0}\exp\paren{j\ln\lambda+k\ln j-j\ln j+j-\frac{1}{2}\ln\paren{2\pi j}},\notag
\end{align}
where in Eqn.~\eqref{eqn:app_stirling} we use Stirling's approximation (Lemma~\ref{lem:stirling}).
As we know, a sum of exponentials is dominated by the largest term. Let us compute the largest one. 
Define function
\[f\paren{j}\coloneqq-j\ln j+\paren{\ln\lambda+1}j+\paren{k-1/2}\ln j-\frac{1}{2}\ln\paren{2\pi}.\]
Its first derivative is given by
\[\frac{\diff f}{\diff j}=\ln\lambda+\frac{k-1/2}{j}-\ln j.\]
Equaling it zero and solving the equation, we get the critical point
\[j^*\coloneqq\frac{k-1/2}{W\paren{\frac{k-1/2}{\lambda}}},\]
where $W(\cdot)$ is the \emph{Lambert $W$ function} which is the inverse of $ g(x) = xe^x $. The function $ W(\cdot) $ satisfies the following estimate for sufficiently large $x$,
\[W(x)=\ln x-\ln \ln x+o(1).\]
Note that, by Eqn.~\eqref{eqn:pois_mean},
\[\frac{k-1/2}{\lambda}=\paren{cn-1/2}\cdot2\cdot2^{c_3n\sqrt{\eps}}\stackrel{n\to\infty}{\to}\infty.\]
Hence
\begin{align}
W\paren{\frac{k-1/2}{\lambda}}\asymp&\ln\paren{\paren{cn-1/2}\cdot2} + c_3n\sqrt{\eps}\ln 2 + \ln\paren{\ln\paren{\paren{cn-1/2}\cdot2} + c_3n\sqrt{\eps}\ln 2 }\notag\\
=&\ln2\cdot c_3\sqrt{\eps}\cdot n(1+o(1)).\notag
\end{align}
We thus have
\[j^*\asymp\frac{c}{\ln2\cdot c_3\sqrt{\eps}}.\]
Plug this into $f$,
\begin{align}
f(j^*)=&-j^*\ln j^*+\paren{\ln\lambda+1}j^*+\paren{cn-1/2}\ln j^*-\frac{1}{2}\ln\paren{2\pi}\notag\\
=&-\frac{c}{\ln2\cdot c_3\sqrt{\eps}}\ln\paren{\frac{c}{\ln2\cdot c_3\sqrt{\eps}}}+\paren{-\ln2\cdot c_3\sqrt{\eps}\cdot n+1}\frac{c}{\ln2\cdot c_3\sqrt{\eps}}\notag\\
&+\paren{cn-1/2}\ln\paren{\frac{c}{\ln2\cdot c_3\sqrt{\eps}}}-\frac{1}{2}\ln\paren{2\pi}\notag\\
=& -n\paren{1+o(1)}\paren{c-c\ln\frac{c}{\ln2\cdot c_3\sqrt{\eps}}}. \notag
\end{align}
Finally,  the exponent of expression~\eqref{eqn:mmt_bd} is
\begin{align}
D+c\ln\frac{L}{2}-\ln\frac{c_2}{\delta}\approx&c-c\ln\frac{c}{\ln2\cdot c_3\sqrt{\eps}}+c\ln\frac{L}{2}-\ln\frac{c_2}{\delta}\notag\\
=&c\ln L-\paren{c+1}\ln\frac{1}{\delta}+c-c\ln\frac{2c}{\ln2\cdot c_3\sqrt{c_1}}-\ln c_2.\notag
\end{align}
In order for it to be positive as $\delta\to0$, we had better set $L=\paren{1/\delta}^a$, where $ac>c+1$, i.e., $a>1+1/c$. If we only assume the first $k=cn<n$ moments are Poissonian for some $c=\frac{1}{1+\gamma}<1$ where $\gamma>0$ is a some small positive constant, then we need to take $a>2+\gamma$.

Using arguments similar to those in the previous section, we get:
\begin{lemma}
	If Conjecture~\ref{conj:poissonianity} is true, then there exists a lattice $ \Lf $ such that $ \Lf\cap\cB(0,\sqrt{nP}) $ has rate $ C(P,N)-\delta $ and is $ \left(P,N ,\cO(1/\delta^{1+1/c})\right)$-list decodable.
	\label{lemma:poissonmoment_listsize}
\end{lemma}

\subsection{Remark}
Careful readers might have observed that in order to prove Lemma~\ref{lemma:poissonmoment_listsize}, we do not really need the first $ cn $ moments to be Poisson. It suffices to show that the first  $ cn $ moments are bounded from above by a quantity that is subexponential in $ n $. However, we are optimistic that a result similar to Conjecture~\ref{conj:poissonianity} can indeed be proved for the Haar distribution on $ \sL_n $.

\section{Concluding remarks and open problems}\label{sec:rk_open_prob}
In this paper we initiate a systematic study of the list size problem for codes over $\bR$. In particular, upper bounds on list sizes of nested Construction-A lattice codes and infinite Construction-A lattices are exhibited. Similar upper bounds are also obtained for an ensemble of regular infinite constellations. Matching lower bounds for such an ensemble are provided. Other coding-theoretic properties are studied by the way. Our lower bound for random spherical codes also matches the upper bound in previous work. A caveat is that all of our bounds are concerned with \emph{typical} scaling of the list sizes of \emph{random} codes sampled from the ensembles of interest. The extremal list sizes \emph{may} be smaller than our lower bounds. We conclude the paper with several open questions.

\begin{enumerate}
	\item  Careful readers might have already noted that a missing piece in this work is a list size lower bound for random Construction-A lattice codes. We had trouble replicating the arguments in~\cite{guruswami2013combinatorial}. We leave it as an open question to get a $\poly(1/\delta)$ list size lower bound.
	\item Can one sample efficiently from the Haar distribution on the spaces of our interest? In particular, can one sample efficiently a generator matrix $\bfG$ from the Haar distribution $\mu$ on $\speclin(n,\bR)$? Can one sample eficiently a lattice $\Lambda$ from the Haar distribution $\mu$ on $\speclin(n,\bR)/\speclin(n,\bZ)$? To this end, we can think of $\speclin(n,\bR)$ as a codimensional-one hypersurface in $\bR^{n^2}$ cut off by the equation $\det(\bfG)=1$. Readers from the Monte Carlo Markov Chain (MCMC) community may be interested in such problems.
	\item A very intriguing question which we are unable to resolve in this work is to bring down the exponential list size of random Construction-A lattice codes. We do not believe that our upper bound is tight. A starting step towards this goal is probably to  obtain an averaging formula custom tailored for Construction-A lattices. Indeed, Loeliger~\cite{loeliger-1997-avg-formula-constr-a-lin-codes} has proved a first-order averaging formula for (appropriately scaled) Construction-A lattices as an analog of Siegel's formula for Haar lattices. Specifically, consider an ensemble of Construction-A-type lattices $\Lambda\coloneqq\frac{1}{\alpha}(\cC+q\bZ^n)$ where $\cC\sim\gr(\kappa,\bF_q^n)$ is a uniformly random $\kappa$-dimensional subspace of $\bF_q^n$. Then for any bounded measurable   compactly supported function $\rho\colon\bR^n\to\bR$,  it holds that
	\begin{equation*}
	\exptover{\cC\sim\sC_{n,\kappa}}{\sum_{\vx\in\Lambda\setminus\curbrkt{0}}\rho(\vx)}=\frac{1}{\card{\sC_{n,\kappa}}}\sum_{\cC\in\sC_{n,\kappa}}\sum_{\vx\in\Lambda\setminus\curbrkt{0}}\rho(\vx)\xrightarrow{\alpha\to\infty,q/\alpha\to\infty}\det(\Lambda)^{-1}\int_{\bR^n}\rho(\vx)\diff\vx,
	\end{equation*}
	where $\sC_{n,\kappa}\coloneqq\gr(\kappa,\bF_q^n)$ and the covolume
	\[\det(\Lambda)=\paren{\frac{1}{\alpha}}^n[\bZ^n\colon(\cC+q\bZ^n)]=\card{\bZ^n/(\cC+q\bZ^n)}/\alpha^{n}=q^{n-\kappa}/\alpha^n,\]
	is kept fixed. 
	Can one lift Loeliger's formula to $k$-variate functions $\rho\colon(\bR^n)^k\to\bR$ and get a higher-order averaging formula for Construction-A lattices as an analog of Rogers's formula for Haar lattices?
	\item Can one compute similar moments for random Construction-A lattices? Given a random Construction-A lattice $\Lambda=q^{-1}\cC+\bZ^n$ where $\cC$ is a $\kappa$-dimensional random linear code in $\bF_q^n$, compute the $k$-th moment $\expt{\card{\Lambda\cap\cB_2^n(\vy,\sqrt{nN})}^k}$ for any $\vy\in\bR^n$ and for $k$ as large as possible.  For \emph{random binary linear code} over $\bF_2^n$ of rate $1-H(p)+\delta$,\footnote{Note that such a code operates at a rate \emph{above} capacity and the corresponding moments they are interested in are exponentially large. Indeed, they instead consider \emph{centered} moments $\expt{(X-\expt{X})^k}$.} Linial and Mosheiff~\cite{linial-mosheiff-2018-wt-dist-rand-bin-lin} recently managed to \emph{characterize} the first $\cO(n/\log n)$ moments of the number of codewords in a Hamming \emph{sphere} of radius $np$. It turns out that the normalized centered moment \[\frac{\expt{(\card{\cC\cap\sham(0,np)}-\expt{\card{\cC\cap\sham(0,np)}})^k}}{\var{\card{\cC\cap\sham(0,np)}}^{k/2}}\] behaves like the moment of a Gaussian (recall Fact~\ref{fact:gaussian_mmt}) up to some threshold $k<k_0$, where $k_0$ is 3 or 4 for $\delta$ not too small. From $k_0$ on, linearity quickly kicks in and dominates the behaviour of the moments. 
	\item We showed that  Haar lattices of rate $\frac{1}{2}\log\frac{P}{N}-\delta$ are $(P,N,\poly(1/\delta))$ whp conditioned on Conjecture~\ref{conj:poissonianity}. Can one show other coding-theoretic goodness properties under the conjecture?
	It is known that Haar lattices are good for packing \cite{shlosman-tsfasman-2000-random-packing}. 
	Are they also good for covering, AWGN, quantization, etc.?
	\item In this paper, the list decodability of two ensembles (Construction-A and Haar) of lattices are considered. The ultimate goal is to find an \emph{explicit} $(P,N,\poly(1/\delta))$-list decodable lattice code of rate $\frac{1}{2}\log\frac{P}{N}-\delta$. Recently Kaufman and Mass~\cite{kaufman-mass-2018-lattices-from-hdx} constructed explicit lattices of good \emph{distance}  from high dimensional expanders. 
	However, there is no explicit bound on the covolume of the lattice. 
	Therefore, it is unclear whether their construction is good for packing. 
	Moreover, their results are conditioned on the  conjecture that the cohomology group of Ramanujan complexes with \emph{integer} coefficients is large. 
	\item  Our lower bounds on list sizes only indicate  typical behaviours of ensembles of random lattices. This does not exclude the existence of codes with smaller list sizes. Can one prove a lower bound on list sizes of \emph{general} codes over reals?
	Namely, for any  $2^{nR}$ points on $\cS^{n-1}\paren{0,\sqrt{nP}}$, how large can $L$ be such that
	one can always find a position $\vy$ to which there are at least $L$ points that are  $\sqrt{nN}$-close?

\end{enumerate}

\section*{Acknowledgement}
YZ thanks Noah Stephens-Davidowitz for sharing his expertise on lattices, in particular, introducing him the Poisson heuristics as a prediction of the behaviors of random lattices when he was visiting MIT and for exchanging multiple informative emails afterwards. YZ thanks Mary Wootters for clarifying the state of the art of list decodability of random linear codes. YZ also wants to thank Boris Bukh, Chris Cox, Sidharth Jaggi, Nicolas Resch and Tomasz Tkocz for several inspiring discussions, respectively. Part of this work was done when YZ was visiting CMU under the mentorship of Venkatesan Guruswami who listened to the progress and provided generous encouragement at the early stage. Part of this work was done when SV was a postdoc at CUHK.

\appendices

\section{Table of notation}\label{sec:tab_notation}
\begin{center}
	\begin{longtable}{|p{0.091\textwidth}|p{0.168\textwidth}|p{0.25\textwidth}|p{0.414\textwidth}|}
		\hline
		\textbf{Symbol} & \textbf{Section} & \textbf{Description} & \textbf{Definition/Value/Range} \\ \hline
		
		$A_{n-1}$  & Throughout the paper & Area of an $(n-1)$-dimensioinal unit sphere & $A_{n-1}\coloneqq\area(\cS_2^{n-1})$ \\\hline
		$\cA$ & Sec.~\ref{sec:ic},~\ref{sec:regular_ic} & Cube of side length $\alpha$ & $\cA\coloneqq[-\alpha/2,\alpha/2)^n$ \\\hline
		\multirow{3}{0.091\textwidth}{$C$} & \multirow{2}{0.168\textwidth}{Sec.~\ref{sec:finitefield_prior_work}} & List decoding capacity & $C\coloneqq1-H_q(p)\in[0,1]$ \\\cline{3-4}
		 &  & List decoding capacity & $C\coloneqq1-p\in[0,1]$ \\\cline{2-4}
		 & Throughout  the paper & List decoding capacity & $C\coloneqq\frac{1}{2}\log\frac{P}{N}\in\bR_{\ge0}$ \\\hline
		\multirow{3}{0.091\textwidth}{$\cC$} & Sec.~\ref{sec:finitefield_prior_work} & Code & $\cC\in\binom{\bF_q^n}{q^{nR}}$   \\\cline{2-4}
		 & Throughout the paper & Code & $\cC\subset\bR^n$ of size $2^{nR}$ \\\cline{2-4}
		 & Sec.~\ref{sec:results},~\ref{sec:ic},~\ref{sec:regular_ic} & IC & $\cC\subset\bR^n$ \\\hline
		$C(L)$  & Sec.~\ref{sec:finitefield_prior_work} &  List-$L$ capacity & See Eqn.~\eqref{eqn:list_l_cap} \\\hline
		$C_{\rand}(\cW)$ & Sec.~\ref{sec:finitefield_prior_work} & Random code capacity   & $C_{\text{rand}}(\cW)\coloneqq\max_P\min_{P_{\bfx\bfs\bfy}=P_\bfx P_\bfs W_{\bfy|\bfx\bfs}\colon P_\bfx=P}I(\bfx;\bfy)$ \\\hline
		\multirow{2}{0.091\textwidth}{$\bfG$} & \multirow{2}{0.168\textwidth}{Throughout the paper} &  Generator matrix of a linear code   & $\bfG\in\bF_q^{n\times \kappa}$ \\\cline{3-4}
		 &  & Generator matrix of a lattice & $\bfG\in\bR^{n\times \kappa}$ \\\hline
		\multirow{2}{0.091\textwidth}{$k$} & Sec.~\ref{sec:haarlattice_introduction},~\ref{sec:ld_haar} & Degree of moments & $k=cn$ \\\cline{2-4}
		 & Sec.~\ref{sec:rk_open_prob} & Degree of moments & $k=\cO(n/\log n)$ \\\hline
		$\ell$ & Sec.~\ref{sec:ld_constr_a} &  Log of list size  & $\ell\coloneqq\log_q(L+1)$ \\\hline
		$L$ & Throughout the paper & List size & $ L\in[q^{nR}]$   \\\hline
		$\cL$ & Throughout the paper & List & $\cL\in\binom{\cC}{\le L}$   \\\hline
		
		$\sL_n$ & Sec.~\ref{sec:haarlattice_introduction},~\ref{sec:ld_haar} & Space of determinant-1 lattices  & $\sL_n\coloneqq\curbrkt{\Lambda\le\bR^n\text{ lattice}\colon \det(\Lambda)=1}$  \\\hline
		$m$ & Sec.~\ref{sec:lb_ls_spherical},~\ref{sec:ld_constr_a} & Message & $m\in[q^{nR}]$   \\\hline
		\multirow{2}{0.091\textwidth}{$M$} & Sec.~\ref{sec:lb_ls_spherical},~\ref{sec:ld_constr_a} &  Number of messages/size of codebook  & $M\coloneqq|\cM|=q^{nR}$ \\\cline{2-4}
		& Sec.~\ref{sec:finitefield_prior_work} &  Symmetrizability  & See Eqn.~\eqref{eqn:sym} \\\hline
		
		$\cM$ & Sec.~\ref{sec:lb_ls_spherical},~\ref{sec:ld_constr_a} & Set of messages & $\cM\coloneqq\curbrkt{0,1,\cdots,2^{nR}-1}$  \\\hline
		$n$ & Throughout the paper &  Blocklength & $n\in\bZ_{>0}$   \\\hline
		$N$ & Throughout the paper & Adversary's power constraint & $N\in\bR_{>0}$   \\\hline
		$p$ & Sec.~\ref{sec:finitefield_prior_work} & Adversary's power constraint & $p\in[0,1]$   \\\hline
		$P$ & Throughout the paper & Transmit power constraint & $ P\in\bR_{>0} $   \\\hline
		$\cP(\Lambda)$ &  Sec.~\ref{sec:ld_constr_a} & Fundamental parallelepiped & $\cP(\Lf)\coloneqq \{ \bfG\vx:\vx\in [0,1 )^n \}$  \\\hline
		
		$q$ & Throughout the paper &  Characteristic of finite field & Prime number  \\\hline
		$Q_\Lambda(\cdot)$ & Sec.~\ref{sec:ld_constr_a}  & Lattice quantizer & See Eqn.~\eqref{eqn:lattice_quant}  \\\hline
		
        $\rcov$ & Sec.~\ref{sec:ld_constr_a} & Covering radius of a lattice & See  Appendix~\ref{sec:primer_lattices}  \\\hline
		
		\multirow{2}{0.091\textwidth}{$\reff$} & Sec.~\ref{sec:ld_constr_a} & Effective radius of a lattice & See Appendix~\ref{sec:primer_lattices}  \\\cline{2-4}
		& Sec.~\ref{sec:results},~\ref{sec:ic},~\ref{sec:regular_ic} & Effective radius of an infinite constellation & See Def.~\ref{def:ic} \\\hline
		$\rpack$ & Sec.~\ref{sec:ld_constr_a} & Packing radius of a lattice & See Appendix~\ref{sec:primer_lattices}  \\\hline
		$R$ & Throughout the paper & Rate of a code & $\frac{\log|\cC|}{n}\in\bR_{>0} $  \\ \hline
		
		$\vs$ & Throughout the paper  & Jamming vector & $\vs\in\cB(0,\sqrt{nN})$ \\\hline
		$V(\cC)$ & Sec.~\ref{sec:results},~\ref{sec:ic},~\ref{sec:regular_ic} & Effective volume of an IC & $V(\cC)=1/\Delta(\cC)$ \\\hline
		$V_n$   & Throughout the paper & Volume of an $n$-dimensioinal unit ball & $V_n\coloneqq\vol(\cB_2^n)$ \\\hline
		$\cV(\Lambda)$ & Sec.~\ref{sec:ld_constr_a},~\ref{sec:ic}  & Fundamental Voronoi region & $\cV(\Lf) \coloneqq \{ \vx\in\bR^n: Q_{\Lf}(\vx) = 0 \}$ \\\hline
		$W$ & Sec.~\ref{sec:lb_ls_spherical},~\ref{sec:regular_ic} & Witness of list decodability & See Eqn.~\eqref{eqn:w_spherical},~\eqref{eqn:w_ic}   \\\hline
		
		$W(\cdot|\cdot,\cdot)$ & Sec.~\ref{sec:finitefield_prior_work}  & Transition probability of an AVC & $W\colon\cY\times\cX\times\cS\to[0,1]$ \\\hline
		$\cW$ & Sec.~\ref{sec:finitefield_prior_work} & AVC & $\cW\coloneqq\curbrkt{W(\cdot|\cdot,s),\vs\in\cS}$    \\\hline
		$\vx$ & Throughout the paper & Transmitted codeword  & $\vbfx\in\cC$  \\\hline
		$X$ & Sec.~\ref{sec:ld_haar} & Number of lattice points in a ball & See Eqn.~\eqref{eqn:x_def}   \\\hline
		$\vy$ & Throughout the paper &  Received word & $\vy=\vx+\vs\in\cB(0,\sqrt{nP}+\sqrt{nN})\setminus\cB(0,\sqrt{nP}-\sqrt{nN})$ \\\hline
		$Y$ & Sec.~\ref{sec:haarlattice_introduction},~\ref{sec:ld_haar} & Poisson random variable &  $Y\sim\pois(V/2)$ \\\hline
		$\cY$ & Throughout the paper & Net for $\vy$'s & See specific definitions  \\\hline
		$\vbfz$ & Sec.~\ref{sec:regular_ic} & AWGN  & $\bR^n\ni\vbfz\sim\cN(0,\sigma^2\bfI)$ \\\hline
		
		$\alpha$  & Throughout the paper & Side length of $\cA$  & $\alpha\in\bR_{>0}$ \\\hline
		\multirow{2}{0.091\textwidth}{$\delta$} & Throughout the paper &  Gap to capacity & $\delta\coloneqq C-R\in\bR_{>0} $   \\\cline{2-4}
		& Sec.~\ref{sec:results},~\ref{sec:ic},~\ref{sec:regular_ic} & Gap between $\reff(\cC)$ and $\sqrt{nN}$ & $\delta\coloneqq\log\frac{\reff(\cC)}{\sqrt{nN}}$ \\\hline
		$\Delta(\cC)$ & Sec.~\ref{sec:results},~\ref{sec:ic},~\ref{sec:regular_ic}  & Density of an IC  & $\Delta(\cC)\coloneqq \limsup_{a\to\infty}\frac{|\cC\cap [0,a]^n|}{a^n}$ \\\hline
		
		$\eps$ & Throughout the paper & Parameter of a net & See specific definitions \\\hline
		$\bfTheta$  & Sec.~\ref{sec:haarlattice_introduction}  & Rogers's ensemble & $\bfTheta=\bfTheta(\theta_1,\cdots,\theta_{n-1},\omega)$ (Eqn.~\eqref{eqn:rogers_ensemble})  \\\hline
		$\kappa$  & Throughout the paper & Dimension of a linear code or a lattice & $\kappa\in\curbrkt{0,1,\cdots,n}$ \\\hline
		$\Lf$ & Throughout the paper & Lattice & $ \Lf\le \bR^n $   \\\hline
		$\mu$ & Sec.~\ref{sec:haarlattice_introduction},~\ref{sec:ld_haar}  & Haar measure on $\speclin(n,\bR)$, $\sL_n$ or $\genlin(n,\bR)$  & See Theorem~\ref{thm:haar_slnr} \\\hline
		$\tau$ & Sec.~\ref{sec:finitefield_prior_work} &  Gap to list decoding radius & $\tau\coloneqq 1-1/q-p\in\bR_{>0} $   \\\hline
		$\Phi$ & Sec.~\ref{sec:ld_constr_a} &  Natural embedding   & $\Phi\colon \bF_q\to\bZ$ \\\hline
		$\psi$ & Sec.~\ref{sec:lb_ls_spherical},~\ref{sec:ld_constr_a},~\ref{sec:regular_ic} & Encoding function   & $\psi\colon\cM\to\cC$ \\\hline
		
		$[\cdot]\mod\cA$ & Sec.~\ref{sec:ic},~\ref{sec:regular_ic} & Quantization error wrt $\alpha\bZ^n$ & $[\cdot]\mod\cA\coloneqq\cdot\mod\alpha\bZ^n$ \\\hline
		$[\cdot]\mod\Lambda$ & Sec.~\ref{sec:ld_constr_a}  & Lattice quantization error & $[\cdot]\mod\Lambda\coloneqq\cdot- Q_\Lambda(\cdot)$ \\\hline
		
		$(\cdot)^*$ & Sec.~\ref{sec:ic},~\ref{sec:regular_ic}  & Set modulo $\alpha\bZ^n$ & $(\cdot)^*\coloneqq\cdot\mod\cA$ \\\hline
		
	\end{longtable}
\end{center}

\section{Proof of Eqn.~\eqref{eqn:expt-toshow} and~\eqref{eqn:var-toshow}}\label{sec:calc_prop_lbound_randomspherical}

\subsection{A covering lemma}
\label{sec:cov-lemma}
Before proving Eqn.~\eqref{eqn:expt-toshow} and~\eqref{eqn:var-toshow}, we need the following lemma. 
It guarantees the existence of a covering of a sphere which is sufficiently spread out in the sense that the fraction of points in any spherical cap does not deviate much from the corresponding volume ratio. 
\begin{lemma}
\label{lem:covering-sphere}
Let $ r>0 $ and $ \eps>0 $ sufficiently small. 
There exists a subset $ \cY $ of the sphere $ \cS^{n-1}(0,\sqrt{nr}) $ such that
\begin{enumerate}
	\item \label{itm:cov-1} for every $ \vy\in\cS^{n-1}(0,\sqrt{nr}) $, there exists $ \vy'\in\cY $ with $ \|\vy - \vy'\|\le\sqrt{n\eps} $;
	\item \label{itm:cov-2} $ |\cY|=(c/\sqrt{\eps})^{n(1+o(1))} $ for some $c>0$ that is independent of $n$ and $ \eps $ but depends on $ r $;
	\item \label{itm:cov-3} for every $ \vy\in\cS^{n-1}(0,\sqrt{nr}) $ and every $ 0<\rho<r $, 
	\begin{align}
	\frac{\card{\cY\cap\C^{n-1}(\vy,\sqrt{n\rho})}}{|\cY|} &\ge \frac{1}{2}\frac{\area\paren{\C^{n-1}(\vz,\sqrt{n}(\sqrt{\rho} - \sqrt{\eps_\ell}))}}{\area\paren{\cS^{n-1}(\sqrt{nr})}}  , \notag \\
	\frac{\card{\cY\cap\C^{n-1}(\vy,\sqrt{n\rho})}}{|\cY|} &\le \frac{3}{2}\frac{\area\paren{\C^{n-1}(\vz,\sqrt{n}(\sqrt{\rho} + \sqrt{\eps_u}))}}{\area\paren{\cS^{n-1}(\sqrt{nr})}} , \notag 
	\end{align}
	where
	\begin{align}
	\eps_\ell &\coloneqq \paren{\frac{1}{2\sqrt{r}} + \frac{3}{2}}^2{\eps}, \quad 
	\eps_u \coloneqq \frac{9}{4}{\eps}. \notag 
	\end{align}
\end{enumerate}
\end{lemma}

\begin{proof}
Let $ \cY $ be a set of $ M = (1+o(1))\sqrt{2\pi n}\paren{4\sqrt{\frac{r}{\eps}}}^{n-1} $ points $ \vbfy_1,\cdots,\vbfy_M $ each independent and uniformly distributed on $ \cS^{n-1}(0,\sqrt{nr}) $. 
Note that $ M = (c/\sqrt{\eps})^{n+o(n)} $ for some $ c $ independent of $ n $ and $ \eps $ (but dependent on $ r $), which satisfies property~\ref{itm:cov-2}. 
We will show that such a $\cY$ satisfies all properties in Lemma~\ref{lem:covering-sphere} with high probability. 

By a standard volume argument, there exists a $ \sqrt{n\eps_1} $-net $\cZ$ of $ \cS^{n-1}(0,\sqrt{n\eps_1}) $ satisfying properties~\ref{itm:cov-1} and~\ref{itm:cov-2} with $ \eps $ replaced with $ \eps_1=\eps/4 $ (and the constant $c$ needs to be adjusted accordingly). 
\begin{align}
& \prob{\exists \vy\in\cS^{n-1}(0,\sqrt{nr}),\;\forall i\in[M],\;\|\vy - \vbfy_i\|>\sqrt{n\eps}} \notag \\
&\le \prob{\exists \vz\in\cZ,\;\forall i\in[M],\;\|\vz - \vbfy_i\|>\sqrt{n\eps} - \sqrt{n\eps_1}} \notag \\
&\le \sum_{\vz\in\cZ} \prod_{i\in[M]} \prob{\|\vz - \vbfy_i\|>\sqrt{n\eps/4}} \notag \\
&\le \paren{\frac{c}{\sqrt{\eps_1}}}^{n+o(n)} \paren{1 - \frac{\area\paren{\C^{n-1}(\sqrt{n\eps/8})}}{\area\paren{\cS^{n-1}(\sqrt{nr})}}}^M \label{eqn:rad-geom} \\
&\le \paren{\frac{c}{\sqrt{\eps_1}}}^{n+o(n)} \paren{1 - \frac{\vol\paren{\cB^{n-1}(\sqrt{n\eps/8})}}{\area\paren{\cS^{n-1}(\sqrt{nr})}}}^M \label{eqn:bound-cap-by-ball} \\
&= \paren{\frac{c}{\sqrt{\eps_1}}}^{n+o(n)} \paren{1 - \frac{V_{n-1}}{A_{n-1}}\sqrt{\frac{\eps}{8r}}^{n-1}}^M \notag \\
&= \paren{\frac{c}{\sqrt{\eps_1}}}^{n+o(n)} \paren{1 - \frac{1+o(1)}{\sqrt{2\pi n}}\sqrt{\frac{\eps}{8r}}^{n-1}}^M \label{eqn:area-vol-ratio} \\
&= \paren{\frac{c}{\sqrt{\eps_1}}}^{n+o(n)} \paren{1 - \frac{1+o(1)}{\sqrt{2\pi n}}\sqrt{\frac{\eps}{8r}}^{n-1}}^{{(1+o(1))}{\sqrt{2\pi n}}\sqrt{\frac{8r}{\eps}}^{n-1}\sqrt{2}^{n-1}} \label{eqn:choice-m} \\
&\le \paren{\frac{c}{\sqrt{\eps_1}}}^{n+o(n)}e^{-\sqrt{2}^{n-1}}. \label{eqn:ineq-e}
\end{align}
Eqn.~\eqref{eqn:rad-geom} follows since the set $ \curbrkt{\vy\in\cS^{n-1}(\sqrt{nr}):\|\vz - \vy\|\le\sqrt{n\eps/4}} $ forms a cap of radius $ \sqrt{n\eps'} $ where $ \eps' $ can be determined by inspecting the geometry. 
Specifically, 
$\sqrt{\eps'} = \sqrt{r}\sin\theta$, 
where $ \theta $ satisfies
\begin{align}
\cos\theta = \frac{r+r - \eps/4}{2r} = 1 - \frac{\eps}{8r}. \notag 
\end{align}
Therefore, 
\begin{align}
\eps' = \sqrt{r}\sqrt{1 - \paren{1 - \frac{\eps}{8r}}^2} 
\ge \sqrt{r}\sqrt{\frac{\eps}{8r}} 
= \sqrt{\eps/8}, \notag 
\end{align}
where the inequality follows since $ 1 - (1 - x)^2 \ge x $ for any $ 0\le x\le 1 $. 
Eqn.~\eqref{eqn:bound-cap-by-ball} follows since $ \area\paren{\C^{n-1}(\gamma)}\ge\vol\paren{\cB^{n-1}(\gamma)} $ for any $ \gamma>0 $.
In Eqn.~\eqref{eqn:area-vol-ratio}, the ratio $ V_{n-1}/A_{n-1} $ is given by
\begin{align}
\frac{V_{n-1}}{A_{n-1}}
&= \frac{\frac{1}{\sqrt{\pi(n-1)}}\paren{\frac{2\pi e}{n-1}}^{\frac{n-1}{2}}}{\sqrt{\frac{n}{\pi}}\paren{\frac{2\pi e}{n}}^{\frac{n}{2}}}(1+o(1))
= \frac{1}{\sqrt{2\pi en}}\paren{\frac{n}{n-1}}^{n/2} (1+o(1))
\to \frac{1}{\sqrt{2\pi n}}(1+o(1)). \notag 
\end{align}
Eqn.~\eqref{eqn:choice-m} is by the choice of $M$. 
Eqn.~\eqref{eqn:ineq-e} follows from the inequality $ (1-1/x)^x\le 1/e $ for $ x\ge1 $. 
Therefore, property~\ref{itm:cov-1} holds with probability at least $ 1 - e^{-e^{\Omega(n)}} $.

To show property~\ref{itm:cov-3}, we quantize the interval $ [0,r] $ using an $ \eps_2 $-net $ \{0,\eps_2,2\eps_2,\cdots,\floor{r/\eps_2}\eps_2\} $. 
Such a net satisfies that for any $ \rho\in[0,r] $, there exists $ \tau\in\cI $ with $ |\tau - \rho|\le\eps_2 $. 
Let
\begin{align}
\gamma_\ell &\coloneqq \frac{\area\paren{\C^{n-1}(\sqrt{n}(\sqrt{\rho} - \sqrt{\eps_\ell}))}}{\area\paren{\cS^{n-1}(\sqrt{nr})}}, \quad 
\gamma_u \coloneqq \frac{\area\paren{\C^{n-1}(\sqrt{n}(\sqrt{\rho} + \sqrt{\eps_u}))}}{\area\paren{\cS^{n-1}(\sqrt{nr})}}. \notag 
\end{align}
We then bound the probability that property~\ref{itm:cov-3} is violated. 
\begin{align}
& \prob{\exists\vy\in\cS^{n-1}(0,\sqrt{nr}),\;\exists\rho\in[0,r],\; \frac{\card{\cY\cap\C^{n-1}(\vy, \sqrt{n\rho})}}{M} \notin \left[\frac{1}{2}\gamma_\ell,\frac{3}{2}\gamma_u\right]} \label{eqn:tobound-cov3} \\
&\le \prob{\exists \vz\in\cZ,\;\exists\tau\in\cI,\;
\|\vz - \vy\|\le\sqrt{n\eps_1},\;
|\tau - \rho|\le\eps_2,\;
\frac{\card{\cY\cap\C^{n-1}(\vz,\sqrt{n\tau})}}{M} \notin \left[\frac{1}{2}\gamma_\ell,\frac{3}{2}\gamma_u\right]} \notag \\
&\le \sum_{\vz\in\cZ}\sum_{\tau\in\cI} \left(\prob{\frac{\card{\cY\cap\C^{n-1}(\vz,\sqrt{n\tau_\ell})}}{M} < \frac{1}{2}\gamma_\ell} 
+ \prob{\frac{\card{\cY\cap\C^{n-1}(\vz,\sqrt{n\tau_u})}}{M} > \frac{3}{2}\gamma_u}\right). \label{eqn:errorbound-cov3} 
\end{align}
The radii $ \tau_\ell $ and $ \tau_u $ are such that
\begin{align}
\C^{n-1}(\vz,\sqrt{n\tau_\ell})
\subset \C^{n-1}(\vy,\sqrt{n\rho})
\subset \C^{n-1}(\vz,\sqrt{n\tau_u}). \notag 
\end{align}
In what follows, we derive lower and upper bounds on $ \tau_\ell $ and $ \tau_u $, respectively. 
The geometry is depicted in Fig.~\ref{fig:geom-cov}. 
\begin{figure}[htbp]
	\centering
	\begin{subfigure}[t]{0.45\textwidth}
		\centering
		\includegraphics[width=0.95\linewidth]{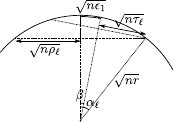}
		\caption{}
		\label{fig:geom-cov1}
	\end{subfigure}
	~
	\begin{subfigure}[t]{0.45\textwidth}
		\centering
		\includegraphics[width=0.95\linewidth]{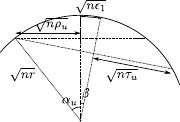}
		\caption{}
		\label{fig:geom-cov2}
	\end{subfigure}
	\caption{We consider a cap of radius $ \sqrt{n\rho} $ centered at $ \vy $ on the sphere $ \cS^{n-1}(\sqrt{nr}) $. Due to quantization error, the center of the cap is distorted to $ \vz $ which is at most $ \sqrt{n\eps_1} $ from $ \vy $ and the radius of the cap is distorted by $ \eps_2 $. 
	Let $ \rho_\ell \coloneqq \rho - \eps_1 $ and $ \rho_u \coloneqq \rho + \eps_2 $. 
	The radii $ \tau_\ell $ and $ \tau_u $ are such that the original cap is sandwiched between two tilted caps: $ \C^{n-1}(\vz,\sqrt{n\tau_\ell})\subset\C^{n-1}(\vy,\sqrt{n\rho})\subset\C^{n-1}(\vz,\sqrt{n\tau_u}) $. }
	\label{fig:geom-cov}
\end{figure}

For notational convenience, let $ \rho_\ell \coloneqq \rho - \eps_2 $ and $ \rho_u \coloneqq \rho + \eps_2 $. 
Define $ \alpha_\ell,\alpha_u,\beta $ be defined as:
\begin{align}
\sin\alpha_\ell = \sqrt{\frac{\rho_\ell}{r}}, \quad 
\sin\alpha_u = \sqrt{\frac{\rho_u}{r}}, \quad 
\cos\beta = \frac{r + r - \eps_1}{2r} = 1 - \frac{\eps_1}{2r}. \notag 
\end{align}
For $ \tau_\ell $, we have
\begin{align}
\sqrt{\tau_\ell} 
&= \sqrt{r}\sin(\alpha_\ell - \beta) \notag \\
&= \sqrt{r} (\sin\alpha_\ell\cos\beta - \cos\alpha_\ell\sin\beta) \notag \\
&= \sqrt{r}\paren{\sqrt{\frac{\rho_\ell}{r}}\paren{1 - \frac{\eps_1}{2r}} - \sqrt{1 - \frac{\rho_\ell}{r}}\sqrt{1 - \paren{1 - \frac{\eps_1}{2r}}^2}} \notag \\
&\ge \sqrt{\rho_\ell}\paren{1 - \frac{\eps_1}{2r}} - \sqrt{r - \rho_\ell}\sqrt{\frac{\eps_1}{r}} \label{eqn:elem-ineq} \\
&\ge \sqrt{\rho - \eps_2}\left(1 - \frac{\eps_1}{2r}\right) - \sqrt{\eps_1} \label{eqn:rho-lessthan-r} \\
&\ge \paren{\sqrt{\rho} - \sqrt{\eps_2}}\paren{1 - \frac{\eps_1}{2r}} - \sqrt{\eps_1} \label{eqn:sqrt-ineq} \\
&\ge \sqrt{\rho} - \frac{\sqrt{\rho}\eps_1}{2r} - \sqrt{\eps_2} - \sqrt{\eps_1} \label{eqn:drop-positive} \\
&\ge \sqrt{\rho} - \paren{\frac{1}{2\sqrt{r}} + 1}\sqrt{\eps_1} - \sqrt{\eps_2}. \label{eqn:rho-lessthan-r-again} 
\end{align}
Eqn.~\eqref{eqn:elem-ineq} follows from the elementary inequality: $ 1 - (1 - x)^2\le2x $ for any $ x\ge0 $. 
Eqn.~\eqref{eqn:rho-lessthan-r} is by the assumption $ \rho< r $. 
Eqn.~\eqref{eqn:sqrt-ineq} follows from the fact that $ \sqrt{x - y}\ge\sqrt{x} - \sqrt{y} $ for any $ x\ge y\ge0 $. 
In Eqn.~\eqref{eqn:drop-positive}, we drop the term $ \frac{\sqrt{\eps_2}\eps_1}{2r} $. 
Eqn.~\eqref{eqn:rho-lessthan-r-again} follows since $ \rho< r $ and $ \eps_1\le\sqrt{\eps_1} $ for $ 0\le\eps_1\le1 $. 

For $ \tau_u $, we have
\begin{align}
\sqrt{\tau_u} 
&= \sqrt{r}\sin(\alpha_u + \beta) \notag \\
&= \sqrt{r}(\sin\alpha_u\cos\beta + \cos\alpha_u\sin\beta) \notag \\
&= \sqrt{r}\paren{\sqrt{\frac{\rho_u}{r}}\paren{1 - \frac{\eps_1}{2r}} + \sqrt{1 - \frac{\rho_u}{r}}\sqrt{1 - \paren{1 - \frac{\eps_1}{2r}}^2}} \notag \\
&\le \sqrt{\rho_u} + \sqrt{r}\sqrt{\frac{\eps_1}{r}} \notag \\
&= \sqrt{\rho + \eps_2} + \sqrt{\eps_1} \notag \\
&\le \sqrt{\rho} + \sqrt{\eps_2} + \sqrt{\eps_1}. \label{eqn:sqrt-ineq-another} 
\end{align}
Eqn.~\eqref{eqn:sqrt-ineq-another} follows from the elementary inequality: $ \sqrt{x + y}\le\sqrt{x} + \sqrt{y} $ for any $ x,y\ge0 $.

Set $ \eps_2 = \eps $ and 
\begin{align}
\sqrt{\eps_\ell} &\coloneqq \paren{\frac{1}{2\sqrt{r}} + 1}\sqrt{\eps_1} + \sqrt{\eps_2}
= \paren{\frac{1}{2\sqrt{r}} + 1}\sqrt{\eps} + \sqrt{\eps/4}
= \paren{\frac{1}{2\sqrt{r}} + \frac{3}{2}}\sqrt{\eps}, \notag \\
\sqrt{\eps_u} &\coloneqq \sqrt{\eps_1} + \sqrt{\eps_2}
= \sqrt{\eps} + \sqrt{\eps/4}
= \frac{3}{2}\sqrt{\eps}. \notag 
\end{align}
Then the first term of the summand in Eqn.~\eqref{eqn:errorbound-cov3} is at most 
\begin{align}
\prob{\frac{\card{\cY\cap\C^{n-1}(\vz,\sqrt{n\tau_\ell})}}{M} < \frac{1}{2}\gamma_\ell}
&\le \prob{\card{\cY\cap\C^{n-1}(\vz,\sqrt{n}(\sqrt{\rho} - \sqrt{\eps_\ell}))} < \frac{1}{2}\gamma_\ell M}. \notag 
\end{align}
Note that 
\begin{align}
\expt{\card{\cY\cap\C^{n-1}(\vz,\sqrt{n}(\sqrt{\rho} - \sqrt{\eps_\ell}))}} 
= \frac{\area\paren{\C^{n-1}(\sqrt{n}(\sqrt{\rho} - \sqrt{\eps_\ell}))}}{\area\paren{\cS^{n-1}(\sqrt{nr})}}M
&= \gamma_\ell M. \notag 
\end{align}
Hence by the Chernoff bound, 
the above probability is at most 
\begin{align}
\exp\paren{-\frac{(1/2)^2}{3}\gamma_\ell M}
&= e^{-e^{\Omega(n)}}. \notag 
\end{align}
By similar reasoning, the second term of the summand in Eqn.~\eqref{eqn:errorbound-cov3} is at most $ \exp\paren{-\frac{(1/2)^2}{3}\gamma_u M} = e^{-e^{\Omega(n)}} $. 
Since the concentration bounds are doubly exponentially small, the probability in Eqn.~\eqref{eqn:tobound-cov3} is still $ e^{-e^{\Omega(n)}} $ once we take union bounds over $ \vz\in\cZ $ and $ \tau\in\cI $ which are (singly) exponential in total. 

Finally, a union bound shows that with probability doubly exponentially close to $1$, $\cY$ of size $ M = (c/\sqrt{\eps})^{n-o(n)} $ satisfies properties~\ref{itm:cov-1} and~\ref{itm:cov-3} simultaneously. 
This completes the proof. 
\end{proof}

\subsection{Proof of Eqn.~\eqref{eqn:expt-toshow}}\label{sec:lbound_EW}
\begin{align}
\expt{W}=&\sum_{\cL\in\binom{\cM}{L}}\sum_{\vy\in\cY}\prob{\psi\paren{\cL}\subset\cB^n\paren{\vy,\sqrt{nN}}}\label{eqn:enc_notation}\\
=&\binom{M}{L}\card{\cY}\mu^L\label{eqn:app-mu-def}\\
\ge&\paren{M/L}^L\card{\cY}\mu^L,\notag
\end{align}
where in Eqn.~\eqref{eqn:enc_notation} we use the shorthand notation
\[\psi(\cL)\coloneqq\curbrkt{\psi(m)\colon m\in\cL},\] 
and in Eqn.~\eqref{eqn:app-mu-def}, $\mu$ is defined as follows,
\[\mu\coloneqq\frac{\area\paren{\C^{n-1}\paren{\sqrt{nN}}}}{\area\paren{\cS^{n-1}\paren{\sqrt{nP}}}}.\]

\subsection{Proof of Eqn.~\eqref{eqn:var-toshow}}\label{sec:ubound_varW}
For $\cL=\curbrkt{m_1,\cdots,m_L}$ and $\vy\in\cY$, define
\[\bI\paren{\vy,\cL}\coloneqq\indicator{\psi\paren{\cL}\subset\cB^n\paren{\vy,\sqrt{nN}}}.\]

Now the variance of $W$ can be bounded from above as follows,
\begin{align}
\var{W}=&\expt{W^2}-\expt{W}^2\notag\\
=&\sum_{\vy_1,\vy_2\in\cY}\sum_{\cL_1,\cL_2\in\binom{\cM}{L}}\expt{\bI\paren{\vy_1,\cL_1}\bI\paren{\vy_2,\cL_2}}-\expt{\bI\paren{\vy_1,\cL_1}}\expt{\bI\paren{\vy_2,\cL_2}}\notag\\
\le&\sum_{\substack{\cL_1,\cL_2\in\binom{\cM}{L} \\ \cL\cap\cL_2\ne\emptyset}}\sum_{\vy_1,\vy_2\in\cY}\expt{\bI\paren{\vy_1,\cL_1}\bI\paren{\vy_2,\cL_2}}\label{eqn:indep}\\
=&\card{\cY}^2\sum_{\ell=1}^L\sum_{\substack{\cL_1,\cL_2\in\binom{\cM}{L} \\ \card{\cL_1\cap\cL_2}=\ell}}\probover{\vbfy_1,\vbfy_2,\cC}{\cE(\cL_1,\cL_2,\vbfy_1,\vbfy_2)},\label{eqn:rewrite}
\end{align}
where
\begin{enumerate}
	\item Eqn.~\eqref{eqn:indep} follows since for disjoint $\cL_1$ and $\cL_2$, $\bI\paren{y_1,\cL_1}$ and $\bI\paren{y_2,\bL_2}$ are independent and hence the corresponding summand vanishes; we upper bound the variance by dropping the negative term;
	\item in Eqn.~\eqref{eqn:rewrite} the probability is taken over the code construction and the pair $\vbfy_1,\vbfy_2$ sampled independently and uniformly from $\cY$; the event $\cE(\cL_1,\cL_2)$ is defined as
	\[\cE(\cL_1,\cL_2)\coloneqq\curbrkt{\psi\paren{\cL_1}\subset\cB^n\paren{\vbfy_1,\sqrt{nN}},\;\psi\paren{\cL_2}\subset\cB^n\paren{\vbfy_2,\sqrt{nN}}}.\]
\end{enumerate}
It is easy to verify that for any $m\in\cL_1\cap\cL_2$, 
\begin{align}
\cE(\cL_1,\cL_2)\subset\cE_1(\cL_1,\cL_2)\cap\cE_2(\cL_1,\cL_2)\cap\cE_3(\cL_1,\cL_2), \notag
\end{align}
where
\begin{align}
\cE_1(\cL_1,\cL_2)\coloneqq&\curbrkt{\vbfy_1\in\cB^n\paren{\psi\paren{m},\sqrt{nN}},\;\vbfy_2\in\cB^n\paren{\psi\paren{m},\sqrt{nN}}} ,\notag\\
\cE_2(\cL_1,\cL_2)\coloneqq&\curbrkt{\forall m_1\in\cL_1\setminus\curbrkt{m},\;\psi\paren{m_1}\in\cB^n\paren{\vbfy_1,\sqrt{nN}}} ,\notag\\
\cE_3(\cL_1,\cL_2)\coloneqq&\curbrkt{\forall m_2\in\cL_2\setminus\cL_1,\;\psi\paren{m_2}\in\cB^n\paren{\vbfy_2,\sqrt{nN}}} .\notag
\end{align}
Note that conditioned on $\cE_1$, $\cE_2$ and $\cE_3$ are independent, and
\begin{align}
\prob{\cE_1}=&\paren{\frac{\card{\cY\cap{\C^{n-1}\paren{\sqrt{n\rho}}}}}{\card{\cY}}}^2\eqqcolon\nu^2,\notag\\
\prob{\cE_2\cap\cE_3|\cE_1}=&\prob{\cE_2|\cE_1}\prob{\cE_3|\cE_1}=\mu^{L-1}\mu^{L-\ell}=\mu^{2L-\ell-1},\notag
\end{align}
where
$\rho \coloneqq N\paren{P-N}/P$ as shown in Fig.~\ref{fig:geom1}.
\begin{figure}[htbp]
 	\centering
 	\includegraphics[width=0.5\textwidth]{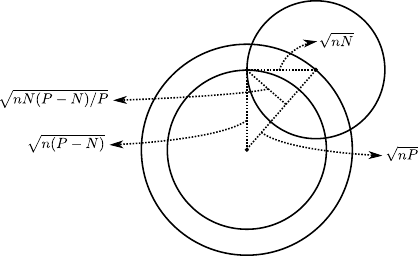}
 	\caption{If $ \psi(m) $ is a codeword on $ \cS^{n-1}(\sqrt{nP}) $, then a ball of radius $ \sqrt{nN} $ around $\psi(m)$ induces a cap of radius $ \sqrt{nN(P-N)/P} $ on the sphere $ \cS^{n-1}(\sqrt{n(P-N)}) $ on which $ \cY $ lives. }
 	\label{fig:geom1}
 \end{figure}

Let us upper bound $ \nu $ and $ \mu $. 
For $\mu$, we have
\begin{align}
\mu\le\frac{\area\paren{\cS^{n-1}\paren{\sqrt{nN}}}}{\area\paren{\cS^{n-1}\paren{\sqrt{nP}}}}= c_12^{-n\frac{1}{2}\log\frac{P}{N}},
\label{eqn:mu-bound}
\end{align}
where $c_1\coloneqq\sqrt{P/N}$.

For $\nu$, we choose $ \cY $ to be a $\sqrt{n\eps}$-covering of $ \cS^{n-1}(\sqrt{nP}) $ for some $ \eps>0 $ to be determined as given by Lemma~\ref{lem:covering-sphere} (with the choice $ r = P-N $). 
Then
\begin{align}
\nu &\le \frac{3}{2}\frac{\area\paren{\C^{n-1}(\sqrt{n}(\sqrt{\rho} + 3\sqrt{\eps}/2))}}{\area\paren{\cS^{n-1}(\sqrt{n(P-N)})}}  
\le \frac{3}{2}\frac{\area\paren{\cS^{n-1}(\sqrt{n}(\sqrt{\rho} + 3\sqrt{\eps}/2))}}{\area\paren{\cS^{n-1}(\sqrt{n(P-N)})}}
= c_2\paren{\frac{\sqrt{\rho} + 3\sqrt{\eps}/2}{\sqrt{P-N}}}^{n}, \label{eqn:nu-bound} 
\end{align}
where
$c_2 \coloneqq \frac{3\sqrt{P}}{2(\sqrt{\rho} + 3\sqrt{\eps}/2)}$.

Note that the number of pairs $\paren{\cL_1,\cL_2}$ with intersection size $\ell$ is
\[K_\ell\coloneqq\binom{M}{\ell}\binom{M-\ell}{L-\ell}\binom{M-L}{L-\ell}\le M^{2L-\ell}.\]
Hence overall we have
\begin{align}
\var{W}\le&\card{\cY}^2\sum_{\ell=1}^LK_\ell\nu^2\mu^{2L-\ell-1}\notag\\
\le&\card{\cY}^2\nu^2\mu^{-1}\sum_{\ell=1}^L\paren{M\mu}^{2L-\ell}\notag\\
\le&\card{\cY}^2\nu^2\mu^{-1} L\paren{M\mu}^{L}\label{eqn:take_l}\\
=& |\cY|^2 L\nu^2M^L\mu^{L-1}, \notag 
\end{align}
where Eqn.~\eqref{eqn:take_l} is obtained by noting that $M\mu\le2^{nR}c_12^{-n\frac{1}{2}\log\frac{P}{N}}= c_12^{-\delta n}$ and taking the dominating term corresponding to $\ell=L$.


\section{A primer on lattices and nested lattice codes}\label{sec:primer_lattices}

For a tutorial introduction to lattices and their applications, see the book by Zamir~\cite{zamir2014latticebook} or the notes by Barvinok~\cite{barvinok2013math669}.

If $ \vv_1,\ldots,\vv_\kappa $  are linearly independent vectors in $ \bR^n $, then the set of all integer linear combinations of $ \vv_1,\ldots,\vv_\kappa $ is called the lattice generated by the vectors $ \vv_1,\ldots,\vv_\kappa $, i.e.,
\[
\Lf \coloneqq \curbrkt{\sum_{i=1}^{\kappa}a_i\vv_i: a_i\in\bZ }.
\]
If  $ \bfG=[\vv_1\cdots \vv_\kappa] $, then we can write $ \Lf = \bfG\bZ^\kappa $. The matrix $ \bfG $ is called a generator matrix for $ \Lf $. The generator matrix of a lattice is not unique.
The integer $ \kappa $ is invariant for a lattice and is called the rank of $ \Lf $. In this  paper, we only consider lattices in $ \bR^n $ having rank $ n $.
It is obvious that $ \Lf $ is a discrete subgroup of $ \bR^n $ under vector addition. It is also a fact that every discrete subgroup of $ \bR^n $ is a lattice~\cite{barvinok2013math669}.

For any lattice $ \Lf $, it is natural to define the quantizer $ Q_{\Lf} $ which maps every point in $ \bR^n $ to the closest lattice point, i.e., for every $ \vx\in\bR^n $,
\begin{equation}
    \label{eqn:lattice_quant}
    Q_{\Lf}(\vx) \coloneqq \argmin{\vy\in\Lf}\Vert \vy-\vx \Vert,
\end{equation}
where we assume that ties (in computing the closest lattice point) are resolved according to some arbitrary but fixed rule. Associated with the quantizer is the quantization error
\[
[\vx]\bmod \Lf \coloneqq \vx - Q_{\Lf}(\vx).
\]

For every lattice $ \Lf $, we define the following parameters:
\begin{itemize}
	\item The set $$ \cP(\Lf)\coloneqq \{ \bfG\vx:\vx\in [0,1 )^n \}, $$
	where $ \bfG $ is a generator matrix of $ \Lf $, is called the fundamental parallelepiped of $ \Lf $.
	\item The fundamental Voronoi region $ \cV(\Lf) $ is the set of all points in $ \bR^n $ which are closest to the zero lattice point. In other words,
	\[
	\cV(\Lf) \coloneqq \{ \vx\in\bR^n: Q_{\Lf}(\vx) = 0 \}.
	\]
	Any set $ \cS\subset  \bR^n  $ such that the set of translates of $ \cS $ by lattice points, i.e., $ \{\cS+\vx :\vx\in\Lf \} $ form a partition of $ \bR^n $, is called a fundamental region of $ \Lf$. It is a fact that every fundamental region of $ \Lf $ has the same volume equal to $\det\Lf\coloneqq |\det(\bfG)| $, where $ \bfG $ is any generator matrix of $ \Lf $. The quantity $\det \Lf $ is called the determinant or covolume of $ \Lf $.
	\item The covering radius $ \rcov(\Lf) $ is the radius of the smallest closed ball in $ \bR^n $ which contains $ \cV(\Lf) $. It is also equal to the length of the largest vector within $ \cV(\Lf) $.
	\item The packing radius $ \rpack(\Lf) $ is the radius of the largest open ball which is contained within $ \cV(\Lf) $. Equivalently, it is half the minimum distance between two lattice points.
	\item The effective radius $~\reff(\Lf) $ is equal to the radius of a ball having volume equal to $ \vol(\cV(\Lf)) $.
\end{itemize}
Clearly, we have $ \rpack(\Lf)\leq~\reff(\Lf)\leq \rcov(\Lf) $.

In the context of power-constrained communication over Gaussian channels, a nested lattice code is typically the set of all lattice points within a convex compact subset of $ \bR^n $, i.e., $ \cC = \Lf\cap\cB $ for some set $ \cB\subset \bR^n $. Usually $ \cB $ is taken to be $ \cB(0,\sqrt{nP}) $ or $ \cV(\Lc) $ for some lattice $ \Lc $ constructed so as to satisfy the power constraint.

If $ \Lc,\Lf $ are two lattices in $ \bR^n $ with the property that $ \Lc\subsetneq\Lf $, then $ \Lc $ is said to be nested within (or, a sublattice of) $ \Lf $. A nested lattice code with a fine lattice $ \Lf $ and coarse lattice $ \Lc\subsetneq \Lf $ is the lattice code $ \Lf\cap\cV(\Lc) $. 

Lattices have been extensively used for problems of packing, covering and communication over Gaussian channels. For many problems of interest, we want to construct high-dimensional lattices $ \Lf $ such that $ \rpack(\Lf)/\reff(\Lf) $ is as large as possible, and $ \rcov(\Lf)/\reff(\Lf) $ is as small as possible. A class of lattices that has these properties is the class of Construction-A lattices, which we describe next.

Let $ q $ be a prime number, and $ \cC_{\mathrm{lin}} $ be an $ (n,\kappa) $ linear code over $ \bFq $. The Construction-A lattice obtained from $ \cC_{\mathrm{lin}} $ is defined to be
\[
\Lf(\cC_{\mathrm{lin}}) \coloneqq \{ \vv\in\bZ^n: [\vv]\bmod (q\bZ^n)\in\Phi(\cC) \},
\] 
where $ \Phi $ denotes the natural embedding of $ \bFq^n $ in $ \bR^n $.
An equivalent definition is that $ \Lf(\cC_{\mathrm{lin}}) = \Phi(\cC_{\mathrm{lin}})+q\bZ^n $.
We make use of the following result to choose our coarse lattices:
\begin{theorem}[\cite{erez2005lattices}]
	For every $ \delta>0 $,
	there exist sequences of prime numbers $ q_n $ and positive integers $ \kappa_n $ such that if $ \cC_{\mathrm{lin}} $ is a randomly chosen linear code\footnote{The $ (n,\kappa_n) $ random code is obtained by choosing an $ n\times \kappa_n $ generator matrix uniformly at random over $ \bFq $.} over $ \bF_{q_n} $, then
	\[
	\Pr\left[ \frac{\rpack(\Lf(\cC_{\mathrm{lin}}))}{\reff(\Lf(\cC_{\mathrm{lin}}))}<\frac{1}{2}-\delta \text{ or } \frac{\rcov(\Lf(\cC_{\mathrm{lin}}))}{\reff(\Lf(\cC_{\mathrm{lin}}))}>1+\delta\right] = o(1).
	\] 
\end{theorem}

\section{Other goodness properties of regular infinite constellations}\label{sec:ic_goodness}

In Sec.~\ref{sec:regular_ic}, for technical reasons\footnote{Specifically, in the proof, we need to take a union bound over centers in a Voronoi region. }, we require the lattice $ \Lc $ to be simultaneously good for packing and covering. 
However, to prove other properties in this section, we can use a simpler construction of ICs which we describe below.


Let $ \alpha>0 $. We allow $ \alpha $ to be a function of $ n $. Define $ \cA\coloneqq [-\alpha/2,\alpha/2)^n $.
We will study infinite constellations of the form $ \cC = \cC'+\alpha\bZ^n $ for finite sets $ \cC'\subset \cA $.
We assume that each point in $ \cC' $ is independent and uniformly distributed in $\cA$. 
In other words, $ \cC $ is obtained by tiling a finite subset of random points from within a cube. See Fig.~\ref{fig:ic} for a pictorial illustration of the construction of such an IC ensemble. Since the IC is a tiling, it suffices to study finite sets of points in the space $ \bR^n\bmod \cA $.~\footnote{Since $ \alpha \bZ^n $ is a lattice, we slightly abuse notation and define $ [\cdot]\bmod\cA \coloneqq [\cdot]\bmod\alpha\bZ^n $.}

Note that if $ \cC' $ forms a group with respect to addition modulo $ \cA $, then the resulting IC is a lattice. Construction-A lattices are essentially obtained by taking $ \cC' $ as an embedding of a linear code over a finite field into $ \cA $.

The density, NLD, effective volume and effective radius of an $(\alpha,M)$ IC are given by
\begin{align}
\Delta(\cC) =& M/\alpha^n, \quad
R(\cC) = \frac{1}{n}\log\frac{M}{\alpha^n}, \quad
V(\cC) = \frac{n}{\log\nicefrac{M}{\alpha^n}}, \quad
\reff(\cC) = \paren{\frac{\alpha^n}{V_nM}}^{1/n}. \label{eqn:param_ic} 
\end{align}
respectively.

For any set $ \cD\subset\bR^n $, define $ \cD^*\coloneqq \cD\bmod\cA  $.

\begin{figure}
	\begin{center}
		\includegraphics[width=9cm]{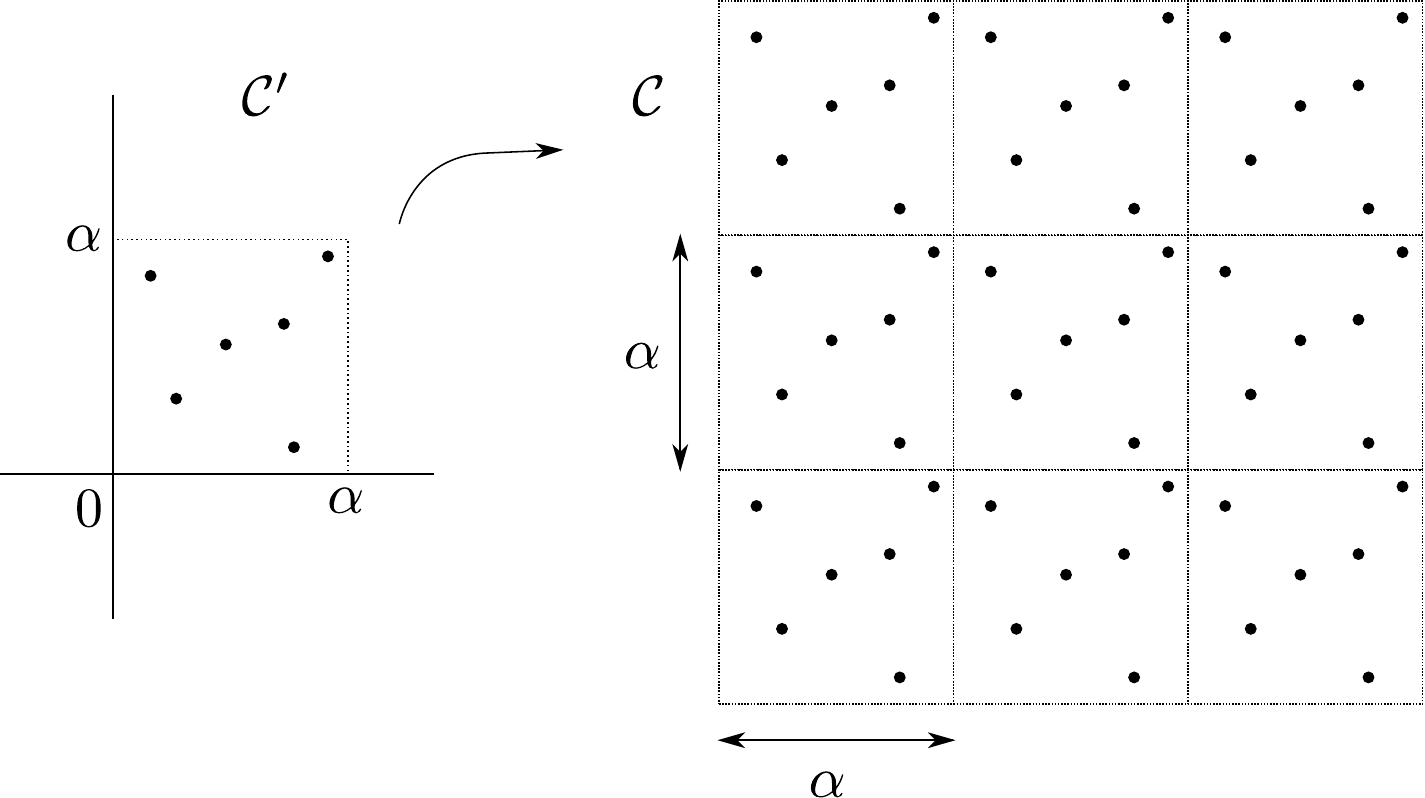}
		\caption{Illustration of the class of infinite constellations studied in Appendix~\ref{sec:ic_goodness}. }
		\label{fig:ic}
	\end{center}
\end{figure}

\subsection{Packing goodness}\label{sec:ic_packing}

The packing radius of an IC $ \rpack(\cC) $ is defined to be half the minimum distance between two points.
We say that an infinite constellation is good for packing if $ \rpack(\cC)/\reff(\cC)\geq 0.5-o(1) $.

As a warm-up, we give a greedy  construction which is good for packing.

Choose $ \alpha $ to be some constant larger than $ 4 $. We will construct an infinite constellation with packing radius at least $ 1 $. 

The IC is constructed iteratively as follows: 
Start with an arbitrary point $ \vx_1 $.  At the $ i $-th step, choose $ \vx_i $ to be an arbitrary point from $ \cA\backslash \cup_{j=1}^{i-1}\cB^*(\vx_j,2) $. We repeat this till the $ \cB^*(\vx_j,2) $'s cover $ \cA $.  Suppose that the algorithm terminates at the $ M $-th step.

The construction ensures that the packing radius  is at least $ 1 $. Moreover,
\[ M\geq \frac{\alpha^n}{\vol(\cB(0,2))}. \]
However, $ \alpha^n/M=V_n\reff^n $. Using this in the above gives $ \rpack/\reff\geq 0.5 $.

\subsection{AWGN goodness}\label{sec:ic_awgn}

We say that an $(\alpha, M)$ infinite constellation $ \cC $ is good for AWGN (Additive White Gaussian Noise)~\cite{erez2005lattices} if for $ \vbfz\sim \cN(0,\sigma^2\bfI) $ and $ \vbfx\sim\cU(\cC\cap \cA) $, we have
\[
\Pr[\Vert \vbfz \Vert>\Vert \vbfx +\vbfz-\vx_j \Vert \text{ for some }\vx_j\in \cC] =2^{-\Theta(n)}
\]
where the probability is over the random choice of the codeword $ \vbfx $ and the noise $ \vbfz $. This is equal to the probability that a codeword different from the transmitted one is closer to the received vector when a random codeword is transmitted through an AWGN channel.

The following proposition recovers the achievability part of Poltyrev's~\cite{poltyrev1994coding} result:
\begin{proposition}
	Fix $ \delta>0 $ and $ N>0 $. A random $ (4\sqrt{n\sigma^2}, M) $ constellation with $ M $ chosen so as to satisfy $~\reff/\sqrt{n\sigma^2}>2^{\delta} $ is good for AWGN with probability  $ 1-2^{-\Theta(n)} $.
\end{proposition}

\begin{proof}

	Since codewords are uniformly chose, it suffices to assume that the first codeword is transmitted. 
	\begin{align*}
	\bE_{\cC}\Pr[\Vert \vbfz \Vert>\Vert \vbfx +\vbfz-\vx_j \Vert \text{ for some }\vx_j\in \cC] &\leq \Pr[\Vert\vbfz\Vert^2>2^{\delta}n\sigma^2]+\Pr[\vx_j\in\cB^*(\vx_1+\vbfz,\sqrt{n\sigma^2}2^{\delta/2})\text{ for some }j\neq 1]\\
	&\leq \Pr[\Vert\vbfz\Vert^2>2^{\delta}n\sigma^2]+\sum_{j=2}^M\Pr[\vx_j\in\cB^*(\vx_1+\vbfz,\sqrt{n\sigma^2}2^{\delta/2})]\\
	&\leq 2^{-\Theta(n)}+ M\frac{\vol(\cB(0,\sqrt{n\sigma^2}2^{n\delta/2}))}{\alpha^n}\\
	&= 2^{-\Theta(n)}+ \frac{\vol(\cB(0,\sqrt{n\sigma^2}2^{n\delta/2}))}{\vol(\cB(0,\reff))}\\
	&=2^{-\Theta(n)}. \qedhere
	\end{align*}
	
\end{proof}

\subsection{Covering goodness}\label{sec:ic_covering}
We say that an infinite constellation $\cC$ is a $\beta$-covering if for every $\vy\in\bR^n,$
\[
\min_{\vx\in\cC}\Vert \vx-\vy \Vert \leq \beta.
\]
The \emph{covering radius} of an infinite constellation $\cC$, denoted $\rcov(\cC)$, is the largest $\beta>0$ such that $\cC$ is a $\beta$ covering.

We say that a sequence of $(\alpha,M)$ ICs $\cC$ is $(1+\delta)$-good for covering if 
\[
\lim\sup_{n\to\infty}\frac{\rcov(\cC)}{\reff(\cC)} \leq 1+\delta. 
\]

\begin{proposition}
	Fix any $N>0$, and define $C\coloneq \frac{1}{2}\log_2\frac{1}{2\pi e N}$. Also choose $\alpha = 2\sqrt{nN}$, $\epsilon_n = \frac{1}{\log n}$, and $M=2^{nC(1+\epsilon_n)}\alpha^n$.
	
	Then, a random $(\alpha,M)$ IC is $(1+\epsilon_n)2^{\epsilon_n C}$-good for covering\footnote{Essentially, $(1+o(1))$-good for covering.} with probability $1-2^{-2^{\Theta(n/\log n)}}$.
\end{proposition}

\begin{proof}
Define $\cQ = \frac{\epsilon_n\sqrt{N}}{4}\bZ^n$. We can conclude that $\cC$ is a $\sqrt{nN}(1+\epsilon_n/4)$ covering if for all $\vy\in\cQ$, we have
\[
\min_{\vx\in\cC}\Vert \vx-\vy \Vert\leq \sqrt{nN}.
\]
Since 
\[
\reff(\cC) =\left(\frac{\Gamma(n/2+1)}{\pi^{n/2}}\times\frac{\alpha^n}{M}\right)^{1/n} = \left(\frac{\sqrt{n}2^{-R}}{\sqrt{2\pi e}}\right)(1+o(1)) = \sqrt{nN}\times 2^{-\epsilon_n C}, 
\]
this ensures that 
\[
\frac{\rcov(\cC)}{\reff(\cC)} \leq (1+\epsilon_n)2^{\epsilon_n C}. 
\]
It is therefore sufficient to show that 
\[
\Pr\left[\max_{\vy\in\cQ}\min_{\vx\in\cC}\Vert\vx-\vy\Vert>\sqrt{nN}\right] = 2^{-2^{\Theta(n)}}.
\]

To prove the rest, we simply find the probability that there is no point that is $\sqrt{nN}$-close to $\vy\in\cQ$, and then take a union bound over $\vy$. To compute the aforementioned probability, we use a Chernoff bound.

Fix any $\vy\in\cQ$. Suppose that $\cC=\{\vbfx_1,\vbfx_2,\ldots,\vbfx_M\}$, where each $\vbfx_i$ is is independent of the others and uniformly distributed over $ [0,\alpha]^n$. Define $\xi_i(\vy)$ to be the indicator random variable which is $1$ if $\Vert\vbfx_i-\vy\Vert\leq \sqrt{nN}$, and zero otherwise.

We then have
\begin{align*}
\mu &\coloneq \bE\left[\sum_{i=1}^M\xi_i(\vy)\right] \\
&= M\times \frac{\vol(\cB(\vy,\sqrt{nN}))}{\alpha^n} \\
                            &=2^{nR}\vol(\cB(\vy,\sqrt{nN}))\\
                            &=2^{nC\epsilon_n(1+o(1))}. 
\end{align*}
This implies, by the Chernoff bound, that
\begin{align*}
\Pr\left[\sum_{i=1}^M\xi_i(\vy)=0\right] &\leq \Pr\left[\sum_{i=1}^M\xi_i(\vy)\leq \mu/2\right] 
                              \leq e^{-\mu/12} = 2^{-2^{\Theta(n/\log n)}}. 
\end{align*}

Therefore, the probability that random IC is not a $\sqrt{nN}(1+\epsilon/4)$-covering is upper bounded by
\begin{align*}
\Pr\left[\max_{\vy\in\cQ}\min_{\vx\in\cC}\Vert\vx-\vy\Vert>\sqrt{nN}\right] &\leq \left(\frac{4\alpha}{\epsilon_n}\right)^n 2^{-2^{\Theta(n/\log n)}} 
                      =2^{O(n\log n)}2^{-2^{\Theta(n/\log n)}} = 2^{-2^{\Theta(n/\log n)}}. 
\end{align*}
This completes the proof.
\end{proof}

\section{Converse of list decoding capacity theorem for infinite constellations}
\label{app:converse-ic}

\begin{proposition}
\label{prop:converse-ic}
For any $N>0$ and $ \delta>0 $, let $ \cC $ be an arbitrary $ (N,L) $-list decodable IC of NLD $ \frac{1}{2}\log\frac{1}{2\pi eN} + \delta $. 
Then $ L\ge2^{\Omega(\delta n)} $. 
\end{proposition}

\begin{proof}
Let $ \cC\subset\bR^n $ be an arbitrary infinite constellation of NLD $ \frac{1}{2}\log\frac{1}{2\pi eN} + \delta $ for some constant $ \delta>0 $. 
Then there must exist a sufficiently large $P$ such that 
\begin{align}
\frac{1}{n}\log\frac{\card{\cC\cap\cB^n\paren{0,\sqrt{nP}}}}{\vol\paren{\cB^n\paren{0,\sqrt{nP}}}} &\ge \frac{1}{2}\log\frac{1}{2\pi eN} + \frac{\delta}{2}. \notag 
\end{align}
Therefore,
\begin{align}
\frac{1}{n}\log\card{\cC\cap\cB^n\paren{0,\sqrt{nP}}} &\ge \frac{1}{2}\log\frac{1}{2\pi eN} + \frac{1}{n}\log\vol\paren{\cB^n\paren{0,\sqrt{nP}}} + \frac{\delta}{2} \notag \\
&= \frac{1}{2}\log\frac{1}{2\pi eN} + \frac{1}{n}\log\paren{V_n\sqrt{nP}^n} + \frac{\delta}{2} \notag \\
&\asymp \frac{1}{2}\log\frac{1}{2\pi eN} + \frac{1}{2}\log(2\pi eP) + \frac{\delta}{2} \notag \\
&= \frac{1}{2}\log\frac{P}{N} + \frac{\delta}{2}. \notag 
\end{align}
By the list decoding converse for codes with power constraints \cite[Lemma 33]{zhang-quadratic-arxiv}, since the code $ \cC\cap\cB^n\paren{0,\sqrt{nP}} $ has rate larger than the list decoding capacity $ \frac{1}{2}\log\frac{P}{N} $, it must have exponential list sizes. 
That is, there must exist $ \vy\in\bR^n $ such that $ \card{\cC\cap\cB^n\paren{0,\sqrt{nN}}\cap\cB^n\paren{\vy,\sqrt{nN}}}\ge2^{\Omega(\delta n)} $. 
This implies the existence of $ \vy\in\bR^n $ such that $ \card{\cC\cap\cB^n\paren{\vy,\sqrt{nN}}}\ge2^{\Omega(\delta n)} $, which completes the proof.  
\end{proof}

\bibliographystyle{alpha}
\bibliography{ref} 

\end{document}